\documentclass[a4paper, UKenglish,cleveref,thm-restate,numberwithinsect]{lipics-v2021}
	
\pdfoutput=1
\hideLIPIcs%
\nolinenumbers%

\usepackage[export]{adjustbox}
\usepackage[
	disable
]{todonotes}
\usepackage{csquotes}
\usepackage{tikz}
\usepackage{mathtools}
\usepackage{upgreek}
\usepackage{thmtools}

\let\endhere\qedhere

\crefname{condition}{condition}{conditions}
\Crefname{condition}{Condition}{Conditions}

\crefname{statement}{statement}{statements}
\Crefname{statement}{Statement}{Statements}

\crefname{observation}{Observation}{Observations}

\newlength{\myA}%
\newlength{\myB}%

\newcolumntype{C}{>{\centering\arraybackslash}X}

\makeatletter
\def\cleartheorem#1{%
    \expandafter\let\csname#1\endcsname\relax
    \expandafter\let\csname c@#1\endcsname\relax
}
\makeatother

\cleartheorem{definition}
\declaretheorem[
	name=Definition, 
	sibling=theorem,
	style=definition,
	qed={\lipicsEnd}
]{definition}

\usetikzlibrary{arrows, arrows.meta}
\usetikzlibrary{calc}
\usetikzlibrary{decorations.pathmorphing, decorations.pathreplacing, decorations.shapes}
\usetikzlibrary{graphs, graphs.standard}
\usetikzlibrary{math}
\usetikzlibrary{positioning}

\newcommand{\convexpath}[2]{
[
    create hullnodes/.code={
        \global\edef\namelist{#1}
        \foreach [count=\counter] \nodename in \namelist {
            \global\edef\numberofnodes{\counter}
            \node at (\nodename) [draw=none,name=hullnode\counter] {};
        }
        \node at (hullnode\numberofnodes) [name=hullnode0,draw=none] {};
        \pgfmathtruncatemacro\lastnumber{\numberofnodes+1}
        \node at (hullnode1) [name=hullnode\lastnumber,draw=none] {};
    },
    create hullnodes
]
($(hullnode1)!#2!-90:(hullnode0)$)
\foreach [
    evaluate=\currentnode as \previousnode using \currentnode-1,
    evaluate=\currentnode as \nextnode using \currentnode+1,
    ] \currentnode in {1,...,\numberofnodes} {
-- ($(hullnode\currentnode)!#2!-90:(hullnode\previousnode)$)
  let \p1 = ($(hullnode\currentnode)!#2!-90:(hullnode\previousnode) - (hullnode\currentnode)$),
    \n1 = {atan2(\y1,\x1)},
    \p2 = ($(hullnode\currentnode)!#2!90:(hullnode\nextnode) - (hullnode\currentnode)$),
    \n2 = {atan2(\y2,\x2)},
    \n{delta} = {-Mod(\n1-\n2,360)}
  in
    {arc [start angle=\n1, delta angle=\n{delta}, radius=#2]}
}
-- cycle
}

\definecolor{cblue}{HTML}{007582}
\definecolor{cred}{HTML}{FF0000}

\tikzstyle{regular}=[anchor=center, circle, inner sep = .5pt]
\tikzstyle{edge}=[dotted]
\tikzstyle{guard}=[edge, very thick]
\tikzstyle{blue}=[regular, text=cblue]
\tikzstyle{red}=[regular, text=cred]

\tikzmath{
    \exampled=2.4;
    \exampleh=sqrt(3/4*(\exampled*\exampled));
}
\tikzmath{
    \cstep=1;
    \cebase=\cstep/2;
    \cvbase=0;
    \cheight=sqrt(3/4*(\cstep*\cstep));
} %
\newcommand*{\nat}{\mathbb{N}}
\newcommand*{\natpos}{\nat_{\geq{} 1}}
\newcommand*{\reell}{\mathbb{R}}

\newcommand*{\isdef}{\coloneqq}

\newcommand*{\isomorphic}{\cong}

\newcommand{\dom}{\operatorname{dom}}
\newcommand{\img}{\operatorname{img}}
\newcommand*{\pto}{\rightharpoonup}

\newcommand*{\pot}{\mathcal{P}}
\newcommand*{\kpot}{\mathcal{P}_{k}}

\newcommand*{\all}{\text{for all}}
\newcommand*{\hi}{{\hat{\imath}}}
\newcommand*{\hj}{{\hat{\jmath}}}

\newcommand*{\tupel}[1]{\overline{#1}}
\newcommand*{\class}[1]{\mathsf{#1}}
\newcommand*{\classC}{\mathfrak{C}}
\newcommand*{\classIC}{\mathfrak{C}_{\mathsf{I}}}

\newcommand*{\Hom}{\operatorname{Hom}}
\newcommand*{\hV}{\ensuremath{h_V}}
\newcommand*{\hE}{\ensuremath{h_E}}

\let\oldphi\phi%
\let\phi\varphi%
\let\varphi\oldphi%

\let\emptyset\varnothing%

\newcommand*{\set}[1]{\ensuremath{\{#1\}}}

\newcommand*{\enum}[1]{\ensuremath{\langle #1 \rangle}}

\renewcommand*{\mid}{\,:\,}

\newcommand*{\union}{\cup}
\newcommand*{\intersect}{\cap}
\newcommand*{\bigunion}{\bigcup}

\newcommand*{\disunion}{\mathrel{\dot{\cup}}}

\newcommand*{\T}{T}
\newcommand*{\F}{F}
\newcommand*{\Ts}{\dot{T}}
\newcommand*{\sF}{\F_{\text{s}}}
\newcommand*{\induced}[2]{{#1}[{#2}]}

\newcommand*{\troot}{\ensuremath{\omega}}
\newcommand*{\troots}{\operatorname{\Omega}}
\newcommand*{\height}{\operatorname{height}}
\newcommand*{\level}{\operatorname{level}}

\newcommand*{\p}{\operatorname*{P}}
\newcommand*{\lcv}{\operatorname*{\land}} %

\newcommand*{\TW}{\class{TW}}
\newcommand*{\TWk}{\TW_k}
\newcommand*{\Pn}{\mathcal{P}}

\newcommand*{\HG}{\mathcal{H}}
\newcommand*{\GG}{\mathcal{G}}

\newcommand*{\V}{\ensuremath{V}}
\newcommand*{\E}{\ensuremath{E}}
\newcommand*{\f}{\beta}

\newcommand*{\hd}{\operatorname{hd}}
\newcommand*{\shd}{\operatorname{shd}}

\newcommand*{\hGamma}{\widehat{\Gamma}}
\newcommand*{\sGamma}{\Gamma_{\text{s}}}
\newcommand*{\hsGamma}{\widehat{\Gamma}_{\text{s}}}
\newcommand*{\ulGamma}{\widetilde{\Gamma}}

\newcommand*{\HD}{\class{HD}}
\newcommand*{\HDk}{\HD_k}
\newcommand*{\IHD}{\class{IHD}}
\newcommand*{\IHDk}{\IHD_k}

\newcommand*{\SHD}{\class{SHD}}
\newcommand*{\SHDk}{\SHD_k}
\newcommand*{\ISHD}{\class{ISHD}}
\newcommand*{\ISHDk}{\ISHD_k}

\newcommand*{\GHW}{\class{GHW}}
\newcommand*{\GHWk}{\GHW_k}
\newcommand*{\IGHW}{\class{IGHW}}
\newcommand*{\IGHWk}{\IGHW_k}

\newcommand*{\IG}[1]{\ensuremath{#1}}
\newcommand*{\I}{\IG{I}}
\newcommand*{\J}{\IG{J}}

\newcommand*{\kLI}{L}
\newcommand*{\Mf}{M_f}

\newcommand*{\blueN}{\operatorname{\upbeta}}
\newcommand*{\R}{\mathcal{R}}
\newcommand*{\B}{\mathcal{B}}

\newcommand*{\labels}{\operatorname{labels}}
\newcommand*{\XB}{X_b}
\newcommand*{\XR}{X_r}

\newcommand*{\changeB}[3]{#1\langle#2{\to}#3\rangle}
\newcommand*{\removeB}[2]{\changeB{#1}{#2}{\bullet}}
\newcommand*{\changeR}[3]{#1[#2 {\to} #3]}
\newcommand*{\removeR}[2]{\changeR{#1}{#2}{\bullet}}
\newcommand*{\transition}[2]{#1[{\leadsto} #2]}

\newcommand*{\domb}[1]{\ensuremath{\textup{db}_{#1}}}
\newcommand*{\dombQ}{\domb{Q}}
\newcommand*{\domr}[1]{\ensuremath{\textup{dr}_{#1}}}
\newcommand*{\domrQ}{\domr{Q}}

\newcommand*{\RvertexMap}[1]{\operatorname{succ}^{R}_{#1}}
\newcommand*{\BvertexMap}[1]{\operatorname{succ}^{B}_{#1}}

\newcommand*{\isoR}{\pi_{R}}
\newcommand*{\isoB}{\pi_{B}}

\newcommand*{\GLI}{\class{GLI}}
\newcommand*{\GLIDk}{\class{GLI}^k}

\newcommand*{\Logic}[1]{\textsf{#1}}

\newcommand*{\CL}{\Logic{C}}

\newcommand*{\GC}{\Logic{GC}}
\newcommand*{\GCk}{\GC^{k}}
\newcommand*{\GCD}[1]{\GC_{#1}}
\newcommand*{\GCDk}{\GCD{k}}

\newcommand*{\RGCk}{\Logic{RGC}^k}

\newcommand*{\varx}{\mathtt{x}}
\newcommand*{\vare}{\mathtt{e}}
\newcommand*{\varv}{\mathtt{v}}

\newcommand*{\VARE}{\textsf{VAR}_\vare}
\newcommand*{\VARV}{\textsf{VAR}_\varv}

\newcommand*{\free}{\textstyle\operatorname*{free}}
\newcommand*{\freeE}{\free_\vare}
\newcommand*{\freeV}{\free_\varv}
\newcommand*{\gdepth}{\operatorname*{gd}}
\newcommand*{\vset}{\operatorname{vars}}

\newcommand*{\IInt}{\mathcal{I}}
\newcommand*{\assignmentE}{\nu_{\vare}}
\newcommand*{\assignmentV}{\nu_{\varv}}

\newcommand*{\bigland}{\bigwedge}

\newcommand*{\LogGuard}[1]{\begingroup\def\x{#1}\ifx\x\@empty\Delta\else\Delta_{\x}\fi\endgroup}
\newcommand*{\existsgeq}[1]{\exists^{\scalebox{0.6}{$\geq$}#1}}
\newcommand*{\existseq}[1]{\exists^{\scalebox{0.6}{$=$}#1}}
\newcommand*{\logeq}{\kern1pt{=}\kern1pt}
\newcommand*{\qsep}{\mathbin{.}}

\title{On Homomorphism Indistinguishability and Hypertree Depth}

\author{Benjamin Scheidt}{Humboldt-Universität zu Berlin, Germany}{benjamin.scheidt@hu-berlin.de}{https://orcid.org/0000-0003-2379-3675}{}%

\authorrunning{B.\ Scheidt}

\Copyright{Benjamin Scheidt}



\relatedversion{}

\acknowledgements{We thank Nicole Schweikardt for helpful discussions.}
 
\begin{document}

\maketitle

\begin{abstract}
$\GCk$ is a logic introduced by Scheidt and Schweikardt (2023) to express properties of hypergraphs
It is similar to first-order logic with counting quantifiers ($\CL$) adapted to the hypergraph setting.
It has distinct sets of variables for vertices and for hyperedges and requires vertex variables to be guarded by hyperedge variables on every quantification.

We prove that two hypergraphs $G$, $H$ satisfy the same sentences in the logic $\GCk$ with \emph{guard depth} at most $k$ if, and only if, they are homomorphism indistinguishable over the class of hypergraphs of strict hypertree depth at most $k$.
This lifts the analogous result for tree depth $\leq k$ and sentences of first-order logic with counting quantifiers of quantifier rank at most $k$ due to Grohe (2020) from graphs to hypergraphs.
The guard depth of a formula is the quantifier rank with respect to hyperedge variables, and strict hypertree depth is a restriction of hypertree depth as defined by Adler, Gaven\v{c}iak and Klimo\v{s}ová (2012).
To justify this restriction, we show that for every $H$, the strict hypertree depth of $H$ is at most 1 larger than its hypertree depth, and we give additional evidence that strict hypertree depth can be viewed as a reasonable generalisation of tree depth for hypergraphs. \end{abstract}

\section{Introduction}%
\label{sec:introduction}

The ($k$-dimensional) Weisfeiler-Leman algorithm describes a technique to classify the vertices (or $k$-tuples) of a graph, by iteratively computing a colouring (i.e., a classification) of the vertices (or $k$-tuples), which gets refined each iteration until it stabilises.
While it can be used as a way to (imperfectly) test graphs for isomorphism, it has found many other -- seemingly very different -- uses, e.g.\ reducing the cost of solving linear programs~\cite{Grohe2014}, as graph kernels~\cite{Shervashidze} or even as an architecture for graph neural networks~\cite{Xu2018,Morris2019,Grohe2021a,Grohe2020a}.
For a more in-depth overview on the expressive power of the Weisfeiler-Leman algorithm itself, consult~\cite{Kiefer2020} as a starting point.
The success of the Weisfeiler-Leman algorithm can in part be explained by its simplicity, but also by the fact that it appears to capture the structure of a graph really well.
This can be explained by its connection to first-order logic with counting quantifiers and to homomorphism counts over graphs of bounded tree width.
A classical result due to Cai, Fürer and Immerman~\cite{Cai1992} and Immerman and Lander~\cite{Immerman1990} says that two graphs are indistinguishable by the $k$-dimensional Weisfeiler-Leman algorithm if, and only if, they satisfy the same sentences of first-order logic with counting quantifiers ($\CL$) and $k{+}1$ variables ($\CL^{k{+}1}$).

Dvo\v{r}ák~\cite{Dvorak2010} and Dell, Grohe, Rattan~\cite{Dell2018} related the Weisfeiler-Leman algorithm to homomorphism counts over graphs of bounded tree width (this was subsequently generalised to relational structures of bounded tree width by Butti and Dalmau~\cite{Butti2021}).
They showed that two graphs are homomorphism indistinguishable over the class $\TWk$ of graphs of tree width at most $k$ if, and only if, they are indistinguishable by the $k$-dimensional Weisfeiler-Leman algorithm.
Here, two graphs $G$ and $H$ are \emph{homomorphism indistinguishable} over a class $\classC$ of graphs, if the number of homomorphisms from $F$ into $G$ equals the number of homomorphisms from $F$ into $H$ for all $F \in \classC$.
Dvo\v{r}ák used the previously mentioned connection to $\CL^{k{+}1}$ and an inductive characterisation of the graphs of bounded tree width to prove this result, while Dell et al.\ used elaborate algebraic techniques on vectors containing homomorphism counts.

In recent years, a whole theory has emerged around homomorphism indistinguishability.
There are characterisations of homomorphism indistinguishability for classes of graphs other than $\TWk$ (cf.~\cite{Dawar2021,Fluck2024,Montacute2022,Roberson2022}), among which we would like to emphasise the following: A classical result by Lovász~\cite{Lovasz1967}, stating that two graphs are isomorphic if, and only if, they are homomorphism indistinguishable over all graphs; a well-received result by Man\v{c}inska and Roberson~\cite{Mancinska2020}, stating that two graphs are homomorphism indistinguishable over the class of planar graphs if, and only if, they are quantum isomorphic; and, of importance for this paper, Grohe~\cite{Grohe2020} showed that two graphs are homomorphism indistinguishable over the graphs of tree depth at most $m$ if, and only if, they satisfy the same sentences of $\CL$ with quantifier rank at most $m$ ($\CL_m$).
There is also work concerned with a more fundamental analysis of homomorphism counting from restricted classes (cf.~\cite{Boeker2019a,Grohe2022,Neuen2023,Rattan2023,Seppelt2023}).

Some real-world problems can be represented by hypergraphs in a much more natural way than by graphs.
The great track record of the Weisfeiler-Leman method poses the question, whether a similar algorithm exists that works on \emph{hypergraphs}.
A direct application of the Weisfeiler-Leman algorithm on the incidence structure of a hypergraph is sometimes used.
But Böker noted in~\cite{Boeker2019}, that this approach does not capture the hypergraph structure well, since the algorithm will mix up hyperedges and vertices.
Thus, a proper variant of the Weisfeiler-Leman algorithm that works on hypergraphs is, to the best of our knowledge, still missing.
We believe that establishing results analogous to the ones mentioned so far can give valuable insight on how the algorithm should operate on hypergraphs.
A first step from this angle is a result by Scheidt and Schweikardt~\cite{Scheidt2023}, who lift Dvo\v{r}ák's result to hypergraphs by proving the following: two hypergraphs $G$, $H$ are homomorphism indistinguishable over the class $\GHWk$ of hypergraphs of generalised hypertree width at most $k$ if, and only if, they satisfy the same sentences of the logic $\GCk$.
$\GCk$ is a novel logic introduced in~\cite{Scheidt2023}.
It has distinct variables for vertices and for hyperedges and counting quantifiers for both variable types.
The main feature of $\GCk$ is that it bounds the number of variables for hyperedges by $k$, and it requires that vertex variables are \enquote{guarded by} (i.e., contained in) hyperedge variables on every quantification.
\subparagraph*{Contributions} 
As the main contribution of this work, we show that two hypergraphs satisfy the same sentences of the logic $\GCk$ with guard depth at most $k$ if, and only if, they are homomorphism indistinguishable over the class of hypergraphs of strict hypertree depth at most $k$ (\cref{thm:main}).
The guard depth is the quantifier depth of the hyperedge variables.
This theorem follows from an inductive characterisation of the class of hypergraphs of strict hypertree depth $\leq k$, combined with the main technical lemmas of Scheidt and Schweikardt~\cite{Scheidt2023}.
We believe that this inductive characterisation is interesting on its own, since the same technique combined with the core lemmata in Dvo\v{r}ák's work~\cite{Dvorak2010} can be used to give a concise proof of the analogous result on graphs due to Grohe~\cite{Grohe2020}, which was independently recognised and shown by Fluck et al.~\cite{Fluck2024} recently.
Strict hypertree depth is a mild restriction of hypertree depth as defined by Adler, Gaven\v{c}iak and Klimo\v{s}ová~\cite{Adler2012}.
This (as it turns out only slight) deviation from hypertree depth is surprising at first.
Because of the properties and relations between strict hypertree depth and hypertree depth we acquire in this paper, we claim that strict hypertree depth can be viewed as a reasonable generalisation of tree depth for hypergraphs too.
In particular, we show that the strict hypertree depth of a hypergraph is at most 1 larger than its hypertree depth (\cref{thm:shd-is-almost-hd}).
Moreover, we show that the distinguishing power of homomorphism counts from hypergraphs of hypertree depth at most $k$ is different from the distinguishing power of homomorphism counts from their respective incidence graphs (\cref{thm:hd-hom-are-skewed}).
Compared to other hypergraph parameters, this is very unexpected.
\subparagraph*{Organisation} The remainder of the paper is organised as follows.
\Cref{sec:preliminaries} is dedicated to the introduction of the necessary notation and definitions.
In particular, we introduce incidence graphs as representations of hypergraphs that will be used throughout the paper, following Böker~\cite{Boeker2019} and Scheidt and Schweikardt~\cite{Scheidt2023}.
The notions of (strict) hypertree depth are introduced in \cref{sec:hypertree-depth}, followed by \cref{sec:homomorphisms} where we handle the differences between homomorphisms between hypergraphs and homomorphisms between incidence graphs.
In \cref{sec:k-labeled-incidence-graphs} we introduce $k$-labeled incidence graphs that were the principle tool used in~\cite{Scheidt2023} to achieve their result.
We utilise them in \cref{sec:characterising-shd} to give an inductive characterisation of the hypergraphs of strict hypertree depth at most $k$ (\cref{thm:inductive_def_hypertree_depth}).
\Cref{sec:gc} is devoted to the logic $\GCk$.
In \cref{sec:main-result} we combine the results from \cref{sec:characterising-shd} and \cref{sec:homomorphisms} with the results from~\cite{Scheidt2023} to obtain \cref{thm:main}.
\Cref{sec:conclusion} concludes the paper with a summary of the results obtained in this paper, as well as an outlook on further research directions. 

\section{Preliminaries}%
\label{sec:preliminaries}

Since we heavily rely on the work by Scheidt and Schweikardt~\cite{Scheidt2023}, we will keep our notation close to theirs.
We denote the set of natural numbers \emph{including} 0 by $\nat$, the set of \emph{positive} natural numbers by $\natpos$, and we write $[n]$ to denote the set $\set{ 1, 2, \dots, n }$.
To denote isomorphism of two objects, we use $\isomorphic$.
A tuple is denoted using a bar, e.g.\ $\tupel{a}$.
For a given $\ell$-tuple $\tupel{a}$, we use $a_i$ to denote the $i$-th element of $\tupel{a}$, i.e., $\tupel{a} = (a_1, a_2, \dots, a_{\ell})$.
For any set $S$, let $\pot(S)$ be the set of subsets of $S$ and let $\kpot(S)$ be the subsets of cardinality exactly $k$.
If $S$ is a set of sets, let $\bigunion S = \bigunion_{s \in S} s$.

For a finite set $S$ of cardinality $\ell \in \nat$, a total order $<$ on $S$ and any number $d \in \nat$, we say that $\enum{ i_{d{+}1}, i_{d{+}2} \dots, i_{d{+}\ell}}$ is the \emph{$<$-enumeration of $S$}, if $i_{d{+}1} < i_{d{+}2} < \cdots < i_{d{+}\ell}$ and $\set{ i_{d{+1}}, \dots, i_{d{+}\ell} } = S$.
If the order $<$ is clear from the context, we simply say that $\enum{ i_{d{+}1}, i_{d{+}2} \dots, i_{d{+}\ell}}$ is the enumeration of $S$.
Note that we usually let $d = 0$, i.e., we usually write $\enum{ i_1, \dots, i_\ell }$.
Furthermore, the enumeration $\enum{ i_{d{+}1}, i_{d{+}2} \dots, i_{d{+}\ell}}$ of $S$ is empty if, and only if, $S$ is empty.

We denote a partial function $f$ from $A$ to $B$ by $f\colon A \pto B$, and we let $\dom(f) \isdef \set{ a \in A \mid f(a) \text{ is defined} }$ and $\img(f) \isdef \set{ b \in B \mid \text{ex. } a \in A \text{ s.t.\ } f(a) = b }$.
We say that two functions $f$ and $g$ are \emph{compatible}, if $f(x) = g(x)$ for all $x \in \dom(f) \intersect \dom(g)$.
We identify a (partial) function $f$ with the set $\set{ (x, f(x)) \mid x \in \dom(f) }$ whenever we are using set notation on functions.
For example, we write $f \subseteq g$ to indicate $\dom(f) \subseteq \dom(g)$ and $f(x) = g(x)$ for all $x \in \dom(f)$.
In particular, by $f \union g$ we denote the function $h$ with $\dom(h) = \dom(f) \union \dom(g)$ and $h(x) = f(x)$ for all $x \in \dom(f)$ and $h(x) = g(x)$ for all $x \in \dom(g) \setminus \dom(f)$.
Note that $f$ has \emph{precedence} over $g$, but this only matters if $f$ and $g$ are not compatible.
For a (partial) function $f$ and a set $S \subseteq \dom(f)$ we write $f(S)$ to denote $\set{ f(x) \mid x \in S }$, and we call the function $g \subseteq f$ with $\dom(g) \isdef S$ the \emph{restriction of $f$ to $S$}.
Finally, we define partial functions inline like this: $\set{ a \to 3, b \to 2, c \to 5 }$.
In particular, the empty set $\emptyset$ denotes a partial function with empty domain.
\subparagraph*{Graphs, Trees and Forests} An (undirected) graph is a tuple $G = (\V(G), \E(G))$, where $\V(G)$ is a finite set and $\E(G) \subseteq \pot_2(\V(G))$.
For a set $S \subseteq \V(G)$, $\induced{G}{S}$ denotes the subgraph induced by $S$, i.e., $\V(\induced{G}{S}) \isdef S$ and $\E(\induced{G}{S}) \isdef \E(G) \intersect \pot_2(S)$.
A connected component of a graph is a maximal induced subgraph that is connected.
A tree is a connected acyclic graph and a forest is a graph were each connected component is a tree.
A rooted tree $\T$ is a tree with some distinguished node that we call its \emph{root}, which we denote by $\troot_\T$.
A rooted forest $\F$ is the disjoint union of a collection of rooted trees.
It therefore has a set of roots denoted by $\troots_\F$.
We may omit the index if it is clear from the context.
Note that a rooted tree is also a rooted forest and that every node $n$ in a rooted forest is contained in a unique connected component which is a tree that we call the \emph{tree for $n$} and whose root is the \emph{root for $n$}.

For a rooted forest $\F$ we let $\leq_{\F}$ be the induced partial order on the nodes, i.e., the roots are the minimal elements and $s \leq_\F t$ if $s$ is on the unique path from $t$ to its root in $\Omega$.
By $\p(s,t)$ we denote the set of nodes on the path from $s$ to $t$ (including $s$ and $t$).
In particular, if no path from $s$ to $t$ exists, $\p(s,t) = \emptyset$.
By $\p(s)$ we denote the set of nodes on the path from $s$ to the root for $s$ and by $\lcv(s,t)$ we denote the unique element, if it exists, where the paths $\p(s)$, $\p(t)$ join, i.e., $\lcv(s,t) \isdef \max_{\leq_\F}(\p(s) \intersect \p(t))$.
Notice that $\lcv(s,t)$ is undefined iff $s$ and $t$ are not in the same tree, and that $\lcv(s,t) = s$, iff $s \leq_\F t$ (and conversely, $\lcv(s,t) = t$ iff $t \leq_\F s$).

The subtree $\T_t$ induced by $t \in \V(\F)$ is the tree $\induced{\F}{V}$ with root $t$ and $V \isdef \set{ s \in \V(\T) \mid t \leq_\F s }$.
The \emph{level} of a node $s \in \V(\F)$ is defined as the number of elements on the path from $s$ to its root, i.e., $\level(s) \isdef |\p(s)|$.
The \emph{height} of a rooted tree $\T$ is the maximal level, i.e., $\height(\T) \isdef \max \set{ \level(t) \mid t \in \V(\T) }$ and the height of a node $t \in \V(\T)$ is the height of its induced subtree $\T_t$, i.e., $\height(t) \isdef \height(\T_t)$.

\subparagraph*{Hyper- and Incidence Graphs}
A \emph{hypergraph} is a tuple $\HG = (V, E, \beta)$, where $V$ and $E$ are disjoint finite sets and $\beta$ is a total function from $E$ to $\pot(V)$ with $\V = \bigunion_{e \in E} \beta(e)$.
We call the elements of $V$ \emph{vertices} and the elements of $E$ \emph{hyperedges} and for every $e \in E$, we call $\beta(e)$ its \emph{contents}.
We denote $V$ by $\V(\HG)$, $E$ by $\E(\HG)$ and $\beta$ by $\f_{\HG}$, though we may omit the index if there is no ambiguity.
Notice that, in general, multiple hyperedges with the same content and hyperedges without content are allowed.
We call $\HG$ \emph{simple} if $\f$ is injective.

An \emph{incidence graph} is a tuple $\I = (R, B, E)$ consisting of two disjoint finite sets $R$ and $B$ of \emph{red} and \emph{blue} vertices and a relation $E \subseteq B \times R$.
We denote $R$ by $\R(\I)$, $B$ by $\B(\I)$ and $E$ by $\E(\I)$.
For every $e \in \B(\I)$, we let $\blueN(e) \isdef \set{ v \in \R(\I) \mid (e,v) \in \E(\I) }$.
Notice that $\blueN$ is equivalent in its function to the map $\f$ for a hypergraph, hence we denote them similarly.
We only consider incidence graphs where for every $v \in \R(\I)$ there is an $e \in \B(\I)$ such that $(e,v) \in \E(\I)$.

It is easy to see that we can assign an incidence graph to every hypergraph and the other way around: For every hypergraph $\HG$ we let $\I_{\HG} \isdef (\V(\HG), \E(\HG), E)$ where $E \isdef \set{ (e, v) \mid e \in \E(\HG), v \in \f(e) }$.
Conversely, for every incidence graph $\I$ we let $\HG_\I \isdef (\R(\I), \B(\I), \f)$ where $\f(e) \isdef \set{ v \in \R(\I) \mid (e, v) \in \E(\I) }$ for all $e \in \B(\I)$.

For every set $S \subseteq \E(\HG)$ we define the \emph{induced subhypergraph} $\induced{\HG}{S}$ as $(V', S, \f_\HG')$ where $V' \isdef \bigunion_{e \in S} \f(e)$ and $\f_\HG'$ is the restriction of $\f_\HG$ to $S$.
We say that a hypergraph is connected if its incidence graph is connected.
An induced subhypergraph is a connected component, if its corresponding incidence graph is a connected component.

By $\Pn_{n}$ we denote the path of $n$ hyperedges, where each hyperedge contains 2 vertices.
I.e., we let $\V(\Pn_n) = [n{+}1]$, $\E(\Pn_n) = \set{ e_i \mid i \in [n] }$ and $\f(e_i) = \set{ i, i{+}1 }$ for all $i \in [n]$.
We may use different names for the vertices if it is convenient.

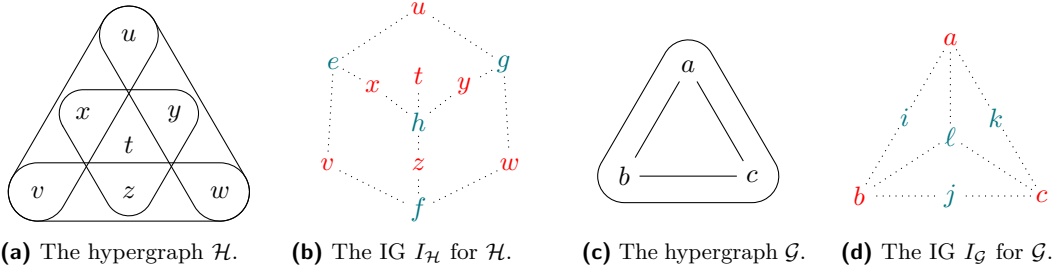
\begin{figure}
	\centering
	\begin{subfigure}[b]{0.24\textwidth}
		\centering

\begin{tikzpicture}[
	every node/.style={regular},
]
	\node (a) at (0,0) {$v$};
	\node (b) at (\exampled,0) {$w$};
	\node (f) at ($(a)!.5!(b)$) {$z$};
	\node (c) at ($(f)+(0,\exampleh)$) {$u$};
	\node (d) at ($(a)!.5!(c)$) {$x$};
	\node (e) at ($(b)!.5!(c)$) {$y$};
	\coordinate (m) at ($(d)!.5!(e)$);
	\node (g) at ($(m)!.4!(f)$) {$t$};
	\draw[draw=black] \convexpath{d,e,f}{9pt};
	\draw[draw=black] \convexpath{a,b}{11pt};
	\draw[draw=black] \convexpath{a,c}{11pt};
	\draw[draw=black] \convexpath{b,c}{11pt};
\end{tikzpicture}

 		\caption{The hypergraph $\HG$.}%
		\label{fig:example:complex-hg}
	\end{subfigure}
	\hfill
	\begin{subfigure}[b]{0.24\textwidth}
		\centering

\begin{tikzpicture}[
	every node/.style={regular},
]
	\tikzmath{
		\dist=\exampled/4;
	}

	\node[red] (a) at (0,0) {$v$};
	\node[red] (b) at (\exampled,0) {$w$};
	\node[red] (f) at ($(a)!.5!(b)$) {$z$};
	\node[red] (c) at ($(f)+(0,\exampleh)$) {$u$};
	\node[red] (d) at ($(a)!.5!(c)$) {$x$};
	\node[red] (e) at ($(b)!.5!(c)$) {$y$};
	\coordinate (m) at ($(d)!.5!(e)$);

	\node[blue] (i) at ($(d)+(150:\dist)$) {$e$};
	\node[blue] (j) at ($(f)+(-90:\dist)$) {$f$};
	\node[blue] (k) at ($(e)+(30:\dist)$) {$g$};
	\node[blue] (l) at ($(m)!.45!(f)$) {$h$};

	\node[red] (g) at ($(l)+(90:{\dist})$) {$t$};

	\graph[edges={edge}]{
		(i) -- {
			(a), (d), (c)
		},
		(k) -- {
			(c), (e), (b)
		},
		(j) -- {
			(a), (f), (b)
		},
		(l) -- {
			(d), (e), (f), (g)
		}
	};
\end{tikzpicture}
 		\caption{The IG $\I_\HG$ for $\HG$.}%
		\label{fig:example:complex-ig}
	\end{subfigure}
	\hfill
	\begin{subfigure}[b]{0.21\textwidth}
		\centering

\begin{tikzpicture}[
]
	\tikzmath{
		\mydiam=\exampled*.7;
		\myheight=sqrt(3/4*(\mydiam*\mydiam));
	}
	\node (a) at (0,0) {$b$};
	\node (b) at ($(a) + (\mydiam,0)$) {$c$};
	\node (c) at ($(a)!.5!(b) + (0,\myheight)$) {$a$};
	\draw (a) -- (b) -- (c) -- (a);
	\draw[draw=black] \convexpath{a,c,b}{10pt};
\end{tikzpicture}
 		\vspace{.3cm}
		\caption{The hypergraph $\GG$.}%
		\label{fig:example:simple-hg}
	\end{subfigure}
	\hfill
	\begin{subfigure}[b]{0.21\textwidth}
		\centering

\begin{tikzpicture}[
	every node/.style={regular},
]
	\tikzmath{
		\mydiam=\exampled*1;
		\myheight=sqrt(3/4*(\mydiam*\mydiam));
	}
	\node[red] (a) at (0,0) {$b$};
	\node[red] (b) at ($(a) + (\mydiam,0)$) {$c$};
	\node[red] (c) at ($(a)!.5!(b) + (0,\myheight)$) {$a$};

	\node[blue] (i) at ($(a)!.5!(c)$) {$i$};
	\node[blue] (j) at ($(a)!.5!(b)$) {$j$};
	\node[blue] (k) at ($(b)!.5!(c)$) {$k$};
	\node[blue] (l) at ($(a)!.5!(b) + (0,{.38*\myheight})$) {$\ell$};

	\graph[edges={edge}]{
		(i) -- {
			(a), (c)
		},
		(k) -- {
			(c), (b)
		},
		(j) -- {
			(a), (b)
		},
		(l) -- {
			(a), (c), (b)
		}
	};
\end{tikzpicture}
 		\vspace{.2cm}
		\caption{The IG $\I_\GG$ for $\GG$.}%
		\label{fig:example:simple-ig}
	\end{subfigure}
	\caption{Examples for hypergraphs and their corresponding incidence graphs.}%
	\label{fig:example:hgs-and-igs}
\end{figure}
\begin{example}\label{example:hgs-and-igs}
	The hypergraph $\HG$ illustrated in \cref{fig:example:complex-hg} is defined as 
	$\V(\HG) = \set{ u,v,w,x,\allowbreak y,z,t }$, 
	$\E(\HG) = \set{ e,f,g,h }$ and 
	$\f_\HG = \set{ e \to \set{ u,v,x }, f \to \set{ v,w,z }, g \to \set{ u,w,y },\allowbreak h \to \set{ t,x,y,z } }$.
	Its incidence graph $\I_\HG$, depicted in \cref{fig:example:complex-ig}, is defined by 
	$\R(\I_\HG) = \V(\HG)$, 
	$\B(\I_\HG) = \E(\HG)$ and 
	$\E(\I_\HG) = \set{ (e,u), (e,v), (e,x), (f,v), (f,w), (f,z), (g,u), (g,w), (g,y),\allowbreak (h,t), (h,x), (h,y), (h,z) }$.

	The hypergraph $\GG$ depicted in \cref{fig:example:simple-hg} is defined as  
	$\V(\GG) = \set{ a,b,c }$, 
	$\E(\GG) = \set{ i,j,k,\ell }$ and 
	$\f_\GG \isdef \set{ i \to \set{ a,b }, j \to \set{ b,c }, k \to \set{ a,c }, \ell \to \set{ a,b,c } }$.
	Its incidence graph $\I_\GG$, depicted in \cref{fig:example:simple-ig}, is defined as 
	$\R(\I_\GG) = \V(\GG)$, 
	$\B(\I_\GG) = \E(\GG)$ and 
	$\E(\I_\GG) = \set{ (i,a), (i,b), (j,b), (j,c),\allowbreak (k,a), (k,c), (\ell,a), (\ell,b), (\ell,c) }$.
\end{example}

\subsection{Hypertree Depth}%
\label{sec:hypertree-depth} %
The following definition of \emph{elimination forest} and \emph{hypertree depth} is due to Adler, Gaven\v{c}iak and Klimo\v{s}ová~\cite{Adler2012}, though they refer to elimination forests as \enquote{decomposition forests}.
We call them elimination forests, since this reflects their conceptual similarity to elimination forests for graphs and avoids confusion with hypertree decompositions.
Further, we define this notion in terms of incidence graphs, because we mainly work on those.
Do notice however, that this definition easily translates to hypergraphs and that it is equivalent to the one given by Adler, Gaven\v{c}iak and Klimo\v{s}ová.
\begin{definition}[Hypertree Depth and Elimination Forests,~\cite{Adler2012}]%
	\label{def:hd}
	Let $\I$ be an incidence graph.
	An \emph{elimination forest} $(\F, \Gamma)$ for $\I$ consists of a forest $\F$ and a mapping $\Gamma\colon \V(\F) \to \B(\I)$ such that \cref{def:hd:vertex-complete,def:hd:edge-containment,def:hd:shared-heritage} hold.
	We write $\hGamma(t)$ as shorthand for $\blueN(\Gamma(t))$.
	\begin{enumerate}\crefalias{enumi}{condition}
		\item\label{def:hd:vertex-complete} \emph{Completeness for vertices}:
		For every red vertex $v \in \R(\I)$, there is a $t \in \V(\F)$ such that $v \in \hGamma(t)$.
		\item\label{def:hd:edge-containment} \emph{Hyperedge-Containment}:
		For every blue vertex $e \in \B(\I)$ there are nodes $s, t \in \V(\F)$ such that $s \leq_\F t$ and $\blueN(e) \subseteq \bigunion \hGamma(\p(s,t))$.
		\item\label{def:hd:shared-heritage} \emph{Shared heritage}:
		For all $s, t \in \V(\F)$, if $\hGamma(s) \intersect \hGamma(t) \neq \emptyset$, then $\lcv(s,t)$ is defined and $\hGamma(s) \intersect \hGamma(t) \subseteq \bigunion \hGamma(\p(\lcv(s,t)))$.
	\end{enumerate}

	The intuition behind \cref{def:hd:shared-heritage} is that hyperedges can only share the vertices contained in their common ancestors in the elimination forest.
	The height of an elimination forest $(\F, \Gamma)$ is simply the height of $\F$.
	The \emph{hypertree depth of $\I$} is defined as the minimal height over all elimination forests for $\I$, and we denote it by $\hd(\I)$.
	Analogously, we let $\hd(\HG) \isdef \hd(\I_\HG)$ for all hypergraphs.
	We write $\IHDk$ to denote the class of incidence graphs of hypertree depth at most $k$ and $\HDk$ to denote the corresponding class of hypergraphs.
\end{definition}

We call an elimination forest $(\F, \Gamma)$ \emph{strict} if $\Gamma$ is bijective.
It is easy to see that the first and second condition are trivially satisfied when $\Gamma$ is bijective, thus we can redefine the notion of strict elimination forest in a more succinct manner.

\begin{definition}[Strict Elimination Forest]%
	\label{def:shd}
	Let $\I$ be an incidence graph.
	A \emph{strict elimination forest} $(\F, \Gamma)$ for $\I$ consists of a forest $\F$ and a \emph{bijective} function $\Gamma\colon \V(\F) \to \B(\I)$ satisfying \cref{def:hd:shared-heritage} of \cref{def:hd}.
	The \emph{strict hypertree depth of $\I$}, denoted by $\shd(\I)$, is defined as the minimal height over all strict elimination forests for $\I$.
	Again, we let $\shd(\HG) \isdef \shd(\I_\HG)$ for every hypergraph $\HG$.
	We write $\ISHDk$ to denote the class of incidence graphs of strict hypertree depth at most $k$ and $\SHDk$ to denote the corresponding class of hypergraphs.
\end{definition}

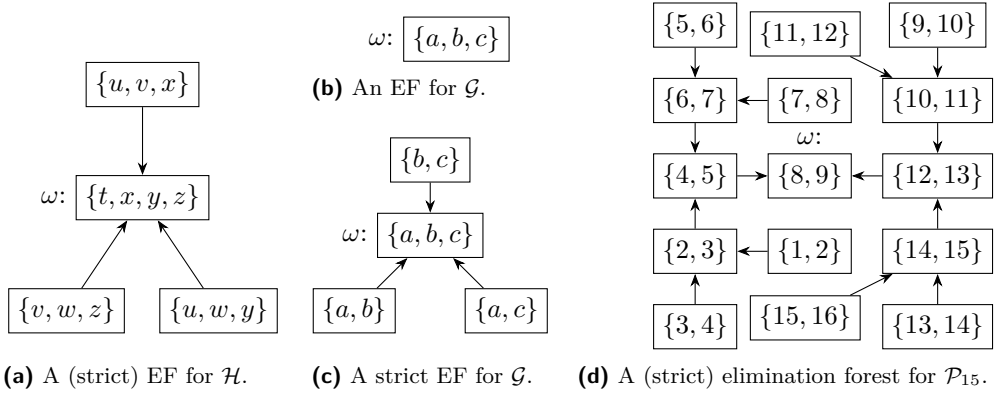
\begin{figure}
	\centering
	\begin{subfigure}{0.27\textwidth}
		\centering

\begin{tikzpicture}[
	every node/.style={draw, rectangle},
	sibling distance=2cm,
]
	\node {$\set{ u,v,x }$}
	child { node[label=left:{$\troot$:}]{ $\set{ t,x,y,z }$ } 
		edge from parent [-Stealth]
		child { node {$\set{ v,w,z }$} edge from parent [Stealth-] }
		child { node {$\set{ u,w,y }$} edge from parent [Stealth-] }
	};
\end{tikzpicture} 		\medskip
		\caption{A (strict) EF for $\HG$.}%
		\label{fig:example:complex-hd}
	\end{subfigure}
	\hfill
	\begin{subfigure}{0.23\textwidth}
		\centering

\begin{tikzpicture}[
	every node/.style={draw, rectangle}
]
	\node[label=left:{$\troot$:}]{ $\set{ a,b,c }$ };
\end{tikzpicture} %
		\caption{An EF for $\GG$.}%
		\label{fig:example:simple-hd}%
		\bigskip

\begin{tikzpicture}[
	every node/.style={draw, rectangle},
	sibling distance=2cm,
	level distance=1cm,
]
	\node {$\set{ b,c }$}
	child { node[label=left:{$\troot$:}]{ $\set{ a,b,c }$ } 
		edge from parent [-Stealth]
		child { node {$\set{ a,b }$} edge from parent [Stealth-] }
		child { node {$\set{ a,c }$} edge from parent [Stealth-] }
	};
\end{tikzpicture} %
		\medskip
		\caption{A strict EF for $\GG$.}%
		\label{fig:example:simple-shd}%
	\end{subfigure}
	\hfill
	\begin{subfigure}{.46\textwidth}
		\centering

\begin{tikzpicture}[
	every node/.style={draw, rectangle},
	every edge/.style={draw, Stealth-},
	node distance=.4cm,
]
	\node (0) [label=above:{$\troot$:}] at (0,0) {$\set{ 8,9 }$};
	\node (00) [left= of 0] {$\set{ 4,5 }$};
	\node (000) [below= of 00] {$\set{ 2,3 }$};
	\node (001) [above= of 00] {$\set{ 6,7 }$};
	\node (0000) [right= of 000] {$\set{ 1,2 }$};
	\node (0001) [below= of 000] {$\set{ 3,4 }$};
	\node (0010) [above= of 001] {$\set{ 5,6 }$};
	\node (0011) [right= of 001] {$\set{ 7,8 }$};

	\node (01) [right= of 0] {$\set{ 12,13 }$};
	\node (010) [above= of 01] {$\set{ 10,11 }$};
	\node (011) [below= of 01] {$\set{ 14,15 }$};
	\node (0100) [above= of 010] {$\set{ 9,10 }$};
	\node (0101) [above left= of 010] {$\set{ 11,12 }$};
	\node (0110) [below= of 011] {$\set{ 13,14 }$};
	\node (0111) [below left= of 011] {$\set{ 15,16 }$};

	\draw (0) edge (00);
	\draw (00) edge (000);
	\draw (00) edge (001);
	\draw (000) edge (0000);
	\draw (000) edge (0001);
	\draw (001) edge (0010);
	\draw (001) edge (0011);
	\draw (0) edge (01);
	\draw (01) edge (010);
	\draw (01) edge (011);
	\draw (010) edge (0100);
	\draw (010) edge (0101);
	\draw (011) edge (0110);
	\draw (011) edge (0111);
\end{tikzpicture} 		\caption{A (strict) elimination forest for $\Pn_{15}$.}%
		\label{fig:example:path-shd}
	\end{subfigure}
	\caption{Elimination forests for various (hyper)graphs.}%
	\label{fig:example:hds-and-shds}
\end{figure}
\begin{example}
	Consider the hypergraphs $\GG$ and $\HG$ as well as their incidence graphs $\I_\GG$, $\I_\HG$ from \cref{example:hgs-and-igs} (see \cref{fig:example:hgs-and-igs}).
	Let $(\F_1, \Gamma_1)$ be defined as follows.
	$\F_1$ is a tree defined by $\V(\F_1) = \set{ t_1, t_2, t_3, t_4 }$ and $\E(\F_2) \isdef \set{ \set{ t_1, t_2 }, \set{ t_1, t_3 }, \set{ t_1, t_4 } }$ with root $t_1$.
	$\Gamma_1$ is a map defined as $\set{ t_1 \to h, t_2 \to e, t_3 \to f, t_4 \to g }$.
	$(\F_1, \Gamma_1)$ is an elimination forest of height $2$ for $\I_\HG$.
	It is depicted in \cref{fig:example:complex-hd}, where we labeled the node $t_i$ with $\hGamma_1(t_i)$.
	Notice, that $(\F_1, \Gamma_1)$ is strict.
	
	Analogously, we depicted two elimination forests for $\I_\GG$ of height 1 and 2 in \cref{fig:example:simple-hd,fig:example:simple-shd}.
	Notice how the elimination forest depicted in \cref{fig:example:simple-hd}, witnessing $\hd(\GG) = 1$, is not strict.
	It is easy to see that $\shd(\GG) \geq 2$ since a bijective map implies that there are as many nodes in the forest as there are edges.
	Thus, \cref{fig:example:simple-shd} witnesses that $\shd(\GG) = 2$.
	
	Finally, \cref{fig:example:path-shd} depicts a (strict) elimination forest for $\Pn_{15}$ witnessing $\shd(\Pn_{15}) \leq 4$.
	Recall that $\Pn_{15}$ is defined by $\V(\Pn_{15}) = [16]$ and $\E(\Pn_{15}) = \set{ \set{ i, i{+}1 } \mid i \in [15] }$.
	\lipicsEnd{}
\end{example}

One can show that the strict hypertree depth of a hypergraph is at most its hypertree depth increased by one.
The main idea is that we can turn every elimination forest into a strict one, if we add a leaf for every hyperedge that is currently not being mapped to below the path that contains it according to \cref{def:hd:edge-containment} in \cref{def:hd}.

\begin{restatable}{theorem}{shdisalmosthd}\label{thm:shd-is-almost-hd}
	For all hypergraphs $\HG$,\;
	$
		\hd(\HG) \; \leq \; \shd(\HG) \; \leq \; \hd(\HG) {+} 1
	$.
\end{restatable}
\begin{proof}
Let $\I$ be an incidence graph.
Since every strict elimination forest is also an elimination forest, $\hd(\I) \leq \shd(\I)$ trivially holds.

To show $\shd(\I) \leq \hd(\I) + 1$ it suffices to show that any elimination forest $(\F, \Gamma)$ of $\I$ can be turned into a strict elimination forest $(\sF, \sGamma)$ of $\I$ where $\height(\sF) \leq \height(F)+1$.
In the following we make $\Gamma$ bijective by inserting new nodes into $\F$ as children of leaves.
Thus, at the end, the height of $\sF$ has increased by at most 1.

But first we show that we may assume that $\Gamma$ is already injective.
If it is not, then there are distinct $s,t \in \V(\F)$ such that $\Gamma(s) = \Gamma(t)$ and w.l.o.g.\ $\level(s) \geq \level(t)$ and $s$ is not the root.
Let $S \isdef \hGamma(s) \intersect \hGamma(t)$\; ($= \hGamma(s) = \hGamma(t)$).
Since $(\F, \Gamma)$ is an elimination forest, $\lcv(s,t)$ must be defined, $\lcv(s,t) \neq s$ and $S \subseteq \bigunion \hGamma(\p(\lcv(s,t)))$ must hold according to \cref{def:hd:shared-heritage}.
If we remove the node $s$ from $\F$ by contracting the edge between $s$ and its parent $p$ (which exists since $s$ can not be the root) and assign $\Gamma(p)$ to the resulting node, the result $(\F', \Gamma')$ is still an elimination forest.
\Cref{def:hd:vertex-complete} is still satisfied because $\hGamma(s) = \hGamma(t)$.
It can easily be verified that \cref{def:hd:edge-containment} is satisfied: Let $e \in \B(\I)$, let $u,v \in \V(\F)$ such that $u \leq_\F v$ and $\blueN(e) \subseteq \bigunion \hGamma(\p(u,v))$.
Obviously, $\blueN(e) \subseteq \bigunion \hGamma(\p(u,\troot))$ where $\troot$ is the root for $u$.
Since $\hGamma(s) \subseteq \bigunion \hGamma(\lcv(s,t))$ and $\lcv(s,t) \in \p(s,\omega) \setminus \set{ s }$, $\p(u, \troot) \setminus \set{ s }$ is still a path in $\F'$ and $\blueN(e) \subseteq \bigunion \hGamma'(\p(u, \troot) \setminus \set{ s })$.
Showing that \cref{def:hd:shared-heritage} is still satisfied is not hard either:
Let $u,v \in \V(\F)$ such that $\hGamma(u) \intersect \hGamma(v)$ is not empty.
If $\lcv(u,v) \leq_\F s$ or $\lcv(u,v)$ is not $\leq_\F$-comparable to $s$, then $\hGamma(u) \intersect \hGamma(v) \subseteq \bigunion \hGamma(\p(\lcv(u,v)))$ obviously still holds.
If $s \leq_\F \lcv(u,v)$, then $s \in \p(\lcv(u,v))$ and $\lcv(s,t) \in \p(\lcv(u,v))$.
Using the same argument as before we can show that $\hGamma'(u) \intersect \hGamma'(v) \subseteq \bigunion \hGamma'(\p(\lcv(u,v)))$ still holds.

In total, we get that $(\F', \Gamma')$ is still an elimination forest and $\height(\F') \leq \height(\F)$.
We can repeat this procedure until the resulting $\Gamma'$ is injective.

Assume that $\Gamma$ is already injective.
If $\Gamma$ is surjective, $\Gamma$ is also bijective.
Otherwise, let $\set{ e_1, e_2, \dots, e_\ell } \isdef \B(\I) \setminus \img(\Gamma)$ be the set of blue vertices that are not being mapped to.
Since $(\F, \Gamma)$ is an elimination forest, there exist $s_i, t_i$ for every $i \in [\ell]$ such that $s_i \leq_\F t_i$ and $\blueN(e_i) \subseteq \bigunion \hGamma(\p(s_i,t_i))$ and where $|\p(s_i, t_i)|$ is minimal.
We insert a fresh node $s_i'$ as a child of $s_i$ for every $e_i$ and map $s_i'$ to $e_i$.
In precise terms, we let $\sF = (\V(\sF), \E(\sF))$ with $\V(\sF) = \V(\F) \disunion \set{ s_1', s_2', \dots, s_\ell' }$ and $\E(\sF) = \E(\F) \union \set{ \set{ s_i, s_i' } \mid i \in [\ell] }$; and $\sGamma(s_i') \isdef e_i$ for all $i \in [\ell]$ and $\sGamma(t) = \Gamma(t)$ for all $t \in \V(\F)$.

Obviously, the height of the tree $\sF$ did not increase by more than 1 relative to $\F$ and the resulting $\sGamma$ is bijective.
Since $\R(\I) = \bigunion \blueN(\B(\I))$ by definition, \cref{def:hd:vertex-complete} is trivially satisfied.
Since $\sGamma$ is bijective, \cref{def:hd:edge-containment} is trivially satisfied as well.
It remains to verify \cref{def:hd:shared-heritage}: 
Let $u,v \in \V(\sF)$ such that $S \isdef \hsGamma(u) \intersect \hsGamma(v)$ is not empty.
If $u,v \in \V(\F)$ then this still holds since $(\F, \Gamma)$ is an elimination forest and all $s_i'$ are leaves.
Thus assume w.l.o.g.\ that $u = s_i'$ for some $i \in [\ell]$.
Let $n_1, \dots, n_{\hi}$ be the nodes in $\p(s_i, t_i)$ where $S \intersect \hsGamma(n_j)$ is not empty.

\emph{Case 1}: If $v \neq s_j'$ for every $j \in [\ell]$, the nodes $\lcv(n_j, v)$ must exist and $\emptyset \neq \hGamma(n_j) \intersect \hGamma(v) \subseteq \bigunion \hGamma(\p(\lcv(n_j,v)))$ holds.
Thus, also $\lcv(u,v)$ must exist.
Since $\F$ is a forest, all $\lcv(n_j, v)$ must lie on a path, since otherwise $\F$ would contain a cycle.
This means that for the largest $\lcv(n_j, v)$ according to $\leq_\F$, both $\lcv(n_j,v) \leq_{\sF} \lcv(u, v)$ and $S \subseteq \bigunion \hGamma(\p(\lcv(n_j, v)))$ hold.
Thus, $S \subseteq \bigunion \hsGamma(\p(\lcv(u,v)))$.

\emph{Case 2}: If $v = s_j'$ for a $j \in [\ell]$, let $m_1, \dots, m_{\hj}$ be the nodes in $\p(s_j, t_j)$ where $S \intersect \hsGamma(m_{j'})$ is not empty.
Consider the set of pairs $M \isdef \set{ (n_{i'}, m_{j'}) \mid \hGamma(n_{i'}) \intersect \hGamma(m_{j'}) \neq \emptyset }$.
Again, for all $(n_{i'}, m_{j'}) \in M$, $\lcv(n_{i'}, m_{j'})$ must exist and $\hGamma(n_{i'}) \intersect \hGamma(m_{j'}) \subseteq \bigunion \hGamma(\p(\lcv(n_{i'}, m_{j'})))$.
Again, since $\F$ is a forest, they must all lie on a path and $\lcv(u,v)$ must exist.
For the largest $\lcv(n_{\tilde{\imath}}, m_{\tilde{\jmath}})$ according to $\leq_\F$, it therefore holds that $\hGamma(n_{i'}) \intersect \hGamma(m_{j'}) \subseteq \bigunion \hGamma(\p(\lcv(n_{\tilde{\imath}}, m_{\tilde{\jmath}})))$ for \emph{all} $(n_{i'}, m_{j'}) \in M$.
Since $\blueN(u) \subseteq \bigunion \hGamma(\p(s_i, t_i))$ and $\blueN(v) \subseteq \bigunion \hGamma(\p(s_j, t_j))$ hold by construction, it must also hold that $
	S = \bigunion_{(n_{i'}, m_{j'}) \in M} \hGamma(n_{i'}) \intersect \hGamma(m_{j'})
$,
and since $\lcv(n_{\tilde{\imath}}, m_{\tilde{\jmath}}) \leq_{\sF} \lcv(u,v)$, we are done. \end{proof}

Before closing this section, it is easy to see that, due to \cref{def:hd:shared-heritage} (shared heritage), every elimination forest of a connected incidence graph is also an \emph{elimination tree}.
\begin{lemma}\label{lem:restrict_to_connected}
	Let $\I$ be a connected incidence graph.
	For every elimination forest $(\F, \Gamma)$ of $\I$, $\F$ is a tree.
\end{lemma} 

\subsection{Homomorphisms}%
\label{sec:homomorphisms} %
While hypergraphs and incidence graphs are conceptually close, their \enquote{natural} notions of homomorphisms are not the same.
Since our interest lies in hypergraphs, but we are mainly working on incidence graphs in this paper, we have to relate these notions.
Following Scheidt and Schweikardt, we use the same definitions as Böker~\cite{Boeker2019}.

A homomorphism from a hypergraph $\HG$ into another hypergraph $\GG$ is a pair of functions $(\hV\colon \V(\HG) \to \V(\GG), \hE\colon \E(\HG) \to \E(\GG))$ such that for every $e \in \E(\HG)$ the equality $\hV(\f(e)) = \f(\hE(e))$ holds.

A homomorphism from an incidence graph $\I$ into another incidence graph $\J$ is a pair of mappings $(\hV\colon \R(\I) \to \R(\J), \hE\colon \B(\I) \to \B(\J))$ such that $(\hE(e), \hV(v)) \in \E(\J)$ holds for every edge $(e, v) \in \E(\I)$.
This is equivalent to the requirement $\hV(\blueN(e)) \subseteq \blueN(\hE(e))$.
Thus, the equality that we require for a hypergraph homomorphism is relaxed to an inclusion for incidence graphs.

Let $A$, $B$ and $\classC$ be two hypergraphs and a class of hypergraphs or two incidence graphs and a class of incidence graphs.
We denote the number of homomorphisms from $A$ to $B$ by $\hom(A,B)$, and we let $\Hom(\classC, A)$ be the \enquote{vector} that has a row for every $F \in \classC$ containing $\hom(F,A)$.
We say that $A$ and $B$ are \emph{homomorphism indistinguishable} over $\classC$ ($A \equiv_{\classC} B$), if $\Hom(\classC, A) = \Hom(\classC, B)$, i.e., if $\hom(F, A) = \hom(F, B)$ for all $F \in \classC$.

The following crucial theorem relates homomorphism indistinguishability over a class of hypergraphs to homomorphism indistinguishability over the corresponding class $\classIC$ of incidence graphs.
As noted in~\cite{Scheidt2023}, this theorem is implicit in~\cite{Boeker2019}, consult Appendix A of the full version of~\cite{Scheidt2023} for details.
\begin{theorem}[\cite{Boeker2019,Scheidt2023}]\label{thm:ihom-equals-hom}
	Let $\classC$ be a class of hypergraphs and let $\classIC$ be its corresponding class of incidence graphs.
	If $\classC$ is closed under \emph{pumping} and \emph{local merging}, then $\Hom(\classC, \GG) = \Hom(\classC, \HG)$ if, and only if, 
	$\Hom(\classIC, \I_\GG) = \Hom(\classIC, \I_\HG)$ for all hypergraphs $\GG$ and $\HG$.
\end{theorem}
$\classC$ is closed under pumping, if $H' \in \classC$ for every $H \in \classC$, where $H'$ is created from $H$ by inserting a new vertex into \emph{one} arbitrary hyperedge of $H$; and closed under local merging, if $H' \in \classC$ for every $H \in \classC$, where $H'$ is created from $H$ by choosing an arbitrary hyperedge $e$ and then merging two vertices $u,v$ that are both contained in $e$.

It is easy to see that $\SHDk$ is closed under both pumping and local merging, whereas $\HDk$ is only closed under local merging.
See \cref{app:shd-hd-closures} for details.
\begin{proposition}\label{prop:shd-hd-closures}
	Let $k \in \natpos$.
	The class $\SHDk$ is closed under pumping and local merging, the class $\HDk$ is closed under local merging but \textbf{not} under pumping.
\end{proposition}

The following theorem shows that homomorphism indistinguishability over $\HDk$ is not equal to homomorphism indistinguishability over $\SHDk$.
Because $\SHDk \subseteq \HDk$, counting homomorphisms from $\HDk$ is more powerful in the sense that it distinguishes more hypergraphs.
But we also show that it is unequal to homomorphism indistinguishability over $\IHDk$, which is unexpected.
Since this prohibits us from relating $\HDk$ to any fragment of $\GC$, it is conceivable that this could pose a problem in other scenarios too.
Thus, we argue that the notion of strict hypertree depth can be viewed as a reasonable generalisation of tree depth, especially when we recall that $\HD_{k-1} \subseteq \SHDk \subseteq \HDk$.

\begin{figure}
	\centering
	\begin{subfigure}{0.2\textwidth}
		\centering

\begin{tikzpicture}[
	every node/.style={regular, inner sep=0pt}
]
	\node[draw, circle, inner sep=5pt] (s) at (1,1) {};
	\node[draw, circle, inner sep=5pt] (s2) at (1,0) {};
	\node (a) at (0,0) {$\bullet$};
	\node (b) at (0,1) {$\bullet$};
	\node (c) at (1,0) {$\bullet$};
	\node (d) at (1,1) {$\bullet$};

	\draw (a) -- (b) -- (d) -- (c) -- (a);

\end{tikzpicture} 		\smallskip
		\caption{$\GG_1$}
	\end{subfigure}
	\hfill
	\begin{subfigure}{0.2\textwidth}
		\centering

\begin{tikzpicture}[
	every node/.style={regular, inner sep=0pt}
]
	\node[draw, circle, inner sep=5pt] (s) at (0,1) {};
	\node[draw, circle, inner sep=5pt] (s2) at (1,0) {};
	\node (a) at (0,0) {$\bullet$};
	\node (b) at (0,1) {$\bullet$};
	\node (c) at (1,0) {$\bullet$};
	\node (d) at (1,1) {$\bullet$};

	\draw (a) -- (b) -- (d) -- (c) -- (a);

\end{tikzpicture} 		\smallskip
		\caption{$\HG_1$}
	\end{subfigure}
	\hfill
	\begin{subfigure}{0.2\textwidth}
		\centering

\begin{tikzpicture}[
	every node/.style={regular, inner sep=0pt}
]
	\node (s) at (1.5,1) {$\bullet$};
	\node (s2) at (1.5,0) {$\bullet$};
	\node (a) at (0,0) {$\bullet$};
	\node (b) at (0,1) {$\bullet$};
	\node (c) at (1,0) {$\bullet$};
	\node (d) at (1,1) {$\bullet$};

	\draw (a) -- (b) -- (d) -- (c) -- (a);
	\draw (s) -- (d);
	\draw (s2) -- (c);
\end{tikzpicture} 		\medskip
		\caption{$\GG_1'$}
	\end{subfigure}
	\hfill
	\begin{subfigure}{0.2\textwidth}
		\centering

\begin{tikzpicture}[
	every node/.style={regular, inner sep=0pt}
]
	\node (s) at (-.5,1) {$\bullet$};
	\node (s2) at (1.5,0) {$\bullet$};
	\node (a) at (0,0) {$\bullet$};
	\node (b) at (0,1) {$\bullet$};
	\node (c) at (1,0) {$\bullet$};
	\node (d) at (1,1) {$\bullet$};

	\draw (a) -- (b) -- (d) -- (c) -- (a);
	\draw (s) -- (b);
	\draw (s2) -- (c);
\end{tikzpicture} 		\medskip
		\caption{$\HG_1'$}
	\end{subfigure}
	\caption{%
		$(\GG_1, \HG_1)$, $(\GG_1', \HG_1')$  witness \cref{thm:hd-hom-are-skewed} for $k=1$.
		Circles denote singleton hyperedges.%
	}\label{fig:hd-hom-are-skewed}
\end{figure}
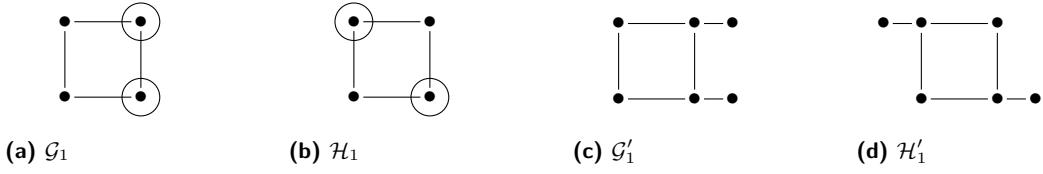

\begin{restatable}{theorem}{hdhomareskewed}\label{thm:hd-hom-are-skewed}%
	\settowidth{\myA}{$\Hom(\SHDk, \GG_k) = \Hom(\SHDk, \HG_k)$,}%
	For every $k \in \natpos$ there exist pairs of hypergraphs $(\GG_k, \HG_k)$ and $(\GG_k', \HG_k')$, such that: 
	\begin{enumerate}\crefalias{enumi}{statement}
		\item\label{stmt:hd-hom-are-sked:towards-shd}
		\makebox[\myA]{$\Hom(\SHDk, \GG_k) = \Hom(\SHDk, \HG_k)$,}\: but\: $\Hom(\HDk, \GG_k) \neq \Hom(\HDk, \HG_k)$;
		
		\item\label{stmt:hd-hom-are-sked:towards-ihd}
		\makebox[\myA][l]{$\Hom(\HDk, \GG_k') = \Hom(\HDk, \HG_k')$,}\: but\: $\Hom(\IHDk, \I_{\GG_k'}) \neq \Hom(\IHDk, \I_{\HG_k'})$.
	\end{enumerate}
	
\end{restatable}

For $k=1$ this is easy to see: A connected hypergraph has strict hypertree depth $1$ iff it consists of a single hyperedge, whereas a connected hypergraph has hypertree depth $1$ if one hyperedge contains all vertices.
It is therefore not hard to see that the statement of the theorem holds for $k=1$ using the hypergraphs depicted in \cref{fig:hd-hom-are-skewed}.
For $k \geq 2$ a similar idea for the construction of $(\GG_k, \HG_k)$ and $(\GG_k', \HG_k')$ works, but we defer the details to \cref{app:hd-hom-are-skewed} as keeping these (rather long) proofs in the main part would distract from the paper's main result. 

\section{k-Labeled Incidence Graphs}%
\label{sec:k-labeled-incidence-graphs}

Our goal is to give an inductive characterisation of the incidence graphs of strict hypertree depth at most $k$ (and thus also of hypergraphs of strict hypertree depth at most $k$).
The concepts presented in this section were first defined in~\cite{Scheidt2023}, and we adopt their notation and phrasing for the most part.
Note that the $k$-labeled incidence graphs defined here are inspired by the concept of $k$-labeled graphs as they are used in~\cite{Courcelle1993,Dvorak2010,Fluck2024,Lovasz2009} and elsewhere.
In particular, $k$-labeled graphs are the main tool used by Dvo\v{r}ák~\cite{Dvorak2010} to prove his result.
In principle, a \emph{$k$-labeled incidence graph} is an incidence graphs that has \emph{labels} attached to some of its red and blue vertices.
We have an unbounded number of red labels that can be attached to red vertices (though the number of labels actually used must always be finite), but we only have $k$ labels that we can attach to blue vertices.
We are allowed to attach multiple labels to the same vertex, but we are not required to use all of them.
Every red label has an assigned \enquote{guard}, which is a blue vertex with a label on it.
In practice, we will require every red labeled vertex to be a neighbour of its guard (i.e., we want it to have a \emph{real} guard, as defined in the next paragraphs), though it makes the proofs easier if we do not enforce this in the definition itself.
This idea is formalised as follows.

A $k$-labeled incidence graph is a tuple $\kLI = (\I, r, b, g)$, where $\I$ is an incidence graph, $r\colon \natpos \pto \R(\I)$, $b\colon [k] \pto \B(\I)$ and $g\colon \natpos \pto [k]$ are partial mappings such that $\dom(r)$ is finite and $\dom(g) = \dom(r)$.
We use $\I_\kLI$, $r_\kLI$, \dots to denote the components of $\kLI$.
But to keep the indices from getting overly complicated, we may write $\I', r', \dots$ and $\I_i, r_i, \dots$ instead of $\I_{\kLI'}, r_{\kLI'}$ and $\I_{\kLI_i}, r_{\kLI_i}, \dots$, respectively.
If it is clear from the context, we may omit the index altogether and simply write $\I, r, b, g$.

We say that $\kLI$ \emph{has real guards (w.r.t.\ $g$)}, if for every $i \in \dom(r)$ we have $g(i) \in \dom(b)$ and $(b(g(i)), r(i)) \in \E(\I)$.
A $k$-labeled incidence graph $\kLI$ is \emph{label-free} if $r_\kLI = b_{\kLI} = g_{\kLI} = \emptyset$.
We call $\I$ the \emph{skeleton} of $\kLI$.
Next, we define some operations on $k$-labeled incidence graphs.
For completeness, mathematically precise definitions can be found in \cref{app:k-labeled}.

For any set $\XR \subseteq \natpos$ of finite size $\ell$ and any tuple $\tupel{v} = (v_1, v_2, \dots, v_\ell) \in {\R(\I_{\kLI})}^n$ we write $\changeR{\kLI}{\XR}{\tupel{v}}$ to denote a copy of $\kLI$ where we modified $r$ such that $r(i_j) = v_j$ for all $j$ along the enumeration $\enum{ i_1, \dots, i_\ell }$ of $\XR$, i.e., we introduce, and change the placement of, some red labels.
Similarly, for any $\XB = \set{ i_1, \dots, i_\ell } \subseteq [k]$ and any $\tupel{e} = (e_1, e_2, \dots, e_\ell) \in {\B(\I_{\kLI})}^{\ell}$ we write $\changeB{\kLI}{\XB}{\tupel{e}}$ to denote a copy of $\kLI$ where we modified $b$ accordingly.
We write $\removeR{\kLI}{\XR}$ ($\removeB{\kLI}{\XB}$) to denote a copy of $\kLI$ where we \emph{removed} the red (blue) labels in $\XR$ ($\XB$).
Note, that we remove \emph{just} the labels and \emph{not} the vertices that carry them.
\Cref{def:precise-def:changing-labels} in \cref{app:k-labeled} gives mathematically precise definitions.

Intuitively, the \enquote{product} $(\kLI_1 \cdot \kLI_2)$ or \emph{glueing} of two $k$-labeled incidence graphs $\kLI_1$, $\kLI_2$ is the $k$-labeled incidence graph $\kLI$ that is created by first taking the disjoint union of $\kLI_1$ and $\kLI_2$, followed by repeatedly merging pairs of red (blue) vertices, that carry a shared red (blue) label.
By merging we mean that we replace these vertices by a single fresh vertex, which inherits their neighbourhoods and labels.
We apply this procedure until there are no more such pairs.
The guard function of $(\kLI_1 \cdot \kLI_2)$ is simply $g_{\kLI_1} \union g_{\kLI_2}$, i.e., in theory, $g_{\kLI_1}$ has precedence over $g_{\kLI_2}$.
In practice, we will require that $g_{\kLI_1}$ and $g_{\kLI_2}$ are compatible, which means the precedence of $g_{\kLI_1}$ will be irrelevant.
Note that the order in which we merge vertices does not matter, and that if a vertex carries two or more labels, all vertices carrying any one of these labels will be replaced by a single fresh vertex that carries \emph{all} those labels and inherits \emph{all} neighbourhoods.
Finally, for $i \in [2]$ we define mappings $\RvertexMap{\kLI_i}\colon \R(\I_{\kLI_i}) \to \R(\I_{\kLI})$ and $\BvertexMap{\kLI_i}\colon \B(\I_{\kLI_i}) \to \B(\I_{\kLI})$ such that $\RvertexMap{\kLI_i}(v)$ is the red vertex of	$\I_{\kLI}$ that corresponds to $v \in \R(\I_i)$, and $\BvertexMap{\kLI_i}(e)$ is the blue vertex of $\I_{\kLI}$ that corresponds to $e \in \B(\I_i)$.
\Cref{def:precise-def:glueop} in \cref{app:k-labeled} gives mathematically precise definitions.

\begin{example}\label{example:glueing}
	Consider the $k$-labeled incidence graphs $\kLI_1$, $\kLI_2$ according to \cref{fig:example:kli1,fig:example:kli2}.
	In particular, we have
	\begin{align*}&
		\begin{array}[pos]{r|cccc}
			i				& 1 & 2 & 3 & 5 \\\hline
			r_1(i)	& u & w & v & w \\
			r_2(i)	& u & \text{--} & v & w
		\end{array}&&
		\begin{array}[pos]{r|ccc}
			i				& 1 & 2 & 3 \\\hline
			b_1(i)	& f & g & h \\
			b_2(i)	& f & g & h 
		\end{array} &&\text{and} &&
		\begin{array}[pos]{r|cccc}
			i				& 1 & 2 & 3 & 5 \\\hline
			g_1(i)	& 2 & 1 & 1 & 1 \\
			g_2(i)	& 2 & \text{--} & 1 & 1
		\end{array}\quad.
	\end{align*}
	The product $(\kLI_1 \cdot \kLI_2)$ is depicted in \cref{fig:example:glueing}.
\end{example}

So far, we should not be allowed to remove a blue label from a vertex, if it serves as the guard of a red label.
But sometimes we want to \emph{transition} from one (real) guard assignment to another (real) guard assignment.
I.e., we want to remove blue labels even if they still guard some red labels, because we guarantee that we introduce new guards for these labels right away.
We formalise this operation as a special partial function, that assigns new guards to existing red labels: We call $f\colon \natpos \pto [k]$ a \emph{transition for $\kLI$ (for $g$)}, if $\emptyset \neq \dom(f) \subseteq \dom(g)$ and for all $i \in \dom(g)$ we have that if $g(i) \in \img(f)$, then $i \in \dom(f)$.
This means that if $f$ reassigns the blue label guarding the red label $i$, then $f$ provides a new guard for $i$.
Applying a transition, denoted by $\transition{\kLI}{f}$, means modifying a copy of $\kLI$ as follows: we want to insert fresh vertices with these blue labels, thus we must first remove all blue labels, that are currently in use, i.e., we must first remove the labels in the set $\XB \isdef \img(f) \intersect \dom(b_{\kLI}) \intersect \img(g_{\kLI})$ from $b$.
Notice that we have to intersect with $\dom(b_{\kLI})$ since we do not require $\kLI$ to have real guards.
After removing the labels in $\XB$, we insert $|\XB|$ new blue vertices into $\I_{\kLI}$, each carrying one of the blue labels in $\XB$, and introduce an edge between $b(f(i))$ and $r(i)$ for all $i \in \dom(f)$.
Finally, we redefine the guard function as $f \union g_{\kLI}$.

Note that this procedure can be easily expressed as the product $(\Mf \cdot \removeB{\kLI}{\XB})$ for a suitably defined $k$-labeled incidence graph $\Mf$.
\Cref{def:precise-def:transition} in \cref{app:k-labeled} gives a mathematically precise definition.

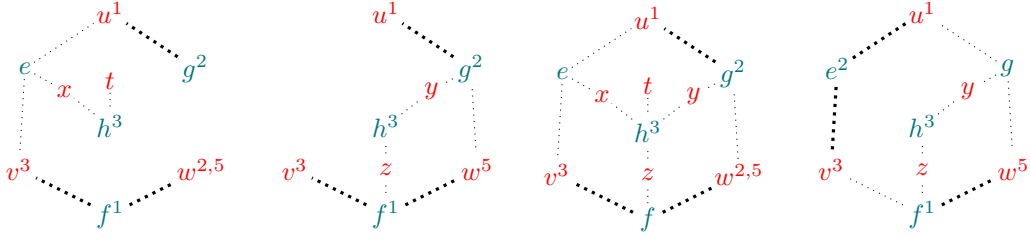
\begin{figure}[t]
	\centering
	\begin{subfigure}{0.24\textwidth}
		\centering

\begin{tikzpicture}[
	every node/.style={regular},
]
	\tikzmath{
		\dist=\exampled/4;
	}

	\node[red] (a) at (0,0) {$v^3$};
	\node[red] (b) at (\exampled,0) {$w^{2,5}$};
	\node[red] (c) at ($(f)+(0,\exampleh)$) {$u^1$};
	\node[red] (d) at ($(a)!.5!(c)$) {$x$};
	\coordinate (m) at ($(d)!.5!(e)$);

	\node[blue] (i) at ($(d)+(150:\dist)$) {$e$};
	\node[blue] (j) at ($(f)+(-90:\dist)$) {$f^1$};
	\node[blue] (k) at ($(e)+(30:\dist)$) {$g^2$};
	\node[blue] (l) at ($(m)!.45!(f)$) {$h^3$};

	\node[red] (g) at ($(l)+(90:{\dist})$) {$t$};

	\graph[edges={edge}]{
		(i) -- {
			(a), (d), (c)
		},
		(k) --[guard] (c),
		(j) --[guard] {
			(a), (b)
		},
		(l) -- {
			(d), (g)
		}
	};
\end{tikzpicture}
 		\caption{The $3$-labeled incidence graph $\kLI_1$.}%
		\label{fig:example:kli1}
	\end{subfigure}
	\hfill
	\begin{subfigure}{0.24\textwidth}
		\centering

\begin{tikzpicture}[
	every node/.style={regular},
]
	\tikzmath{
		\dist=\exampled/4;
	}

	\node[red] (a) at (0,0) {$v^3$};
	\node[red] (b) at (\exampled,0) {$w^5$};
	\node[red] (f) at ($(a)!.5!(b)$) {$z$};
	\node[red] (c) at ($(f)+(0,\exampleh)$) {$u^1$};
	\node[red] (e) at ($(b)!.5!(c)$) {$y$};
	\coordinate (m) at ($(d)!.5!(e)$);

	\node[blue] (j) at ($(f)+(-90:\dist)$) {$f^1$};
	\node[blue] (k) at ($(e)+(30:\dist)$) {$g^2$};
	\node[blue] (l) at ($(m)!.45!(f)$) {$h^3$};

	\graph[edges={edge}]{
		(k) -- {
			(e), (b)
		},
		(k) --[guard] (c),
		(j) --[guard] {
			(a), (b)
		},
		(j) -- (f),
		(l) -- {
			(e), (f)
		}
	};
\end{tikzpicture}
 		\caption{The $3$-labeled incidence graph $\kLI_2$.}%
		\label{fig:example:kli2}
	\end{subfigure}
	\begin{subfigure}{0.24\textwidth}
		\centering

\begin{tikzpicture}[
	every node/.style={regular},
]
	\tikzmath{
		\dist=\exampled/4;
	}

	\node[red] (a) at (0,0) {$v^3$};
	\node[red,rectangle] (b) at (\exampled,0) {$w^{2,5}$};
	\node[red] (f) at ($(a)!.5!(b)$) {$z$};
	\node[red] (c) at ($(f)+(0,\exampleh)$) {$u^1$};
	\node[red] (d) at ($(a)!.5!(c)$) {$x$};
	\node[red] (e) at ($(b)!.5!(c)$) {$y$};
	\coordinate (m) at ($(d)!.5!(e)$);

	\node[blue] (i) at ($(d)+(150:\dist)$) {$e$};
	\node[blue] (j) at ($(f)+(-90:\dist)$) {$f$};
	\node[blue] (k) at ($(e)+(30:\dist)$) {$g^2$};
	\node[blue] (l) at ($(m)!.45!(f)$) {$h^3$};

	\node[red] (g) at ($(l)+(90:{\dist})$) {$t$};

	\graph[edges={edge}]{
		(i) -- {
			(a), (d), (c)
		},
		(k) -- {
			(e), (b)
		},
		(k) --[guard] (c),
		(j) --[guard] {
			(a), (b)
		},
		(j) -- (f);
		(l) -- {
			(d), (e), (f), (g)
		}
	};
\end{tikzpicture}
 		\caption{An example of glueing: $(\kLI_1 \cdot \kLI_2)$.}%
		\label{fig:example:glueing}
	\end{subfigure}
	\hfill
	\begin{subfigure}{0.24\textwidth}
		\centering

\begin{tikzpicture}[
	every node/.style={regular},
]
	\tikzmath{
		\dist=\exampled/4;
	}

	\node[red] (a) at (0,0) {$v^3$};
	\node[red] (b) at (\exampled,0) {$w^5$};
	\node[red] (f) at ($(a)!.5!(b)$) {$z$};
	\node[red] (c) at ($(f)+(0,\exampleh)$) {$u^1$};
	\node[red] (e) at ($(b)!.5!(c)$) {$y$};
	\coordinate (m) at ($(d)!.5!(e)$);

	\node[blue] (i) at ($(d)+(150:\dist)$) {$e^2$};
	\node[blue] (j) at ($(f)+(-90:\dist)$) {$f^1$};
	\node[blue] (k) at ($(e)+(30:\dist)$) {$g$};
	\node[blue] (l) at ($(m)!.45!(f)$) {$h^3$};

	\graph[edges={edge}]{
		(i) --[guard] {
			(a), (c)
		},
		(k) -- {
			(e), (b), (c)
		},
		(j) -- {
			(a), (f)
		},
		(j) --[guard] (b),
		(l) -- {
			(e), (f)
		}
	};
\end{tikzpicture}
 		\caption{An example of applying a transition: $\transition{\kLI_1}{f}$.}%
		\label{fig:example:transition}
	\end{subfigure}
	\caption{$3$-labeled incidence graphs and operations on them.
	Labels are encoded as exponents and the guard function is encoded using thicker edges between the red vertex and its guard.}%
	\label{fig:example:k-labeled-ig-ops}
\end{figure}

\begin{example}
	Consider the $k$-labeled incidence graph from \cref{example:glueing} and \cref{fig:example:kli1}.
	The partial function $f = \set{ 1 \to 2, 3 \to 2 }$ is a transition for $\kLI_1$.
	The result $\transition{\kLI_1}{f}$ of the application of $f$ on $\kLI_1$ is depicted in \cref{fig:example:transition}.
\end{example}

We define the class $\GLI_k^i$ of $k$-labeled incidence graphs that can be constructed in a way that at most $i$ blue labels are removed \enquote{in series}.
\begin{definition}\label{def:k-label-class}
	For $k \in \natpos$ and $i \in \nat$ we define the set $\GLI_k^i$ inductively as follows.
	\begin{description}
		\item[Base case:]\label{def:k-label-class:base}
		$\kLI \in \GLI_k^0$ for all $k$-labeled incidence graphs $\kLI$ with $\dom(r) = \R(\I)$, $\dom(b) = \B(\I)$ and real guards.

		For all $i \in \nat$, if $\kLI \in \GLI_k^i$, then\, $\kLI \in \GLI_k^{i+1}$.
		\smallskip
		
		\item[Glueing]\label{def:k-label-class:glue}
		Let $\kLI_1 \in \GLI_k^{i_1},\ \kLI_2 \in \GLI_k^{i_2}$ have compatible guard functions and $\kLI = (\kLI_1 \cdot \kLI_2)$.\\
		Then,\; $\kLI \in \GLI_k^i$\; where\; $i \isdef \max \set{ i_1, i_2 }$.
		\smallskip
		
		\item[Transitioning]\label{def:k-label-class:transition}
		Let $\kLI \in \GLI_k^i$, let $f$ be a transition for $\kLI$ and $\kLI' = \transition{\kLI_1}{f}$.\\
		Then,\; $\kLI' \in \GLI_k^{i'}$\; where\; $i' \isdef i + |\img(f) \intersect \img(b_{\kLI}) \intersect \img(g_{\kLI})|$.
		\smallskip

		\item[Label-Removal]\label{def:k-label-class:label-remove}
		Let $\kLI \in \GLI_k^{i}$.
		\begin{enumerate}[(a)]
			\item\label{def:k-label-class:red-remove}
			For $\XR \subseteq \dom(r)$,\; $\removeR{\kLI}{\XR} \in \GLI_k^{i}$.
			\smallskip
	
			\item\label{def:k-label-class:blue-remove}
			For $\XB \subseteq \dom(b) \setminus \img(g)$,\; $\removeB{\kLI}{\XB} \in \GLI_k^{i'}$\; where\; $i' \isdef i + |\XB|$.
		\end{enumerate}
	\end{description}
	\smallskip
	Finally, we let $\GLIDk \isdef \GLI_k^k$ for every $k \in \nat$.
\end{definition} 

\section{Characterising Hypergraphs of Strict Hypertree Depth at most k}%
\label{sec:characterising-shd}

In this section we prove that the inductively defined class $\GLIDk$ corresponds precisely to the class $\ISHDk$.
\begin{restatable}{theorem}{inductivedefhd}\label{thm:inductive_def_hypertree_depth}
	An incidence graph $\J$ has strict hypertree depth at most $k$ if, and only if, there exists a label-free $\kLI \in \GLIDk$ such that $\I_\kLI \isomorphic \J$.
\end{restatable}

In the following, we first show how to construct an incidence graph of strict hypertree depth at most $k$ as the skeleton of a label-free $k$-labeled incidence graph in $\GLIDk$ (\cref{lem:bounded-shd-in-glik}).
Then we show that every label-free $\kLI \in \GLIDk$ has strict hypertree depth at most $k$ (\cref{lem:glik-in-bounded-shd}).
\cref{thm:inductive_def_hypertree_depth} follows directly from the combination of these two Lemmata.

For the rest of this section, let $\J \in \ISHDk$ and let $(\T, \Gamma)$ be a strict elimination forest of height $\leq k$ for $\J$.
We can w.l.o.g.\ assume that $\J$ is connected and that $\T$ is a tree (\cref{lem:restrict_to_connected}).
Let $\R(\I) = \set{ v_1, v_2, \dots, v_m }$ where $m \geq 1$.

$(\T, \Gamma)$ will help us decide when to remove (i.e., eliminate) which label from which blue vertex in the following sense.
The core idea is to start with a trivial $k$-labeled incidence graph for every path from a leaf to the root in the elimination tree.
Then we walk bottom-up along these paths and whenever several paths join in a node, we apply red and blue vertex removals in a suitable way on their $k$-labeled incidence graphs, such that afterwards we can glue them together and receive a $k$-labeled incidence graph that is isomorphic to the incidence graph induced by the union of said paths.
This idea is illustrated in \cref{app:example-construction} in \cref{app:characterising-hd}.

For this we need the following notions: 
for a $k$-labeled incidence graph $\kLI$, we say that a red vertex $v \in \R(\I_{\kLI})$ is \emph{unlabeled}, if $v \not\in \img(r_{\kLI})$.
Analogously, a blue vertex $e \in \B(\I_{\kLI})$ is \emph{unlabeled}, if $e \not\in \img(b_{\kLI})$.
For a node $n$ in a tree $\T$, the \emph{subtree with stem induced by $n$} is the tree $\Ts_n$ induced on $\T$ by the set $\p(n) \union \set{ t \in \V(\T) \mid n \leq_\T t }$.
Recall that $\p(n)$ is the set of nodes on the path from $n$ to the root $\troot$ (including $n$ and $\troot$), and notice that the subtree with stem induced by the root is $\T$, i.e., $\T_\troot = \T$, and for every leaf $n$ it is the path from the root $\troot$ to $n$.
For every node $n \in \V(\T)$ we define the set $\labels(n) \isdef \set{ i \in [m] \mid v_i \in \hGamma(n) }$.
To avoid an overload of notation, we will write $\labels(N)$ to denote the set $\bigunion_{n \in N} \labels(n)$ and write $\induced{\J}{\Ts_n}$ to abbreviate $\induced{\J}{\hGamma(\V(\Ts_n))}$.
\begin{restatable}{lemma}{boundedhdinglik}\label{lem:bounded-shd-in-glik}
	For every $n \in \V(\T)$ of level $d$ where $\enum{ t_1, \dots, t_d }$ is a $\leq_\T$-enumeration of $\p(n)$ (i.e., in particular $t_1 = \troot$, $t_d = n$), there exists an $\kLI_n \in \GLI_k^{k-d}$ of the form $(\I, r, b, g)$ such that
	\begin{enumerate}[(A)]
		\crefalias{enumi}{condition}
		\item\label{cond:bounded-shd-in-glik:labels}
		$\dom(b) = [d]$\, and \
		$\dom(r) = \labels(\p(n))$;
		\item\label{cond:bounded-shd-in-glik:largest-guard}
		$g(i) \isdef \min \set{ j \in [d] \mid v_i \in \hGamma(t_j) }$ for every $i \in \dom(g)$;
		\item\label{cond:bounded-shd-in-glik:isomorphism}
		There exists an isomorphism $(\isoR, \isoB)$ between $\I$ and $\induced{\J}{\Ts_n}$ such that
		\begin{enumerate}[(i)]
			\item $\isoR(r(i)) = v_i$ for all $i \in \dom(r)$,\; and
			\item $\isoB(b(j)) = \Gamma(t_j)$ for all $j \in \dom(b)$.
		\end{enumerate}
	\end{enumerate}
	
\end{restatable}
Notice that, in particular, this lemma states $\J \isomorphic \I_{\kLI_{\troot}}$ for the root $\troot$ of $\T$.
But $\kLI_{\troot}$ is not label-free, since $\dom(b_{\kLI_\troot}) = \set{ 1 }$ and $\dom(r_{\kLI_{\troot}}) = \labels(\p(\troot)) = \labels(\troot)$.
But, the lemma also states that $\kLI_{\troot} \in \GLI_k^{d-1}$, since $\level(\troot) = 1$.
Thus, $\J \isomorphic \I_{\kLI'}$ for $\kLI' = \removeB{\removeR{\kLI_{\troot}}{\labels(\troot)}}{\set{ 1 }}$, and in particular, $\kLI' \in \GLI_k^k$ is label-free.
Thus, this lemma shows the forward direction of \cref{thm:inductive_def_hypertree_depth}.

\begin{proof}[Proof of \cref{lem:bounded-shd-in-glik}]
We prove this claim by (inverse) induction over the level of the node $n \in \V(\T)$.
\vspace*{-\abovedisplayskip}
\proofsubparagraph*{Base Case.}
Let  $n$ be a leaf, then $\Ts_n$ is simply the path from $n$ to the root $\troot$.

We let $\kLI_n = (\induced{\J}{\Ts_n}, r, b, g)$ where $r(i) \isdef v_i$ and $g(i) \isdef \min \set{j \in [k] \mid v_i \in \hGamma(t_j)}$ for all $i \in [m]$ and $b(j) \isdef t_j$ for all $j \in [d]$.
$\kLI_n \in \GLI_k^0$ according to the base case of \cref{def:k-label-class}, and thus also $\kLI_n \in \GLI_k^{k-d}$ and \cref{cond:bounded-shd-in-glik:isomorphism,cond:bounded-shd-in-glik:largest-guard,cond:bounded-shd-in-glik:labels} are easily verified.
\vspace*{-\abovedisplayskip}
\proofsubparagraph{Inductive Step.}
Let $n \in \V(\T)$ be a node of level $d$ that is not a leaf.
Let $n_1, \dots, n_\ell$ be the children of $n$ with $\ell \geq 1$ and assume that 
for every $i \in [\ell]$ there exist $\kLI_i \in \GLI_k^{k-(d{+}1)}$ and $(\isoR^i, \isoB^i)$ such that \cref{cond:bounded-shd-in-glik:labels,cond:bounded-shd-in-glik:largest-guard,cond:bounded-shd-in-glik:isomorphism} are satisfied.

For every child $n_i$ let $\XR^{i} \isdef \labels(n_i) \setminus \labels(\p(n))$ and let $\kLI_i' \isdef \removeB{\removeR{\kLI_i}{\XR^i}}{\XB}$ with $\XB \isdef \set{ d{+}1 }$.
We let $\kLI_n \isdef \kLI'_1 \cdot \kLI'_2 \cdot \cdots \cdot \kLI'_\ell$ if $\ell \geq 2$, otherwise $\kLI_n \isdef \kLI_1$.
We claim that $\kLI_n \in \GLI_k^d$ and that it satisfies \cref{cond:bounded-shd-in-glik:labels,cond:bounded-shd-in-glik:largest-guard,cond:bounded-shd-in-glik:isomorphism}.

Let $n_i$ be any child of $n$.
Obviously, $\removeR{\kLI_i}{\XR^i} \in \GLI_k^{k-(d+1)}$.
Let $j \in \dom(r_i)$ such that $g_i(j) = d{+}1$.
By induction hypothesis, this means that $n_i$ is the minimal node according to $\leq_\T$ where $v_j$ appears in $\hGamma$.
Thus, $j \in \labels(n_i)$ but $j \not\in \labels(\p(n))$.
Therefore, $j \in \XR^i$.
Thus, $d{+}1 \not\in \img(g_{\removeR{\kLI_i}{\XR^i}})$.
By \cref{def:k-label-class} this means that 
\begin{equation}
	\removeB{\removeR{\kLI_i}{\XR^i}}{\XB} \in \GLI_k^{k-d}.
\end{equation}
By induction hypothesis, we have $\dom(b_i) = [d{+}1]$.
Thus, $\dom(b_{\kLI_i'}) = [d]$.\\
Since $\p(n)$ is a subset of $\p(n_i)$, $\labels(\p(n))$ is also a subset of $\labels(\p(n_i))$, which means $\labels(\p(n)) \subseteq \dom(r_i)$ by induction hypothesis.
By definition of $\XR^i$, $\labels(\p(n))$ is also a subset of $\dom(r_{\kLI'_i})$.
Let $j \in \dom(r_{\kLI'_i})$ be any red label in $\kLI_i'$.
Since $j \not\in \XR^i$, $j \not\in \labels(n_i)$ or $j \in \labels(\p(n))$.
In both cases, this implies $j \in \labels(\p(n))$.

Since all this holds for every child and their guard functions are compatible according to the induction hypothesis, it is easy to see that $\kLI_n \in \GLI_k^{k-d}$ and that it satisfies \cref{cond:bounded-shd-in-glik:labels}, whereas \cref{cond:bounded-shd-in-glik:largest-guard} follows directly from the induction hypothesis.

It remains to verify \cref{cond:bounded-shd-in-glik:isomorphism} for $n$.
For the following, keep in mind that $\I_{\kLI_i'} = \I_i$ for all $i \in [\ell]$, since the operations we applied to create $\kLI_i'$ only changed the functions $r_i$, $b_i$, $g_i$.

For every unlabeled vertex $e \in \B(\I)$ or $v \in \R(\I)$ we know that there is a unique child $n_{i_e}$ or $n_{i_v}$ of $n$ and a unique $e' \in \B(\I_{i_e})$ or $v' \in \R(\I_{i_v})$ such that $\BvertexMap{\kLI_{i_e}'}(e') = e$ or $\RvertexMap{\kLI_{i_v}'}(v') = v$ and the other way around.

Let $\isoR$ and $\isoB$ be defined by 
\begin{alignat*}{4}
	\isoB(b(j)) &= \Gamma(t_j)\enspace &&\all \ j \in \dom(b),\quad &\isoB(e) &= \isoB^{i_e}(e') \enspace &&\all \ e \not\in \img(b) \quad\text{and} \\
	\isoR(r(j)) &= v_j \ &&\all \ j \in \dom(r),\; &\isoR(v) &= \isoR^{i_v}(v') \ &&\all \ v \not\in \img(r).
\end{alignat*}
\Cref{cond:bounded-shd-in-glik:isomorphism} is satisfied if $(\isoR, \isoB)$ is indeed an isomorphism.
We first show that $\isoR$ and $\isoB$ are bijective.
They clearly are total mappings, thus we have to show that they are surjective and injective.
Let $v_i \in \R(\induced{\J}{\Ts_n})$ be a vertex.
Then $v_i \in \hGamma(\V(\Ts_n))$.
\smallskip

\noindent\emph{Case 1}: If $v_i \in \R(\induced{\J}{\Gamma(\p(n))})$, then by definition, $i \in \labels(\p(n))$.
I.e., $\isoR(r(i)) = v_i$.
Assume there exists another $v' \in \R(\I)$ such that $r(i) \neq v'$ but $\isoR(v') = v_i$.
Then there must be a unique child $n_{j}$ of $n$ such that $\isoR(v') = \isoR^{j}(v') = v_i$.
Since $v_i \in \labels(\p(n))$, $v_i$ is also in $\labels(\p(n_j))$.
Thus, by induction hypothesis, $v' \in \img(r_j)$ and $r_j(i) = v'$.
Since $v' \not\in \img(r)$, the label must have been removed, i.e., $i \in \XR^j$.
This means $i \in \labels(n_j) \setminus \labels(\p(n))$, which contradicts the assumption that $i \in \labels(\p(n))$.
Thus, $v'$ can not exist.
\smallskip

\noindent\emph{Case 2}: Otherwise, $v_i \in \R(\induced{\J}{\Ts_{n_j}})$ for a unique child $n_j$ of $n$.
If there were more than one child containing $v_i$, this would imply the first case, i.e., $v_i \in \R(\induced{\J}{\Gamma(\p(n))})$,  because of \cref{def:hd:shared-heritage} of \cref{def:hd}.
By induction hypothesis, there exists a $v' \in \R(\I_j)$ such that $\isoR^{j}(v') = v_i$.
Since $i \not\in \labels(\p(n))$, either $v' \not\in \img(r_j)$ or $i \in \XR^j$.
In both cases, $v' \not\in \img(r_{\kLI_j'})$.
Thus, for $v \in \R(\I)$ such that $\RvertexMap{\kLI_j'}(v') = v$ we get $\isoR(v) = v_i$.
Since $i \not\in \labels(\p(n))$, this $v$ is unique.
Hence, in total, $\isoR$ is bijective.

Analogously, it is easy to see that $\isoB$ is bijective since every $e \in \B(\induced{\J}{\Ts_n})$ is either in the stem or in $\B(\induced{\J}{\Ts_{n_i}})$ for precisely one child $n_i$ of $n$.
It remains to show that $(e,v) \in \E(\I)$ iff $(\isoB(e), \isoR(v)) \in \E(\induced{\J}{\Ts_n})$.

$(e,v) \in \E(\I)$ holds, iff there is a (not necessarily unique) $j \in [\ell]$ and $e' \in \B(\I_j)$, $v' \in \R(\I_j)$ such that $\BvertexMap{\kLI_j'}(e') = e$, $\RvertexMap{\kLI_j'}(v') = v$ and $(e', v') \in \E(\I_j)$.
By induction hypothesis, this is the case iff $(\isoB^j(e'), \isoR^j(v')) \in \E(\induced{\J}{\Ts_{n_j}})$.
It is easy to verify that, $(\isoB^j(e'), \isoR^j(v')) \in \E(\induced{\J}{\Ts_{n_j}})$ if, and only if, $(\isoB(\BvertexMap{\kLI_j}(e')), \isoR(\RvertexMap{\kLI_j}(v'))) \in \E(\induced{\J}{\Ts_n})$.
Thus, $(e,v) \in \E(\I) \iff (\isoB(e), \isoR(v)) \in \E(\induced{\J}{\Ts_n})$.
 \end{proof}

The following lemma can be shown by induction.
On a high level, the idea of the proof is to only modify the elimination forest, if blue labels are removed.
At that point, we prepend a chain of new nodes to the root(s) of the elimination forest.
If we take the product of two $k$-labeled incidence graphs, we take the union of the forests, and if we remove red labels, we do not alter the forest at all.%
\begin{restatable}{lemma}{glikinboundedhd}\label{lem:glik-in-bounded-shd}
	For every $\kLI \in \GLI_k^d$ of the form $(\I, r, b, g)$ there is a tuple $(\F, \Gamma)$, where $\F$ is a forest of height $\leq d$ and $\Gamma$ is a bijective function from $\V(\F)$ to $\B(\I) \setminus \img(b)$ satisfying \cref{cond:glik-in-bounded-shd:unlabeled-shared-heritage}.
	We write $\hGamma(t)$ as a shorthand for $\blueN(\Gamma(t))$ and $\ulGamma(t)$ as a shorthand for $\hGamma(t) \setminus \img(r)$.
	\begin{enumerate}[(A)]
		\crefalias{enumi}{condition}
		\item\label{cond:glik-in-bounded-shd:unlabeled-shared-heritage}
		For all $s,t \in \V(\F)$, and all $v \in \ulGamma(s) \intersect \ulGamma(t) \neq \emptyset$ it holds that:\\
		$v \in \blueN(b(j))$ for a $j \in \dom(b)$ or $\lcv(s,t)$ is defined and $v \in \bigunion \ulGamma(\p(\lcv(s,t)))$.
	\end{enumerate}
	
\end{restatable}
\noindent Notice that, if $\kLI \in \GLI_k^d$ is label-free, this guarantees a strict elimination forest $(\F, \Gamma)$ of height $d$ for $\I_\kLI$.
This shows the backward direction of \cref{thm:inductive_def_hypertree_depth}.

\begin{proof}[Proof of \cref{lem:glik-in-bounded-shd}]
We prove this by induction over the definition of $\GLI_k^d$.
\begin{description}
	\item[$\vartriangleright$\ Base Case]
	Let $\kLI \in \GLI_k^0$.
	It is easy to see that the empty forest together with the empty map proves the claim for $\kLI$.
	\medskip
	\item[$\vartriangleright$\ Inductive Step]
	Let $\kLI_1 \in \GLI_k^{d_1}$, $\kLI_2 \in \GLI_k^{d_2}$.
	For convenience, let $\kLI_i = (\I_i, r_i, b_i, g_i)$.
	\smallskip
	\item[Induction Hypothesis] 
	There exist $(\F_i, \Gamma_i)$ proving the lemma for $\kLI_i$ with $i \in [2]$.
	We assume for convenience, that $\F_1$ and $\F_2$ are disjoint.	
	\smallskip
	\item[Induction Claim] 
	There exists a tuple $(\F, \Gamma)$ proving the lemma for $\kLI = (\I, r, b, g)$ constructed according to the following cases.
\end{description}
\vspace*{-\belowdisplayskip}
\smallskip

\vspace*{-\abovedisplayskip}
\proofsubparagraph*{Case Glueing}
If $\kLI = (\kLI_1 \cdot \kLI_2)$, we let $\F$ be the disjoint union of $\F_1$, $\F_2$.
For all $t \in \V(\F_1)$, let $\Gamma(t) \isdef \BvertexMap{\kLI_1}(\Gamma_1(t))$ and for all $t \in \V(\F_2)$, let $\Gamma(t) \isdef \BvertexMap{\kLI_2}(\Gamma_2(t))$.
It is easy to see that $\F$ has height at most $d = \max\set{ d_1, d_2 }$ and that $\Gamma$ is a bijective map from $\V(\F)$ to $\B(\I) \setminus \img(b)$.

It remains to \emph{verify \cref{cond:glik-in-bounded-shd:unlabeled-shared-heritage}}:
Let $s, t \in \V(\F)$, let $v \in \ulGamma(s) \intersect \ulGamma(t)$ and let $e_s = \Gamma(s)$ and $e_t = \Gamma(t)$.
If $v \in \blueN(b(j))$ for a $j \in \dom(b)$, there is nothing to show.
Thus, assume that all blue neighbours of $v$ are unlabeled.
Since $v \not\in \img(r)$, either $\RvertexMap{\kLI_1}(v') = v$ holds for a $v' \in \R(\I_1)$, or $\RvertexMap{\kLI_2}(v') = v$ holds for a $v' \in \R(\I_2)$.
We assume that the first is true, i.e., $\RvertexMap{\kLI_1}(v') = v$.
The second case can be handled analogously.

All unlabeled blue vertices $e \in \B(\I) \setminus \img(b)$ that are neighbours of $v$ must have come from $\kLI_1$, because since they are unlabeled, the only way that $(e, v) \in \E(\I)$ can hold is that $(e', v') \in \E(\I_1)$ holds for $e'$ such that $\BvertexMap{\kLI_1}(e') = e$.
Thus, there exist $e_s', e_t' \in \B(\I_1)$ such that $\BvertexMap{\kLI_1}(e_s') = e_s$ and $\BvertexMap{\kLI_1}(e_t') = e_t$.
By definition of $\Gamma$ we can infer that $\Gamma_1(s) = e_s'$ and $\Gamma_1(t) = e_t'$.
I.e., $s, t \in \V(\F_1)$ and because of that, according to the induction hypothesis, $\lcv(s,t)$ is defined and $v' \in \bigunion \ulGamma_1(\p(\lcv(s,t)))$.
Let $p \in \p(\lcv(s,t))$ such that $v' \in \ulGamma(p)$.
This means $(e', v') \in \E(\I_1)$ for $e' = \Gamma_1(p)$ and thus also $(\BvertexMap{\kLI_1}(e'), \RvertexMap{\kLI_1}(v')) \in \E(\I)$.
Since by definition, $\Gamma(p) = \BvertexMap{\kLI_1}(e')$, we get that $v \in \bigunion \ulGamma(\p(\lcv(s,t)))$.

\vspace*{-.5\abovedisplayskip}
\proofsubparagraph*{Case Transitioning}
If $\kLI = \transition{\kLI_1}{f}$ where $f$ is a transition for $\kLI_1$, consider the enumeration $\enum{ j_1, j_2, \dots, j_\ell }$ of $\img(f) \intersect \img(b_1) \intersect \img(g_1)$.
Let $\F$ be the tree constructed from $\F_1$, when we insert a chain $t_1, \dots, t_\ell$ into $\F_1$, connect all roots in $\troots_{\F_1}$ with $t_{\ell}$ and let $\troot_\F = t_1$.
I.e., let $\F = (\V(\F), \E(\F))$ with $\V(\F) \isdef \V(\F_1) \disunion \set{ t_1, t_2, \dots, t_\ell }$, $\troot_\F \isdef t_1$ and
\[
	\E(\F) \isdef \E(\F_1) \;\union\; \set{ \set{ t_i, t_{i{+}1} } \mid i \in [\ell{-}1] } \;\union\; \set{ \set{ \troot, t_\ell } \mid \troot \in \troots_{\F_1} }.
\]
Let $\Gamma(t_i) \isdef \BvertexMap{\kLI_1}(b_1(j_i))$ for all $i \in [\ell]$ and let $\Gamma(t) \isdef \BvertexMap{\kLI_1}(\Gamma_1(t))$ for all other $t \in \V(\F)$ (i.e., for $t \in \V(\F_1)$).
It is easy to see that $\F$ is a forest of height $d = d_1 + |\img(f) \intersect \img(b_1) \intersect \img(g_1)|$.

\smallskip
\noindent\emph{Verifying \cref{cond:glik-in-bounded-shd:unlabeled-shared-heritage}}:
Let $s,t \in \V(\F)$, let $v \in \ulGamma(s) \intersect \ulGamma(t)$.
Let $v' \in \R(\I_1)$ such that $\RvertexMap{\kLI_1}(v') = v$.
If $s = t_i$ for $i \in [\ell]$, then w.l.o.g.\ $\lcv(s,t) = s$ and \cref{cond:glik-in-bounded-shd:unlabeled-shared-heritage} is trivially true.
The same holds if $t = t_i$ for $i \in [\ell]$.
Thus, it remains to consider that $s,t \in \V(\F_1)$.
Let $e_s'$, $e_t'$ be the blue vertices such that $\Gamma_1(s) = e_s'$ and $\Gamma_1(t) = e_t'$.
By construction, they are both unlabeled.
If $v \in \blueN(b(j))$ for a $j \in \dom(b)$ there is nothing left to show.
Thus assume that this is not the case.
Because we assumed that $v \in \ulGamma(s) \intersect \ulGamma(t)$, we know that $(\BvertexMap{\kLI_1}(e_s'), v), (\BvertexMap{\kLI_1}(e_t'), v) \in \E(\I)$.
Hence, $(e_s', v'), (e_t', v') \in \E(\I_1)$ and according to the induction hypothesis, $\lcv(s,t)$ is defined and $v' \in \bigunion \ulGamma(\p(\lcv(s,t)))$ unless $v' \in \blueN(b_1(j))$ for a $j \in \dom(b_1)$.
In the first case we get that $v \in \bigunion \ulGamma(\p(\lcv(s,t)))$ by the same reasoning we used for glueing.
If $v' \in b_1(j)$ but $v \not\in b(j')$ for all $j' \in \dom(b)$, then $j \in \img(f) \intersect \img(b_1) \intersect \img(g_1)$.
I.e., $v' \in \blueN(b_1(j_{i}))$ for some $i \in [\ell]$.
But this also means that $v \in \blueN(\Gamma(t_i))$ and since by construction both $s \leq_\F t_i$ and $t \leq_\F t_i$, $\lcv(s,t)$ is defined and $t_i \in \p(\lcv(s,t))$ and thus $v \in \bigunion \ulGamma(\p(\lcv(s,t)))$.

\vspace*{-.5\abovedisplayskip}
\proofsubparagraph*{Case Label-Removal}~\\
\textcolor{lipicsGray}{\textbf{\textsf{(a)}}}\;
If $\kLI = \removeR{\kLI_1}{\XR}$ for a set $\XR \subseteq \dom(r_1)$, let $\F = \F_1$ and let $\Gamma(t) = \BvertexMap{\kLI_1}(\Gamma_1(t))$ for all $t \in \V(\F)$.

\smallskip
\noindent\emph{Verifying \cref{cond:glik-in-bounded-shd:unlabeled-shared-heritage}}:
Let $s,t \in \V(\F)$, let $v \in \ulGamma(s) \intersect \ulGamma(t)$.
Remember that $\RvertexMap{\kLI_1}$ is the identity in this case.
If $v \not\in \img(r_1)$, then it trivially follows from the induction hypothesis, that \cref{cond:glik-in-bounded-shd:unlabeled-shared-heritage} still holds for $v$.
If $v \in \img(r_1)$, then there exists an $i \in \XR$ such that $r_1(i) = v$.
Thus, $v \in \blueN(b_1(g_1(i)))$, since $\kLI_1$ has real guards.
Thus, $v \in \blueN(b(j))$ for $j = g_1(i)$ still holds.

\medskip
\noindent\textcolor{lipicsGray}{\textbf{\textsf{(b)}}}\;
If $\kLI = \removeB{\kLI_1}{\XB}$ for $\XB \subseteq \dom(b_1) \setminus \img(g_1)$, consider the enumeration $\enum{ j_1, j_2, \dots, j_\ell }$ of $\XB$.
Let $\F$ be the tree constructed from $\F_1$, when we insert a chain $t_1, \dots, t_\ell$ into $\F_1$, connect all roots in $\troots_{\F_1}$ with $t_{\ell}$ and let $\troot_\F = t_1$.
I.e., let $\F = (\V(\F), \E(\F))$ with $\V(\F) \isdef \V(\F_1) \disunion \set{ t_1, t_2, \dots, t_\ell }$, $\troot_\F \isdef t_1$ and
\[
	\E(\F) \isdef \E(\F_1) \;\union\; \set{ \set{ t_i, t_{i{+}1} } \mid i \in [\ell{-}1] } \;\union\; \set{ \set{ \troot, t_\ell } \mid \troot \in \troots_{\F_1} }.
\]
Let $\Gamma(t_i) \isdef \BvertexMap{\kLI_1}(b_1(j_i))$ for all $i \in [\ell]$ and let $\Gamma(t) \isdef \BvertexMap{\kLI_1}(\Gamma_1(t))$ for all other $t \in \V(\F)$ (i.e., for $t \in \V(\F_1)$).

\smallskip
\noindent\emph{Verifying \cref{cond:glik-in-bounded-shd:unlabeled-shared-heritage}}:
If $s = t_i$ for $i \in [\ell]$, then w.l.o.g.\ $\lcv(s,t) = s$ and \cref{cond:glik-in-bounded-shd:unlabeled-shared-heritage} is trivially true.
The same holds if $t = t_i$ for $i \in [\ell]$.
Thus, it remains to consider that $s,t \in \V(\F_1)$.
Let $e_s'$, $e_t'$ be the blue vertices such that $\Gamma_1(s) = e_s'$ and $\Gamma_1(t) = e_t'$.
By construction, they are both unlabeled.
If $v \in \blueN(b(j))$ for a $j \in \dom(b)$ there is nothing left to show.
Thus assume that this is not the case.
Because we assumed that $v \in \ulGamma(s) \intersect \ulGamma(t)$, we know that $(\BvertexMap{\kLI_1}(e_s'), v), (\BvertexMap{\kLI_1}(e_t'), v) \in \E(\I)$.
Hence, $(e_s', v'), (e_t', v') \in \E(\I_1)$ and according to the induction hypothesis, $\lcv(s,t)$ is defined and $v' \in \bigunion \ulGamma(\p(\lcv(s,t)))$ unless $v' \in \blueN(b_1(j))$ for a $j \in \dom(b_1)$.
In the first case we get that $v \in \bigunion \ulGamma(\p(\lcv(s,t)))$ by the same reasoning we used for glueing.
If $v' \in b_1(j)$ but $v \not\in b(j')$ for all $j' \in \dom(b)$, then $j \in \XB$.
I.e., $v' \in \blueN(b_1(j_{i}))$ for some $i \in [\ell]$.
But this also means that $v \in \blueN(\Gamma(t_i))$ and since by construction both $s \leq_\F t_i$ and $t \leq_\F t_i$, $\lcv(s,t)$ is defined and $t_i \in \p(\lcv(s,t))$ and thus $v \in \bigunion \ulGamma(\p(\lcv(s,t)))$.
\qedhere \end{proof} 

\section{The Logic \texorpdfstring{GC\textsuperscript{k}}{GCk}}%
\label{sec:gc}

This section introduces the logic $\GCk$ as defined in~\cite{Scheidt2023} and its restricted fragment $\GCDk$, consisting of all formulas of guard depth at most $k$.
Let $k$ be a positive natural number, that is fixed for this section.
\vspace{-\abovedisplayskip}
\subparagraph*{Variables} 
$\GCk$ uses two different kinds of variables: $\VARV \isdef \set{ \varv_1, \varv_2, \varv_3 \dots }$ to address vertices and $\VARE \isdef \set{ \vare_1, \vare_2, \dots, \vare_k }$ to address hyperedges.
Notice that the number of variables for hyperedges is bounded by $k$, but unbounded for vertices.
We say that a tuple of the form $\tupel{\varv} = (\varv_{i_1}, \dots, \varv_{i_\ell}) \in \VARV^{\ell}$ or $\tupel{\vare} = (\vare_{i_1}, \dots, \vare_{i_\ell}) \in \VARE^{\ell}$ is a \emph{$\varv$- or $\vare$-tuple}, if $i_1 < i_2 < \cdots < i_\ell$.
We let $\vset(\tupel{\varv}) \isdef \set{ \varv_{i_1}, \dots, \varv_{i_\ell} }$ and $\vset(\tupel{\vare}) \isdef \set{ \vare_{i_1}, \dots, \vare_{i_\ell} }$ respectively.
We call $\set{ i_1, \dots, i_\ell }$ the \emph{index set} of $\tupel{v}$ and $\tupel{e}$, respectively.
\vspace{-\abovedisplayskip}
\subparagraph*{Logical Guards} 

The key idea behind $\GCk$ is that on quantification, vertex variables must be \emph{guarded} by hyperedge variables.
This is formalised by a partial function $g\colon \natpos \pto [k]$ with finite domain (similar to the guard function of a $k$-labeled incidence graph, cf.~\cref{sec:k-labeled-incidence-graphs}) and its corresponding \emph{logical guard} $\LogGuard{g} \isdef \bigland_{i \in \dom(g)} E(\vare_{g(i)}, \varv_i)$.
For the special partial function $g$ with empty domain, we let $\LogGuard{g} \isdef \top$, which is a special formula that always evaluates to true.

\begin{definition}\label{def:gck}
	The logic $\GCk$ is inductively defined along with the \emph{free vertex variables}, the \emph{free hyperedge variables} and the \emph{guard depth}, as formalised by the functions 
	\begin{equation*}
		\freeV\colon \GCk \to \pot(\VARV)\, 
		,\quad
		\freeE\colon \GCk \to \pot(\VARE)\,
		,\quad\text{and}\quad
		\gdepth\colon \GCk \to \nat.
	\end{equation*}

	\begin{description}
		\item[Atomic Formulas] 
		For all $i, i' \in \natpos$ and all $j, j' \in [k]$ the following formulas are in $\GCk$: %
		\settowidth{\myA}{$\phi \, = \, E(\vare_j, \varv_i)$\: }%
		\settowidth{\myB}{$\freeV(\phi) \isdef \set{ \varv_i, \varv_{i'} }$}%
		\begin{itemize}
			\item 
			\makebox[\myA][l]{ $\phi \, = \, \varv_{i} \logeq \varv_{i'}$	} 
			with
			\makebox[\myB][l]{ $\freeV(\phi) \isdef \set{ \varv_i, \varv_{i'} }$ } 
			\; and \;
			$\freeE(\phi) \isdef \emptyset$;

			\item 
			\makebox[\myA][l]{ $\phi \, = \, \vare_{j} \logeq \vare_{j'}$ } with
			\makebox[\myB][l]{ $\freeV(\phi) \isdef \emptyset$ } 
			\; and \;
			$\freeE(\phi) \isdef \set{ \vare_{j}, \vare_{j'} }$;
			
			\item 
			\makebox[\myA][l]{ $\phi \, = \, E(\vare_j, \varv_i)$ } with
			\makebox[\myB][l]{ $\freeV(\phi) \isdef \set{ \varv_i }$ } 
			\; and \;
			$\freeE(\phi) \isdef \set{ \vare_j }$.
		\end{itemize}
		In all the above cases, $\gdepth(\phi) \isdef 0$.
		\smallskip

		\item[Inductive Rules] Let $\chi, \psi$ be formulas of $\GCk$.
		The following formulas are in $\GCk$.
		\begin{itemize}
			\item 
			\makebox[\myA][l]{ $\phi \, = \, \lnot \chi$ } 
			with
			$\freeV(\phi) \isdef \freeV(\chi)$
			\; and \;
			$\freeE(\phi) \isdef \freeE(\chi)$, \\ \makebox[\myA][l]{}
			and\; $\gdepth(\phi) \isdef \gdepth(\chi)$;

			\item 
			\makebox[\myA][l]{ $\phi \, = \, (\chi \land \psi)$ } 
			with
			$\freeV(\phi) \isdef \freeV(\chi) \union \freeV(\psi)$
			\ and \
			$\freeE(\phi) \isdef \freeE(\chi) \union \freeE(\psi)$, \\ \makebox[\myA][l]{}
			and\; $\gdepth(\phi) \isdef \max\set{\, \gdepth(\chi),\, \gdepth(\psi)\, }$.
		\end{itemize}
		Note that by the rules defined so far, $\gdepth(\LogGuard{g}) = 0$ for all logical guards $\LogGuard{g}$.
		\smallskip

		We say that \emph{$g\colon \natpos \pto [k]$ is a guard function for $\phi$} if $\dom(g) = \set{ i \mid \varv_i \in \freeV(\phi) }$.

		Let $n \in \natpos$, let $g$ be a guard function for $\psi$ and $\chi = (\LogGuard{g} \land \psi)$.
		The following formulas are in $\GCk$ for every $\varv$-tuple $\tupel{\varv}$ with $\vset(\tupel{\varv}) \subseteq \freeV(\chi)$ and every $\vare$-tuple $\tupel{\vare}$ with $\vset(\tupel{\vare}) \subseteq \freeE(\chi)$:
		\begin{itemize}
			\item 
			\makebox[\myA][l]{ $\phi \, = \, \existsgeq{n} \tupel{\varv} \qsep \chi$ } 
			with
			$\freeV(\phi) \isdef \freeV(\chi) \setminus \vset(\tupel{\varv})$
			\; and \;
			$\freeE(\phi) \isdef \freeE(\chi)$, \\ \makebox[\myA][l]{}
			and\; $\gdepth(\phi) \isdef \gdepth(\chi)$;
			\item 
			\makebox[\myA][l]{ $\phi \, = \, \existsgeq{n} \tupel{\vare} \qsep \chi$ } 
			with
			$\freeV(\phi) \isdef \freeV(\chi)$
			\; and \;
			$\freeE(\phi) \isdef \freeE(\chi) \setminus \vset(\tupel{\vare})$, \\ \makebox[\myA][l]{}
			and\; $\gdepth(\phi) \isdef \gdepth(\chi) + |\vset(\tupel{\vare})|$.
			\endhere{}
		\end{itemize}
	\end{description}
\end{definition}

For convenience, we let $\free(\phi) \isdef \freeV(\phi) \union \freeE(\phi)$ for all $\phi \in \GCk$.
Formulas of $\GCk$ are evaluated over a hypergraph $\HG$ via interpretations $\IInt = (\I_\HG, \assignmentV, \assignmentE)$ that consist of $\HG$'s incidence graph $\I_\HG$ and assignments $\assignmentV\colon \VARV \to \R(\I_\HG)$ and $\assignmentE\colon \VARE \to \B(\I_\HG)$.
The semantics of $\GCk$ are as expected and a definition can be found in Section~6 of the full version of~\cite{Scheidt2023}, thus we do not give one here.
A \emph{sentence} is a formula $\phi \in \GCk$ that has neither free vertex, nor free hyperedge variables, i.e., $\free(\phi) = \emptyset$.
By $\GCk_d$ we denote the fragment $\set{ \phi \in \GCk \mid \gdepth(\phi) \leq d }$, and we let $\GCDk \isdef \GCk_k$.
We write $\GG \equiv_{\Logic{L}} \HG$ to denote that $\GG$ and $\HG$ satisfy the same sentences in the fragment $\Logic{L} \subseteq \GCk$.

For simplicity, we omit logical guards if they are empty or equal to the formula they are guarding.
I.e., we may abbreviate subformulas of the form $(\top \land \phi)$ or $(\phi \land \phi)$ as $\phi$.
We may also omit parentheses in the usual way.
We write $\existseq{n}(\tupel{\varx})\qsep(\LogGuard{} \land \phi)$ as shorthand for $\existsgeq{n}(\tupel{\varx})\qsep(\LogGuard{} \land \phi) \land \lnot \existsgeq{n{+}1}(\tupel{\varx})\qsep(\LogGuard{} \land \phi)$.
Clearly, these shorthands change neither the semantics, nor the free variables, nor the guard depth of a formula.

\begin{example} The sentence $\phi_\GG \isdef \psi_1 \land \psi_2 \land \psi_3$ describes $\GG$ from \cref{example:hgs-and-igs} up to isomorphism, where
	\begin{align*}
		\chi_{n}\; &\isdef 
			\textstyle\bigland_{1 \leq i < j \leq n} \lnot \varv_i \logeq \varv_j
			\; \land \;
			\lnot \existsgeq{1} (
				\varv_{n{+}1}) \qsep \Bigl(
					E(\vare, \varv_{n{+}1}) \;\land\; \bigland_{i \in [n]} \lnot \varv_{n{+}1} = \varv_i
			\Bigr)\,,\\
		\psi_1\, &\isdef\, \existseq{4}(\vare) \qsep \vare \logeq \vare \,, \\
		\psi_2\, &\isdef\, \existseq{1} (\vare) \qsep 
		\existsgeq{1} (\varv_1, \varv_2, \varv_3) \qsep \Bigl(
			\textstyle\bigland_{i \in [3]} E(\vare, \varv_i)
			\; \land \; 
			\chi_3
			\;\land\;
			\textstyle\bigland_{i \in [3]} \existseq{3} (\vare) \qsep E(\vare, v_i)
		\Bigr)\,\vphantom{\Bigl(\Bigr)},\\
		\psi_3\, &\isdef\, \existseq{3} (\vare) \qsep 
		\existsgeq{1} (\varv_1, \varv_2) \qsep \Bigl(
			\textstyle\bigland_{i \in [2]} E(\vare, \varv_i)
			\; \land \; 
			\chi_2
			\;\land\;
			\textstyle\bigland_{i \in [3]} \existseq{3} (\vare) \qsep E(\vare, v_i)
		\Bigr)\,.
	\end{align*}
	It is easily verified that $\phi_\GG \in \GC^1_2$.
	$\chi_n$ is a helper formula, describing that there are precisely $n$ vertices $\varv_1, \dots, \varv_n$ in the hyperedge $\vare$.
	$\psi_1$ describes that there are precisely four hyperedges, $\psi_2$ describes that precisely one hyperedge contains precisely three vertices, each being contained in precisely 3 hyperedges.
	Finally, $\psi_3$ describes that there are exactly 3 hyperedges containing precisely 2 vertices, each being contained in precisely 3 hyperedges.
	It is not hard to see that, in total, this describes $\GG$ up to isomorphism.
\end{example}

Scheidt and Schweikardt~\cite{Scheidt2023} prove their result only for the following restricted variant of $\GCk$, called $\RGCk$.
They mention in the conclusion, that $\RGCk$ and $\GCk$ are equivalent and show in the full version of the paper (Theorem~7.2) how a formula in $\GCk$ can be translated into one in $\RGCk$.
We still need $\RGCk$ since it is used in the formulation and the proof of the two core lemmata of~\cite{Scheidt2023} that we want to borrow.

\begin{definition}[\cite{Scheidt2023}]\label{def:rgck}
	The restriction $\RGCk$ is inductively defined as follows:
	\begin{description}
		\item[Atomic Formulas] 
		$(\LogGuard{g} \land \phi)$ is in $\RGCk$ 
		for all atomic formulas $\phi \in \GCk$ and all guard functions for $\phi$, i.e., all $g\colon \natpos \pto [k]$ with $\dom(g) = \set{ i \mid \varv_i \in \freeV(\phi) }$.
		
		\smallskip
		\item[Inductive Rules]~%
		\begin{itemize}
			\item
			For every formula $(\LogGuard{g} \land \phi) \in \RGCk$, the formula $(\LogGuard{g} \land \lnot \phi)$ is also in $\RGCk$.

			\item
			For $i \in [2]$ and formulas $(\LogGuard{g_i} \land \psi_i) \in \RGCk$, the formula $(\LogGuard{(g_1 \union g_2)} \land (\psi_1 \land \psi_2))$ is in $\RGCk$, if $g_1$ and $g_2$ are compatible.
		\end{itemize}
		\smallskip

		Let $n \in \natpos$, $(\LogGuard{g} \land \phi) \in \RGCk$.
		\begin{itemize}
			\item
			For every $\varv$-tuple $\tupel{\varv}$ with $\vset(\tupel{\varv}) \subseteq \freeV(\phi)$ and index set $S$, the formula $(\LogGuard{\tilde{g}} \land \chi)$ is in $\RGCk$, where
			\[
				\chi \isdef \existsgeq{n} \tupel{\varv} \qsep (\LogGuard{g} \land \phi) 
				\quad\text{and $\tilde{g}$ is the restriction of $g$ to } \dom(g) \setminus S.
			\]
			
			\item
			For every $\vare$-tuple $\tupel{\vare}$ with $\vset(\tupel{\vare}) \subseteq \freeE(\LogGuard{g} \land \phi)$ and index set $S$, the formula 
			\[
				(\LogGuard{\tilde{g}} \land \existsgeq{n} \tupel{\vare} \qsep (\LogGuard{g} \land \phi))
			\]
			is in $\RGCk$, if $\dom(\tilde{g}) = \dom(g)$ and all $i \in \dom(g)$ satisfy
			\begin{equation}\label{eq:syntax:guard}
				\tilde{g}(i) = g(i)
				\quad \text{or} \quad
				\tilde{g}(i) \in S
				\quad \text{or} \quad
				\tilde{g}(i) \not\in \img(g).
			\end{equation}
		\end{itemize}
	\end{description}
	Intuitively, formulas in $\RGCk$ always carry the information, which hyperedge variable currently guards which vertex variable and the logical guards are in a certain sense \enquote{consistent}~\eqref{eq:syntax:guard} along the syntax tree.
\end{definition}
A simple inspection of the inductive proof for Theorem 7.2 in the full version of~\cite{Scheidt2023} shows that the guard depth is unaffected by the translation, giving us the following refined result.
\begin{restatable}{lemma}{rgckisgck}\label{lem:rgck-is-gck}
	For every formula $\phi \in \GCk$ and every guard function $g$ for $\phi$, there exists a formula $(\LogGuard{g} \land \phi_g) \in \RGCk$ such that
	\begin{enumerate}
		\item $(\LogGuard{g} \land \phi) \equiv (\LogGuard{g} \land \phi_g)$,
		\item $\free(\phi) = \free(\phi_g)$,\, and \ $\gdepth(\phi) = \gdepth(\phi_g)$.
	\end{enumerate}
	
\end{restatable} 

\section{Main Result}%
\label{sec:main-result}

We are now ready to plug everything together, which yields our main result.
\begin{restatable}{theorem}{mainthm}\label{thm:main}
	Let $\GG$ and $\HG$ be hypergraphs and let $k \in \natpos$.
	\begin{alignat*}{2}
		\GG \equiv_{\GCDk} \HG 
		&\iff \Hom(\ISHDk, \I_{\GG})\,	&&=\; \Hom(\ISHDk, \I_{\HG}) \\
		&\iff \Hom(\SHDk, \GG)	&&=\; \Hom(\SHDk, \HG).
	\end{alignat*}
\end{restatable}

We use the fact that the proofs for the core Lemmata 8.1 and 8.2 in the work by Scheidt and Schweikardt~\cite{Scheidt2023} actually give us the following refined results.
This is easy to see on inspection of the original proofs (consult Appendix~E in the full version of~\cite{Scheidt2023}), since there is a one-to-one correspondence between the blue label $i$ and the hyperedge variable $\vare_i$ in the proofs for both lemmas: whenever a blue label $i$ is removed, the corresponding variable $\vare_i$ is quantified and vice-versa.

For a $k$-labeled incidence graph $\kLI$ of the form $(\I, r, b, g)$, we let $\IInt_\kLI \isdef (\I, \assignmentV, \assignmentE)$ be defined by $\assignmentV(\varv_i) \isdef r(i)$ for all $i \in \dom(r)$ and $\assignmentE(\vare_j) \isdef b(j)$ for all $j \in \dom(b)$.
\begin{lemma}[implicit in~\cite{Scheidt2023}]\label{lem:main-formula-for-kli}
	Let $\kLI = (\I, r, b, g) \in \GLI_k^i$.
	For every $m \in \nat$ there is a formula $\phi_{\kLI, m}$ with $(\LogGuard{g} \land \phi_{\kLI, m}) \in \RGCk$,
	$\freeV( \LogGuard{g} \land \phi_{\kLI, m} ) = \set{\varv_i \mid i \in \dom(r)}$,\, 
	$\freeE( \LogGuard{g} \land \phi_{\kLI, m} ) = \set{\vare_j \mid j \in \dom(b)}$, and  $\gdepth(\phi_{\kLI,m}) \leq i$,
	such that
	for every $k$-labeled incidence graph $\kLI'$
	with
	$\dom(b_{\kLI'}) \supseteq \dom(b)$,
	$\dom(r_{\kLI'}) \supseteq \dom(r)$,
	and
	with real guards w.r.t.\ $g$ we have:
	$\IInt_{\kLI'} \models \LogGuard{g}$, \;and\; 
	$\hom(\kLI, \kLI') = m \ \iff \ \IInt_{\kLI'} \models \phi_{\kLI, m}$.
\end{lemma}

\begin{lemma}[implicit in~\cite{Scheidt2023}]\label{lem:main-kli-for-formula}
	Let $\chi \isdef (\LogGuard{g} \land \psi) \in \RGCk$ with $\gdepth(\chi) = \ell$ and let
	$m,d \in \nat$ with $m \geq 1$.
	There exists a linear combination 
	$Q \isdef \sum_{i \in [q]} \alpha_i \kLI_i$, and sets 
	$\domrQ = \set{ i \mid \varv_i \in \freeV(\chi) }$ and  
	$\dombQ = \set{ i \mid \vare_i \in \freeE(\chi) }$,
	where for all $i \in [q]$:
	\begin{align*}
		&\alpha_i \in \reell,\quad
		\kLI_i \in \GLI_k^{\ell},\quad
		g_i \, = \, g,\quad
		\dom(b_i) \, = \, \dombQ, 
		\quad\text{and}\quad
		\dom(r_i) \, = \, \domrQ;
	\end{align*}
	such that for all $k$-labeled incidence graphs
	$\kLI'$ with
	$|\B(\I')| = m$,
	$\max \set{|\blueN(e)| \mid e \in \B(\I')} \leq d$
	and
	$\dom(b') \supseteq \dombQ$,
	$\dom(r') \supseteq \domrQ$,
	$g' \supseteq g$, and with real guards w.r.t.\ $g$ 
	we have: \ $\IInt_{\kLI'} \models \LogGuard{g}$, \ and \
	\begin{align*}
		\sum_{i \in [q]} \alpha_i \cdot \hom(\kLI_i, \kLI') = \begin{cases}
			1, &\text{if }\; \IInt_{\kLI'} \models \chi \\
			0, &\text{if }\; \IInt_{\kLI'} \not\models \chi.
		\end{cases}
	\end{align*}
\end{lemma}
\medskip

The proof of \cref{thm:main} works the same way as the one for Theorem~6.1 in~\cite[Section 8]{Scheidt2023}: the second biimplication is provided by \cref{thm:ihom-equals-hom} and \cref{prop:shd-hd-closures}.
The first biimplication is shown via contraposition, where the contraposition of the forward direction uses \cref{lem:main-formula-for-kli} and the one for the backward direction uses \cref{lem:main-kli-for-formula}.

\begin{proof}[Proof of \cref{thm:main}]
	Let $\I = \I_\GG$ and $\J = \I_\HG$.
	If $|\B(\I)| \neq |\B(\J)|$ then $\hom(\I', \I) \neq \hom(\I', \J)$ for the incidence graph $\I' \in \ISHD_1$ that consists of a single blue vertex and no red vertices.
	Similarly, $\I$ and $\J$ are distinguished by a suitable $\GC_1$-sentence of the form $\existsgeq{n}\vare \qsep (\vare \logeq \vare)$.
	If $|\B(\I)| = |\B(\J)|$, consider their corresponding label-free $k$-labeled incidence graphs $\kLI_\I = (\I, \emptyset, \emptyset, \emptyset)$ and $\kLI_\J = (\J, \emptyset, \emptyset, \emptyset)$.
	
	Assume there is an $\I' \in \ISHDk$ such that $\hom(\I', \I) = m_1 \neq m_2 = \hom(\I', \J)$.
	According to \cref{thm:inductive_def_hypertree_depth}, there is a label-free $\kLI \in \GLIDk$ such that $\I' \isomorphic \I_\kLI$, which means $\hom(\kLI, \kLI_\I) = m_1 \neq m_2 = \hom(\kLI, \kLI_\J)$.
	By \cref{lem:main-formula-for-kli} there exists a formula $(\top \land \phi_{\kLI, m_1}) \in \RGCk$ with $\gdepth(\phi_{\kLI, m_1}) \leq k$ such that $\IInt_{\kLI_\I} \models (\top \land \phi_{\kLI, m_1})$ and $\IInt_{\kLI_\J} \not\models (\top \land \phi_{\kLI, m_1})$.
	Hence, $\IInt_{\kLI_\I} \models \phi_{\kLI, m_1}$ and $\IInt_{\kLI_\J} \not\models \phi_{\kLI, m_1}$, and since $\phi_{\kLI, m_1} \in \GCDk$,\; $\GG \not\equiv_{\GCDk} \HG$.
	
	Assume there is a sentence $\phi \in \GLI_k$ with $\gdepth(\phi) = k$ such that $\IInt_{\kLI_\I} \models \phi$ and $\IInt_{\kLI_\J} \not\models \phi$.
	By \cref{lem:rgck-is-gck} there exists a formula $(\top \land \psi) \in \RGCk$ with $\gdepth(\psi) = k$ such that $\IInt_{\kLI_\I} \models (\top \land \psi)$ and $\IInt_{\kLI_\J} \not\models (\top \land \psi)$.
	Let $m \isdef |\B(\I)| = |\B(\J)|$ be the number of hyperedges and let $n \in \nat$ such that $|\blueN(e)| \leq n$ for all $e \in \B(\I)$ and all $e \in \B(\J)$.
	According to \cref{lem:main-kli-for-formula} there exists a linear combination $Q = \sum_{i \in [q]} \alpha_i \kLI_i$ such that $\sum_{i \in [q]} \alpha_i \cdot \hom(\kLI_i, \kLI_\I) = 1$ and $\sum_{i \in [q]} \alpha_i \cdot \hom(\kLI_i, \kLI_\J) = 0$ and $\kLI_i \in \GLI_k^k$ for all $i \in [q]$.
	This means there must be an $i \in [q]$ such that $\alpha_i \cdot \hom(\kLI_i, \kLI_\I) \neq \alpha_i \cdot \hom(\kLI_i, \kLI_\J)$, which means $\hom(\kLI_i, \kLI_\I) \neq \hom(\kLI_i, \kLI_\J)$.
	Since $\domrQ = \dombQ = \emptyset$, $\kLI_i$ is label-free.
	According to \cref{thm:inductive_def_hypertree_depth}, there exists an $\I' \in \ISHD_k$ such that $\I' \isomorphic \I_{\kLI_i}$.
	Thus, $\hom(\I', \I) \neq \hom(\I', \J)$, i.e., $\Hom(\ISHDk, \I) \neq \Hom(\ISHDk, \J)$.
	
	This finishes the proof for the first \enquote{iff}.
	The second is provided by the combination of \cref{thm:ihom-equals-hom} and \cref{prop:shd-hd-closures}.
\end{proof}  

\section{Final Remarks}%
\label{sec:conclusion}

This paper solves one of the open questions of Scheidt and Schweikardt~\cite{Scheidt2023}, who lift a result by Dvo\v{r}ák~\cite{Dvorak2010} from graphs to hypergraphs.
Dvo\v{r}ák shows that homomorphism indistinguishability over the graphs of tree width at most $k$ is equivalent to indistinguishability over first-order logic with counting quantifiers ($\CL$) and $k{+}1$ variables ($\CL^{k{+}1}$).
Scheidt and Schweikardt show that homomorphism indistinguishability over the class $\GHWk$ of hypergraphs of generalised hypertree width at most $k$ is equivalent to indistinguishability over the logic $\GC$ with $k$ guards ($\GCk$).
Grohe~\cite{Grohe2020} gave a result complementing Dvo\v{r}ák's: $\CL$ with quantifier depth at most $m$ ($\CL_m$) matches homomorphism indistinguishability over graphs of tree depth at most $m$.
An obvious expectation was that the distinguishing power of $\GCD{m}$ would match homomorphism indistinguishability over the class $\HD_m$ of hypergraphs of hypertree depth at most $m$ as it is defined by Adler et al.~\cite{Adler2012}.
However, this expectation did not manifest in this exact way.
Instead, we proved that the distinguishing power of $\GCD{m}$ matches homomorphism indistinguishability over hypergraphs of \emph{strict} hypertree depth at most $m$, which is a (mild) restriction of hypertree depth.
Combining \cref{thm:main} with the main result of~\cite{Scheidt2023} yields the following combined result.
\begin{theorem}
	For all hypergraphs $\GG$ and $\HG$, the following equivalences hold:
	\begin{alignat*}{2}
		\GG \equiv_{\GCDk} \HG &\;\iff\; \GG \equiv_{\SHDk} \HG  &&\;\iff\; \I_{\GG} \equiv_{\ISHDk} \I_{\HG} \quad\text{and}\\
		\GG \equiv_{\GCk} \HG &\;\iff\; \GG \equiv_{\GHWk} \HG &&\;\iff\; \I_{\GG} \equiv_{\IGHWk} \I_{\HG}.
	\end{alignat*}
\end{theorem}

We took this unexpected mismatch between $\GCDk$ and $\HDk$ as an opportunity to investigate the relationship between $\HDk$ and $\SHDk$.
In \cref{thm:shd-is-almost-hd} we showed that the strict hypertree depth of a hypergraph is at most 1 larger than its hypertree depth.
\shdisalmosthd*

To show that homomorphism counts from the class $\SHDk$ are just as expressive as homomorphism counts from the class $\ISHDk$, which was necessary to prove \cref{thm:main}, we used an implicit result by Böker~\cite{Boeker2019}, who gives a sufficient set of properties for a class $\classC$ of hypergraphs, such that homomorphism indistinguishability over $\classC$ is the same as homomorphism indistinguishability over the corresponding class $\classIC$ of incidence graphs.
Since $\HDk$ does not have these properties, Böker's result cannot be applied with respect to $\HDk$ and $\IHDk$.
In fact, we showed in \cref{thm:hd-hom-are-skewed} that homomorphism indistinguishability over $\HDk$ is \emph{not} the same as homomorphism indistinguishability over $\IHDk$ and furthermore, that it is also not the same as homomorphism indistinguishability over $\SHDk$.
\hdhomareskewed*

\subparagraph*{Further Research}
It would be very interesting to see if the result by Böker (\cref{thm:ihom-equals-hom}) is tight in the sense that closure under pumping and local merging are sufficient \emph{and required} properties.
I.e., whether for every class $\classC$ that misses one of these properties, homomorphism counts over $\classC$ differ from homomorphism counts over the corresponding class $\classIC$ of incidence graphs in their distinguishing power.

As mentioned in the introduction, this work can be seen as one more step in the search of a \enquote{proper} lifting of the $k$-dimensional Weisfeiler-Leman algorithm to hypergraphs.
Given the relationship between Weisfeiler-Leman, $\CL$ and homomorphism indistinguishability on graphs~\cite{Cai1992,Dawar2021,Dell2018,Dvorak2010,Fluck2024,Grohe2020}, we believe that the proper lifting should admit a similar relationship to the corresponding hypergraph parameters.
Hence, we believe that the distinguishing power of such an algorithm should match homomorphism indistinguishability over the class $\GHWk$ of hypergraphs of generalised hypertree width at most $k$ and thus also indistinguishability by the logic $\GCk$.
Since we believe that $\GCk$ is the natural lifting of $\CL^k$ in this setting, this paper adds to this picture: The $k$-dimensional Weisfeiler-Leman algorithm restricted to $m$ iterations should have the same distinguishing power as the intersection of the classes $\GHWk \intersect \SHD_m$.
Hence, the mismatch we uncovered in this work might propagate to the Weisfeiler-Leman algorithm. 

\bibliography{literature}

\begin{thebibliography}{10}

\bibitem{Adler2012}
Isolde Adler, {Tomá\v{s}} {Gaven\v{c}iak}, and Tereza {Klimo\v{s}ová}.
\newblock Hypertree-depth and minors in hypergraphs.
\newblock {\em Theoretical Computer Science}, 463:84--95, 2012.
\newblock \href {https://doi.org/10.1016/j.tcs.2012.09.007} {\path{doi:10.1016/j.tcs.2012.09.007}}.

\bibitem{Boeker2019a}
Jan B\"{o}ker, Yijia Chen, Martin Grohe, and Gaurav Rattan.
\newblock {The Complexity of Homomorphism Indistinguishability}.
\newblock In Peter Rossmanith, Pinar Heggernes, and Joost-Pieter Katoen, editors, {\em 44th International Symposium on Mathematical Foundations of Computer Science (MFCS 2019)}, volume 138 of {\em Leibniz International Proceedings in Informatics (LIPIcs)}, pages 54:1--54:13, Dagstuhl, Germany, 2019. Schloss Dagstuhl -- Leibniz-Zentrum f{\"u}r Informatik.
\newblock \href {https://doi.org/10.4230/LIPIcs.MFCS.2019.54} {\path{doi:10.4230/LIPIcs.MFCS.2019.54}}.

\bibitem{Butti2021}
Silvia Butti and V{\'\i}ctor Dalmau.
\newblock {Fractional Homomorphism, Weisfeiler-Leman Invariance, and the Sherali-Adams Hierarchy for the Constraint Satisfaction Problem}.
\newblock In Filippo Bonchi and Simon~J. Puglisi, editors, {\em 46th International Symposium on Mathematical Foundations of Computer Science (MFCS 2021)}, volume 202 of {\em Leibniz International Proceedings in Informatics (LIPIcs)}, pages 27:1--27:19, Dagstuhl, Germany, 2021. Schloss Dagstuhl -- Leibniz-Zentrum f{\"u}r Informatik.
\newblock \href {https://doi.org/10.4230/LIPIcs.MFCS.2021.27} {\path{doi:10.4230/LIPIcs.MFCS.2021.27}}.

\bibitem{Boeker2019}
Jan Böker.
\newblock {C}olor {R}efinement, {H}omomorphisms, and {H}ypergraphs.
\newblock In Ignas Sau and Dimitrios~M. Thilikos, editors, {\em Graph-Theoretic Concepts in Computer Science}, volume 11789 of {\em Lecture Notes in Computer Science}, pages 338--350. Springer, Cham, 2019.

\bibitem{Cai1992}
Jin-Yi Cai, Martin Fürer, and Neil Immerman.
\newblock An optimal lower bound on the number of variables for graph identification.
\newblock {\em Combinatorica}, 12(4):389--410, December 1992.
\newblock \href {https://doi.org/10.1007/BF01305232} {\path{doi:10.1007/BF01305232}}.

\bibitem{Courcelle1993}
Bruno Courcelle.
\newblock Graph {{Grammars}}, {{Monadic Second-Order Logic And The Theory Of Graph Minors}}.
\newblock In Neil Robertson and Paul Seymour, editors, {\em Graph {{Structure Theory}}}, number 147 in Contemporary {{Mathematics}}. {American Mathematical Society}, 1993.

\bibitem{Dawar2021}
Anuj Dawar, {Tomá\v{s}} Jakl, and Luca Reggio.
\newblock Lovász-{{Type Theorems}} and {{Game Comonads}}.
\newblock In {\em 2021 36th {{Annual ACM}}/{{IEEE Symposium}} on {{Logic}} in {{Computer Science}} ({{LICS}})}, pages 1--13, June 2021.
\newblock \href {https://doi.org/10.1109/LICS52264.2021.9470609} {\path{doi:10.1109/LICS52264.2021.9470609}}.

\bibitem{Dell2018}
Holger Dell, Martin Grohe, and Gaurav Rattan.
\newblock Lov{\'{a}}sz meets {W}eisfeiler and {L}eman.
\newblock In {\em 45th International Colloquium on Automata, Languages, and Programming, {ICALP} 2018, July 9-13, 2018, Prague, Czech Republic}, volume 107 of {\em LIPIcs}, pages 40:1--40:14. Schloss Dagstuhl - Leibniz-Zentrum f{\"{u}}r Informatik, 2018.
\newblock \href {https://doi.org/10.4230/LIPIcs.ICALP.2018.40} {\path{doi:10.4230/LIPIcs.ICALP.2018.40}}.

\bibitem{Dvorak2010}
Zden{\v{e}}k Dvo{{\v{r}}}{\'{a}}k.
\newblock On recognizing graphs by numbers of homomorphisms.
\newblock {\em Journal of Graph Theory}, 64(4):330--342, 2010.
\newblock \href {https://doi.org/10.1002/jgt.20461} {\path{doi:10.1002/jgt.20461}}.

\bibitem{Fluck2024}
Eva Fluck, Tim Seppelt, and Gian~Luca Spitzer.
\newblock {Going Deep and Going Wide: Counting Logic and Homomorphism Indistinguishability over Graphs of Bounded Treedepth and Treewidth}.
\newblock In Aniello Murano and Alexandra Silva, editors, {\em 32nd EACSL Annual Conference on Computer Science Logic (CSL 2024)}, volume 288 of {\em Leibniz International Proceedings in Informatics (LIPIcs)}, pages 27:1--27:17, Dagstuhl, Germany, 2024. Schloss Dagstuhl -- Leibniz-Zentrum f{\"u}r Informatik.
\newblock \href {https://doi.org/10.4230/LIPIcs.CSL.2024.27} {\path{doi:10.4230/LIPIcs.CSL.2024.27}}.

\bibitem{Grohe2020}
Martin Grohe.
\newblock Counting {Bounded} {Tree} {Depth} {Homomorphisms}.
\newblock In {\em Proceedings of the 35th {Annual} {ACM}/{IEEE} {Symposium} on {Logic} in {Computer} {Science}}, {LICS} '20, pages 507--520, New York, NY, USA, July 2020. Association for Computing Machinery.
\newblock \href {https://doi.org/10.1145/3373718.3394739} {\path{doi:10.1145/3373718.3394739}}.

\bibitem{Grohe2020a}
Martin Grohe.
\newblock {W}ord2vec, {N}ode2vec, {G}raph2vec, {X}2vec: {T}owards a {T}heory of {V}ector {E}mbeddings of {S}tructured {D}ata.
\newblock In {\em Proceedings of the 39th ACM SIGMOD-SIGACT-SIGAI Symposium on Principles of Database Systems}, PODS'20, pages 1--16. ACM, 2020.
\newblock \href {https://doi.org/10.1145/3375395.3387641} {\path{doi:10.1145/3375395.3387641}}.

\bibitem{Grohe2021a}
Martin Grohe.
\newblock The {{Logic}} of {{Graph Neural Networks}}.
\newblock In {\em 2021 36th {{Annual ACM}}/{{IEEE Symposium}} on {{Logic}} in {{Computer Science}} ({{LICS}})}, pages 1--17, June 2021.
\newblock \href {https://doi.org/10.1109/LICS52264.2021.9470677} {\path{doi:10.1109/LICS52264.2021.9470677}}.

\bibitem{Grohe2014}
Martin Grohe, Kristian Kersting, Martin Mladenov, and Erkal Selman.
\newblock Dimension {Reduction} via {Colour} {Refinement}.
\newblock In Andreas~S. Schulz and Dorothea Wagner, editors, {\em Algorithms - {ESA} 2014}, Lecture {Notes} in {Computer} {Science}, pages 505--516, Berlin, Heidelberg, 2014. Springer.
\newblock \href {https://doi.org/10.1007/978-3-662-44777-2_42} {\path{doi:10.1007/978-3-662-44777-2_42}}.

\bibitem{Grohe2022}
Martin Grohe, Gaurav Rattan, and Tim Seppelt.
\newblock {Homomorphism Tensors and Linear Equations}.
\newblock In Miko{\l}aj Boja\'{n}czyk, Emanuela Merelli, and David~P. Woodruff, editors, {\em 49th International Colloquium on Automata, Languages, and Programming (ICALP 2022)}, volume 229 of {\em Leibniz International Proceedings in Informatics (LIPIcs)}, pages 70:1--70:20, Dagstuhl, Germany, 2022. Schloss Dagstuhl -- Leibniz-Zentrum f{\"u}r Informatik.
\newblock \href {https://doi.org/10.4230/LIPIcs.ICALP.2022.70} {\path{doi:10.4230/LIPIcs.ICALP.2022.70}}.

\bibitem{Immerman1990}
Neil Immerman and Eric Lander.
\newblock Describing {{Graphs}}: {{A First-Order Approach}} to {{Graph Canonization}}.
\newblock In Alan~L. Selman, editor, {\em Complexity {{Theory Retrospective}}: {{In Honor}} of {{Juris Hartmanis}} on the {{Occasion}} of {{His Sixtieth Birthday}}, {{July}} 5, 1988}, pages 59--81. {Springer}, {New York, NY}, 1990.
\newblock \href {https://doi.org/10.1007/978-1-4612-4478-3_5} {\path{doi:10.1007/978-1-4612-4478-3_5}}.

\bibitem{Kiefer2020}
Sandra Kiefer.
\newblock The {{Weisfeiler-Leman Algorithm}}: {{An Exploration}} of {{Its Power}}.
\newblock {\em ACM SIGLOG News}, 7(3):5--27, November 2020.
\newblock \href {https://doi.org/10.1145/3436980.3436982} {\path{doi:10.1145/3436980.3436982}}.

\bibitem{Lovasz1967}
László Lovász.
\newblock {O}perations with structures.
\newblock {\em Acta Mathematica Academiae Scientiarum Hungaricae}, 18(3):321--328, 1967.

\bibitem{Lovasz2009}
László Lovász and Balázs Szegedy.
\newblock Contractors and connectors of graph algebras.
\newblock {\em Journal of Graph Theory}, 60(1):11--30, 2009.
\newblock \href {https://doi.org/10.1002/jgt.20343} {\path{doi:10.1002/jgt.20343}}.

\bibitem{Mancinska2020}
Laura {Man\v{c}inska} and David~E. Roberson.
\newblock Quantum isomorphism is equivalent to equality of homomorphism counts from planar graphs.
\newblock In {\em 2020 {{IEEE}} 61st {{Annual Symposium}} on {{Foundations}} of {{Computer Science}} ({{FOCS}})}, pages 661--672, November 2020.
\newblock \href {https://doi.org/10.1109/FOCS46700.2020.00067} {\path{doi:10.1109/FOCS46700.2020.00067}}.

\bibitem{Montacute2022}
Yoàv Montacute and Nihil Shah.
\newblock The {{Pebble-Relation Comonad}} in {{Finite Model Theory}}.
\newblock In {\em Proceedings of the 37th {{Annual ACM}}/{{IEEE Symposium}} on {{Logic}} in {{Computer Science}}}, {{LICS}} '22, pages 1--11, {New York, NY, USA}, August 2022. {Association for Computing Machinery}.
\newblock \href {https://doi.org/10.1145/3531130.3533335} {\path{doi:10.1145/3531130.3533335}}.

\bibitem{Morris2019}
Christopher Morris, Martin Ritzert, Matthias Fey, William~L. Hamilton, Jan~Eric Lenssen, Gaurav Rattan, and Martin Grohe.
\newblock Weisfeiler and {{Leman Go Neural}}: {{Higher-Order Graph Neural Networks}}.
\newblock {\em Proceedings of the AAAI Conference on Artificial Intelligence}, 33(01):4602--4609, July 2019.
\newblock \href {https://doi.org/10.1609/aaai.v33i01.33014602} {\path{doi:10.1609/aaai.v33i01.33014602}}.

\bibitem{Neuen2023}
Daniel Neuen.
\newblock Homomorphism-{{Distinguishing Closedness}} for {{Graphs}} of {{Bounded Tree-Width}}, July 2023.
\newblock \href {https://arxiv.org/abs/2304.07011} {\path{arXiv:2304.07011}}, \href {https://doi.org/10.48550/arXiv.2304.07011} {\path{doi:10.48550/arXiv.2304.07011}}.

\bibitem{Rattan2023}
Gaurav Rattan and Tim Seppelt.
\newblock Weisfeiler-{{Leman}} and {{Graph Spectra}}.
\newblock In {\em Proceedings of the 2023 {{Annual ACM-SIAM Symposium}} on {{Discrete Algorithms}} ({{SODA}})}, Proceedings, pages 2268--2285. {Society for Industrial and Applied Mathematics}, January 2023.
\newblock \href {https://doi.org/10.1137/1.9781611977554.ch87} {\path{doi:10.1137/1.9781611977554.ch87}}.

\bibitem{Roberson2022}
David~E. Roberson.
\newblock Oddomorphisms and homomorphism indistinguishability over graphs of bounded degree, June 2022.
\newblock \href {https://arxiv.org/abs/2206.10321} {\path{arXiv:2206.10321}}, \href {https://doi.org/10.48550/arXiv.2206.10321} {\path{doi:10.48550/arXiv.2206.10321}}.

\bibitem{Scheidt2023}
Benjamin Scheidt and Nicole Schweikardt.
\newblock Counting {Homomorphisms} from {Hypergraphs} of {Bounded} {Generalised} {Hypertree} {Width}: {A} {Logical} {Characterisation}.
\newblock In Jérôme Leroux, Sylvain Lombardy, and David Peleg, editors, {\em 48th {International} {Symposium} on {Mathematical} {Foundations} of {Computer} {Science} ({MFCS} 2023)}, volume 272 of {\em Leibniz {International} {Proceedings} in {Informatics} ({LIPIcs})}, pages 79:1--79:15, Dagstuhl, Germany, 2023. Schloss Dagstuhl -- Leibniz-Zentrum für Informatik.
\newblock Full version available at arXiv: \href{https://arxiv.org/abs/2303.10980}{arXiv:2303.10980 [cs.LO]}.
\newblock \href {https://doi.org/10.4230/LIPIcs.MFCS.2023.79} {\path{doi:10.4230/LIPIcs.MFCS.2023.79}}.

\bibitem{Seppelt2023}
Tim Seppelt.
\newblock {Logical Equivalences, Homomorphism Indistinguishability, and Forbidden Minors}.
\newblock In J\'{e}r\^{o}me Leroux, Sylvain Lombardy, and David Peleg, editors, {\em 48th International Symposium on Mathematical Foundations of Computer Science (MFCS 2023)}, volume 272 of {\em Leibniz International Proceedings in Informatics (LIPIcs)}, pages 82:1--82:15, Dagstuhl, Germany, 2023. Schloss Dagstuhl -- Leibniz-Zentrum f{\"u}r Informatik.
\newblock \href {https://doi.org/10.4230/LIPIcs.MFCS.2023.82} {\path{doi:10.4230/LIPIcs.MFCS.2023.82}}.

\bibitem{Shervashidze}
Nino Shervashidze.
\newblock Weisfeiler-{{Lehman Graph Kernels}}.
\newblock {\em Journal of Machine Learning Research}, 12(77):2539--2561, 2011.
\newblock URL: \url{http://jmlr.org/papers/v12/shervashidze11a.html}.

\bibitem{Xu2018}
Keyulu Xu, Weihua Hu, Jure Leskovec, and Stefanie Jegelka.
\newblock How {{Powerful}} are {{Graph Neural Networks}}?
\newblock In {\em International {{Conference}} on {{Learning Representations}}}, September 2018.
\newblock URL: \url{https://openreview.net/forum?id=ryGs6iA5Km}.

\end{thebibliography}
\clearpage

\appendix

\section{Proof of Proposition~\ref{prop:shd-hd-closures}}\label{app:shd-hd-closures}

\subparagraph*{Closure under Pumping} %
It is easy to see that a \emph{bijective} elimination forest stays a bijective elimination forest, when we add a fresh vertex to the neighbourhood of a single blue vertex, since there already exists a node in the forest that maps to this blue vertex.
Obviously \cref{def:hd:vertex-complete,def:hd:edge-containment} remain satisfied when the vertex is added.
And since it is only contained in this single edge, \cref{def:hd:shared-heritage} remains satisfied as well.

On the other hand, if the elimination forest $(\F, \Gamma)$ is not bijective, problems arise when a fresh vertex is added to a hyperedge $e \not\in \img(\Gamma)$, i.e., one that $\Gamma$ does not map to.
In this case, \cref{def:hd:vertex-complete} and \cref{def:hd:edge-containment} are no longer satisfied.
To fix this, we have to map to $e$ and this may mean inserting a new node into $\F$, which might increase its height.
A simple example where this happens is when we add a vertex to one of the singleton hyperedges of $\GG_k$ from \cref{thm:hd-hom-are-skewed}.

\subparagraph*{Closure under Local Merging} %
Let $\I$ be an incidence graph and let $(\F, \Gamma)$ be an elimination forest for $\I$.
Let $u,v \in \R(\I)$ be distinct vertices and let $e \in \B(\I)$ such that $u,v \in \blueN(e)$.
Let $\pi\colon \R(\I) \to \R(\I)$ be defined by $\pi(u) = v$ and $\pi(w) = w$ for all $w \in \R(\I) \setminus \set{ u }$.
Let $\I'$ be the projection of $\I$ that is defined as follows:
\[
	\R(\I') = \pi(\R(\I)),\; \B(\I') = \B(\I), \quad\text{and}\quad \E(\I') = \set{ (e, \pi(v)) \mid (e, v) \in \E(\I) }.
\]

Now we simply have to verify that $(\F, \Gamma)$ is still an elimination forest for $\I'$.
In particular, if $(\F, \Gamma)$ is a strict elimination forest, it will remain strict, since $\Gamma$ is still bijective.
To avoid confusion, we write $\blueN'(e)$ to denote the set $\set{ v \in \R(\I') \mid (e, v) \in \E(\I') }$ and $\hGamma'(t)$ to denote the set $\blueN'(\Gamma(t))$.
We have to show that \cref{def:hd:vertex-complete,def:hd:edge-containment,def:hd:shared-heritage} still hold with respect to $\blueN'$ and $\hGamma'$.

\smallskip
\noindent\emph{Verifying \cref{def:hd:vertex-complete}.}
Since $\R(\I') \subseteq \R(\I)$ we know that for every $w \in \R(\I')$ there is a $t \in \V(\F)$ such that $w \in \blueN(\Gamma(t))$.
Since by construction $\pi(w) = w$ for all $w \in \R(\I')$, we get that $(e, w) \in \E(\I)$ implies $(e, w) \in \E(\I')$.
Therefore, $w \in \blueN'(\Gamma(t))$, i.e., $w \in \hGamma'(t)$.

\smallskip
\noindent\emph{Verifying \cref{def:hd:edge-containment}.}
Notice that $\B(\I') = \B(\I)$.
Let $e \in \B(\I')$, let $s,t \in \V(\F)$ be the nodes in $\V(\F)$ for $e$ according to \cref{def:hd:edge-containment} with respect to $\I$.
Let $w' \in \blueN'(e)$.
Then there exists a $w \in \R(\I)$ such that $(e, \pi(w)) \in \E(\I)$ and $\pi(w) = w'$.
Since $\blueN(e) \subseteq \bigunion \hGamma(\p(s,t))$, there exists a $p \in \p(s,t)$ such that $w \in \hGamma(p)$, i.e., $(\Gamma(p), w) \in \E(\I)$.
This implies by construction, that $(\Gamma(p), \pi(w)) \in \E(\I')$, i.e., $(\Gamma(p), w') \in \E(\I')$.
Thus, $w' \in \hGamma'(p)$, i.e., $w' \in \hGamma'(\p(s,t))$.
Since we chose an arbitrary $w' \in \blueN'(e)$, this yields $\blueN'(e) \subseteq \bigunion \hGamma'(\p(s,t))$.

\smallskip
\noindent\emph{Verifying \cref{def:hd:shared-heritage}.}
Again, notice that $\B(\I') = \B(\I)$.
Let $s,t \in \V(\F)$ and let $w' \in \hGamma'(s) \intersect \hGamma'(t)$.

\emph{Case 1}:\; Assume there exists a $w \in \R(\I)$ such that $\pi(w) = w'$, and $(\Gamma(s), w), (\Gamma(t), w) \in \E(\I)$.
Thus, $w \in \hGamma(s) \intersect \hGamma(t)$.
Since $(\F, \Gamma)$ is an elimination forest for $\I$, according to \cref{def:hd:shared-heritage} the node $\lcv(s,t)$ exists and $w \in \bigunion \hGamma(\p(\lcv(s,t)))$.
I.e., there is a $p \in \p(\lcv(s,t))$ such that $w \in \hGamma(p)$, i.e., $(\Gamma(p), w) \in \E(\I)$.
Hence, by construction, $(\Gamma(p), w') \in \E(\I')$, i.e., $w' \in \hGamma'(p)$, which means $w' \in \bigunion \hGamma'(\p(\lcv(s,t)))$.

\emph{Case 2}: Otherwise, $\pi(u) = w$ and $\pi(v) = w$, and $w = v$ must hold.
W.l.o.g.\ let $(\Gamma(s), u) \in \E(\I)$ and $(\Gamma(t), v) \in \E(\I)$.
Since $u$ and $v$ appear in a common hyperedge $e \in \B(\I)$, there are $s',t' \in \V(\F)$ such that $u,v \in \bigunion \hGamma(\p(s',t'))$.
I.e., there are $p_u, p_v \in \p(s',t')$ such that $u \in \hGamma(p_u)$ and $v \in \hGamma(p_v)$.
Thus, $u \in \hGamma(p_u) \intersect \hGamma(s)$ and $v \in \hGamma(t) \intersect \hGamma(p_v)$.
This means that $u \in \bigunion \hGamma(\p(\lcv(s,p_u)))$ and $v \in \bigunion \hGamma(\p(\lcv(t, p_v)))$.
Since $p_u, p_v \in \p(s', t')$, either $p_u \leq_\F p_v$ or $p_v \leq_\F p_u$, and because of that also $\lcv(s,p_u) \leq_\F \lcv(t,p_v)$ or $\lcv(t,p_v) \leq_\F \lcv(s,p_u)$.
Since $\F$ is a forest, this means that $\lcv(s,t)$ exists and 
$\lcv(s,t) \geq_\F \lcv(\lcv(s,p_u))$ or $\lcv(s,t) \geq_{\F} \lcv(t,p_v)$.
Therefore, $u \in \bigunion \hGamma(\p(\lcv(s,t)))$ or $v \in \bigunion \hGamma(\p(\lcv(s,t)))$, i.e., $w' \in \bigunion \hGamma'(\p(\lcv(s,t)))$. %

\section{Proof of Theorem~\ref{thm:hd-hom-are-skewed}}\label{app:hd-hom-are-skewed}

\begin{observation}\label{obs:hom-of-conn-hg-is-conn}
	Let $\HG$, $\GG$ be graphs and let $(\hV, \hE)$ be a homomorphism from $\HG$ into $\GG$.
	If $\HG$ is connected, then the homomorphic image of $\HG$ in $\GG$ is also connected.
\end{observation}

For every hypergraph $\HG$ and $e \in \E(\HG)$, let $\HG \setminus e$ be the hypergraph $\HG'$ with $\V(\HG') = \V(\HG) \setminus \f_{\HG}(e)$, $\E(\HG') = \E(\HG) \setminus \set{ e }$ and let $\f_{\HG'}$ be the corresponding restriction of $\f_{\HG}$.

\begin{lemma}\label{lem:homs-from-shd-are-short}
	A connected hypergraph $\HG \in \SHDk$ cannot have a homomorphism $(\hV, \hE)$ into $\Pn_n$ for any $n \geq 2^k$, such that $\hV$ and $\hE$ are surjective.
\end{lemma}
\begin{proof}
	Clearly this holds for $k = 1$ since all connected $\HG \in \SHD_1$ consist of a single hyperedge, which cannot be mapped to two hyperedges.

	Assume there is a $k \geq 2$ for which the statement of the lemma does not hold.
	Consider the smallest $k$, for which there is a connected hypergraph $\HG \in \SHDk$ and a homomorphism $(\hV, \hE)$ from $\HG$ into $\Pn_n$ for $n \geq 2^k$.
	Let $(\T, \Gamma)$ be a strict elimination tree for $\HG$ and let $\troot$ be the root of $\T$.
	Let $\HG' \isdef \HG \setminus \Gamma(\troot)$.
	Then the restriction $(\hV', \hE')$ of $(\hV, \hE)$ to $\V(\HG')$ and $\E(\HG')$ is still a homomorphism from $\HG'$ into $\Pn_n$, but $\hV'$ or $\hE'$ (or both) may no longer be surjective and $\shd(\HG') = k-1$.
	
	If $\HG'$ is not connected, consider the connected components $C_1, \dots, C_\ell$ of $\HG'$.
	The restriction $(\hV^i, \hE^i)$ of $(\hV', \hE')$ to $C_i$ is a homomorphism from $C_i$ to $\Pn_n$ and according to \cref{obs:hom-of-conn-hg-is-conn} it is connected, but since $k$ is minimal, it covers a path of at most $2^{k{-}1}-1$ hyperedges.
	And since $\HG$ was connected, these paths must all touch $\hE(\Gamma(\troot))$.
	Thus, the number of hyperedges $\HG$ covered could have been at most $1 + 2 \cdot (2^{k{-}1}-1) = 2^k - 1$, which contradicts our assumption.
	Therefore, $\HG'$ must still be connected.
	
	If $\hV'$ and $\hE'$ are surjective, $k$ was not minimal, i.e., $\hV'$ or $\hE'$ must no longer be surjective.
	Since $\HG'$ is still connected, the homomorphic image must be as well.
	Thus, $\hE(\Gamma(\troot))$ must be the first or the last edge.
	I.e., $\HG'$ has a surjective homomorphism to $\Pn_{n-1}$.
	Since $\shd(\HG') = k-1$ and $n \geq 2^k$ and $k$ was minimal, this means that $n-1 < 2^{k-1}$.
	This is the case, iff $n=1$ and $k=0$, which contradicts $k \geq 2$.
\end{proof}

Notice that this lemma only holds for $\SHDk$.
If we do not require $(\T, \Gamma)$ to be strict, the proof fails since we can not be sure that the hypertree depth of $\HG'$ is smaller than that of $\HG$.
But since $\HDk \subseteq \SHD_{k{+}1}$, we get the following corollary.
\begin{corollary}\label{cor:homs-from-hd-are-short}
	A connected hypergraph $\HG \in \HDk$ cannot have a homomorphism $(\hV, \hE)$ into $\Pn_n$ for any $n \geq 2^{k{+}1}$, such that $\hV$ and $\hE$ are surjective.
\end{corollary}

For the proof of \cref{thm:hd-hom-are-skewed}, we need to tighten this result a bit further.
\begin{lemma}\label{lem:homs-from-hd-are-short}
	A connected hypergraph $\HG \in \HDk$ cannot have a homomorphism $(\hV, \hE)$ into $\Pn_n$ for any $n \geq 2^{k{+}1}-2$, such that $\hV$ and $\hE$ are surjective.
\end{lemma}
\begin{proof}
	Assume for contradiction, that $\HG \in \HDk$ has a homomorphism $(\hV, \hE)$ into $\Pn_{n}$ for $n \geq 2^{k{+}1}-2$, such that $\hV$ and $\hE$ are surjective.
	Let $\E(\Pn_n) = \set{ e_1, \dots, e_n }$ and $\V(\Pn_n) = \set{ v_1, \dots, v_{n{+}1} }$.

	Let $(\T, \Gamma)$ be an elimination forest for $\HG$ of height $k$.
	There are $e, e' \in \E(\HG)$ and $v \in \f(e), v' \in \f(e')$ such that $\hE(e) = e_1$, $\hE(e') = e_n$, and $\hV(v) = v_1$, $\hV(v') = v_{n{+}1}$.
	Consider the hypergraph $\HG'$ with $\V(\HG') = \V(\HG) \disunion \set{ u,w }$, $\E(\HG') = \E(\HG) \disunion \set{ f,g }$ and $\f_{\HG'} = \f_{\HG} \union \set{ f \to \set{ v, u }, g \to \set{ v',w } }$.
	It is easy to see that $(\T, \Gamma)$ can be turned into a strict elimination tree for $\HG'$ of height $k{+}1$: First we apply the method presented in the proof of \cref{thm:shd-is-almost-hd}.
	Then we insert two new nodes that map to $f$ and $g$.
	It is not hard to see that we can add these new nodes as children of \emph{inner nodes}, which means the height does not increase any further.

	It is also not hard to see that $(\hV', \hE')$ is a homomorphism from $\HG'$ to $\Pn_{n{+}2}$, where $(\hV', \hE')$ are defined in the following way:
	\begin{itemize}
		\item 
		Let $\hV'(u) = v_1$, $\hV'(w) = v_{n{+}3}$ and for all $v \in \V(\HG)$, let $\hV'(v) = v_{i{+}1}$, where $i \in [n{+}1]$ such that $\hV(v) = v_i$.

		\item
		Let $\hE'(f) = e_1$, $\hE'(g) = e_{n{+}2}$ and for all $e \in \E(\HG)$, let $\hE'(e) = e_{i{+}1}$, where $i \in [n]$ such that $\hE(e) = e_i$.
	\end{itemize}
	This contradicts \cref{lem:homs-from-shd-are-short}.
\end{proof}

\begin{observation}\label{obs:hd-inv-to-subedges}
	Inserting a hyperedge into $\HG$, whose contents are a subset of an already existing hyperedge does not increase the hypertree depth.
\end{observation}

\begin{observation}[due to~\cite{Adler2012}]\label{obs:hd-of-paths}
	For all $n \in \natpos$, $\hd(\Pn_n) = \lfloor \log(n+2) \rfloor$.
\end{observation}

\hdhomareskewed*
\begin{proof}
\noindent\textcolor{lipicsGray}{\textbf{\textsf{1.}}}\;
For all $k \geq 2$ let $\GG_k$ and $\HG_k$ be the hypergraphs with 
\begin{align*}
	\V(\GG_k) &= \V(\HG_k) = \V(\Pn_{2^{k{+}1}+1}) = \set{ u_1, u_2, \dots }\,,\\
	\E(\GG_k) &= \E(\HG_k) = \E(\Pn_{2^{k{+}1}+1}) \disunion \set{ e,f,g }\,, \quad\text{and}\\
	\f_{\GG_k} &= \set{ e \to \set{ u_{2^{k{+}1}+2}, u_1 }, f \to \set{ u_1 }, g \to \set{ u_{2^k+1} } } \union \f_{\Pn_{2^{k{+}1}+1}}\,,\\
	\f_{\HG_k} &= \set{ e \to \set{ u_{2^{k{+}1}+2}, u_1 }, f \to \set{ u_1 }, g \to \set{ u_{2^k+2} } } \union \f_{\Pn_{2^{k{+}1}+1}}\,.
\end{align*}
I.e., $\GG_k$ and $\HG_k$ both are circles of $2^{k{+}1}+2$ edges, with two additional singleton hyperedges.
The difference is in the distance between those: There are $2^k$ hyperedges between them in $\GG_k$, but $2^k+1$ in $\HG_k$.

$\GG_k$ and $\HG_k$ are distinguished by $\Pn'$ with $\V(\Pn') = \V(\Pn_{2^k}) = \set{ v_1, v_2, \dots }$, $\E(\Pn') = \E(\Pn') \disunion \set{ f', g' }$ and $\f_{\Pn'} = \set{ f' \to v_1, g' \to v_{2^k+1} } \union \f_{\Pn_{2^k}}$.
Using \cref{obs:hd-of-paths,obs:hd-inv-to-subedges} we know that $\hd(\Pn') = \lfloor \log(2^{k}+2) \rfloor$ and since $k \geq 2$, $\hd(\Pn') = k$.
Thus, $\GG_k$ and $\HG_k$ are homomorphism distinguishable over $\HDk$.

To show that they are indistinguishable over $\SHDk$ we use \cref{lem:homs-from-shd-are-short}.
Let $\HG \in \SHDk$.
We may assume that $\HG$ is connected.
According to \cref{obs:hom-of-conn-hg-is-conn,lem:homs-from-shd-are-short} we know that for every homomorphism $(\hV, \hE)$ from $\HG$ to $\GG_k$, either $f \not\in \img(\hE)$ or $g \not\in \img(\hE)$.
The same holds for the homomorphisms from $\HG$ to $\HG_k$.
It is easy to see that this means we can give a bijective mapping between the homomorphisms into $\GG_k$ and those into $\HG_k$.
I.e., $\hom(\HG, \GG_k) = \hom(\HG, \HG_k)$.

\medskip
\noindent\textcolor{lipicsGray}{\textbf{\textsf{2.}}}\;
We use a similar construction and \cref{lem:homs-from-hd-are-short}.
For all $k \geq 2$ let $\GG_k'$ and $\HG_k'$ be the hypergraphs with 
\begin{align*}
	\V(\GG_k') &= \V(\HG_k') = \V(\Pn_{2^{k{+}2}+1}) \disunion \set{ v,w } = \set{ u_1, u_2, \dots } \disunion \set{ v,w }\,,\\
	\E(\GG_k') &= \E(\HG_k') = \E(\Pn_{2^{k{+}2}+1}) \disunion \set{ e,f,g }\,, \quad\text{and}\\
	\f_{\GG_k} &= \set{ e \to \set{ u_{2^{k{+}2}+2}, u_1 }, f \to \set{ u_1, v }, g \to \set{ u_{2^{k+1}-2}, w } } \union \f_{\Pn_{2^{k{+}1}+1}}\,,\\
	\f_{\HG_k} &= \set{ e \to \set{ u_{2^{k{+}2}+2}, u_1 }, f \to \set{ u_1, v }, g \to \set{ u_{2^{k+1}-1} , w }} \union \f_{\Pn_{2^{k{+}1}+1}}\,.
\end{align*}
I.e., $\GG_k$ and $\HG_k$ both are circles of $2^{k{+}2}+2$ edges, with two additional hyperedges \enquote{sticking out as handles}.
Again, the only difference is in the distance between those handles: There are $2^{k{+}1}-3$ hyperedges between them in $\GG_k$ and $2^{k+1}-2$ in $\HG_k$.

With the same argument as before one can see that $\Hom(\HDk, \GG_k) = \Hom(\HDk, \GG_k)$, because this time the handles are not included in the outer edges of the path, rather they \emph{are} the outer edges of the path.
Thus, the path a homomorphism would have to map to, to distinguish $\GG_k$ and $\HG_k$ contains $2^{k{+}1}-3+2 = 2^{k{+}1}-1$ edges.

Consider $\Pn'$ with $\V(\Pn') = \V(\Pn_{2^{k+1}-3}) = \set{ v_1, v_2, \dots }$, $\E(\Pn') = \E(\Pn_{2^{k+1}-3}) \disunion \set{ f', g' }$ and $\f_{\Pn'} = \set{ f' \to v_1, g' \to v_{2^k-1} } \union \f_{\Pn_{2^{k{+}1}-3}}$.
Obviously, $\hd(\Pn') = \lfloor \log(2^{k+1}-3+2) \rfloor = k$ and $\Pn'$ is \enquote{too short} to distinguish $\GG_k$ and $\HG_k$.
But we can map the singleton hyperedges $f'$ and $g'$ of $\Pn'$ onto the handles if we consider homomorphisms from $\I_{\Pn'}$.
Thus, $\hom(\I_{\Pn'}, \GG_k) \neq \hom(\I_{\Pn'}, \HG_k)$.\qedhere \end{proof}
 
\section{Missing Definitions in Section~\ref{sec:k-labeled-incidence-graphs}}%
\label{app:k-labeled}

The following is taken almost verbatim from the full version of~\cite{Scheidt2023} to assure that we are using the same definitions as them,
so that we can also use their proofs.

\begin{definition}\label{def:precise-def:changing-labels}
	Let $\kLI = (\I, r, b, g)$ be a $k$-labeled incidence graph.
	Let $\XR \subseteq \natpos$ be finite, and let $\XB \subseteq [k]$.
	\begin{itemize}
	\item
	Removing from the red vertices all the labels in $\XR$ is achieved by the operation $\removeR{\kLI}{\XR} \isdef (\I, r', b, g')$ where $r'$ is the restriction of $r$ to $\dom(r) \setminus \XR$ and $g'$ is the restriction of $g$ to $\dom(g) \setminus \XR$.
		
	\item
	Removing from the blue vertices all the labels in $\XB$ is achieved by the operation $\removeB{\kLI}{\XB} \isdef (\I,r,b',g)$ where $b'$ is the restriction of $b$ to $\dom(b) \setminus \XB$.

	\item
	Let $\enum{ i_1, \dots, i_\ell }$ be the enumeration of $\XR$.
	For every $\tupel{v} = (v_1, \ldots, v_\ell) \in {\R(\I)}^\ell$ we let $\changeR{\kLI}{\XR}{\tupel{v}} \isdef (\I,r',b,g)$ with $\dom(r') = \dom(r) \union \XR$ and $r'(i_j) = v_j$ for all $j \in [\ell]$ and $r'(i) = r(i)$ for all $i \in \dom(r) \setminus \XR$.

	\item
	Let $\enum{ i_1, \dots, i_\ell }$ be the enumeration of $\XB$.
	For every $\tupel{e}=(e_1,\ldots,e_\ell) \in {\B(\I)}^\ell$ we let $\changeB{\kLI}{\XB}{\tupel{e}} \isdef (\I,r,b',g)$ with $\dom(b') = \dom(b) \union \XB$ and $b'(i_j) = e_j$ for all $j \in [\ell]$ and $b'(i) = b(i)$ for all $i \in \dom(b) \setminus \XB$.
	\endhere%
	\end{itemize}
\end{definition}

\begin{definition}[Glueing $k$-labeled incidence graphs]%
	\label{def:precise-def:glueop}
	Let $\kLI_i \isdef (\I_i, r_i, b_i, g_i)$ be a $k$-labeled incidence graph for $i\in[2]$.
	The \emph{glueing operation} produces the $k$-labeled incidence graph
	$(\kLI_1 \cdot \kLI_2) \isdef (\I, r, b, g)$ in the following	way.
	\smallskip

	\noindent Let $\I' \isdef \I_1 \disunion \I_2$ be	the disjoint union of $\I_1$ and $\I_2$.
	Precisely, we let
	\begin{align*}
		\R(\I') &\;\;\isdef\;\; \textstyle\bigcup_{i \in [2]} (\R(\I_i) \times \set{i}), \quad
		\B(\I') \;\;\isdef\;\; \textstyle\bigcup_{i \in [2]}(\B(\I_i) \times \set{i}) \quad\text{and}\\
		\E(\I') &\;\;\isdef\;\; \textstyle\bigcup_{i \in [2]} \set{ \bigl( (e,i), (v,i) \bigr) \mid (e,v) \in \E(\I_i)}.
	\end{align*}
	
	Let $\sim_R$ be the equivalence relation on $\R(\I')$ obtained as the reflexive, symmetric, and transitive closure of the relation $\set{ \bigl( (r_1(j),1),(r_2(j),2) \bigr) \mid j \in \dom(r_1) \intersect \dom(r_2)}$.
	Let ${[v]}_{\sim_R}$ denote the equivalence class of each $v \in \R(\I')$.

	Let $\sim_B$ be the equivalence relation on $\B(\I')$ obtained as the reflexive, symmetric, and transitive closure of the relation
	$\set{ \bigl( (b_1(j),1), (b_2(j),2) \bigr) \mid j \in \dom(b_1) \intersect \dom(b_2) }$.
	Let ${[e]}_{\sim_B}$ denote the equivalence class of each $e \in \B(\I')$.
	\smallskip

	\noindent $\I$ is the incidence graph $(\R(\I), \B(\I), \E(\I))$ for
	\begin{align*}
		\R(\I) &\isdef \set{{[v]}_{\sim_R} \mid v\in\R(\I')}, \quad
		\B(\I) \isdef \set{{[e]}_{\sim_B} \mid e \in \B(\I') } \quad\text{and}\\
		\E(\I) &\isdef \set{ \bigl( {[(e,i)]}_{\sim_B}\,,\,{[(v,i)]}_{\sim_R} \bigr) \mid i \in [2],\ (e,v) \in \E(\I_i)}.
	\end{align*}

	\noindent Next we define the functions $r, b, g$.
	\begin{description}
		\item[\textup{(r)}] 
		Let $r\colon \natpos \pto \R(\I)$ with $\dom(r) = \dom(r_1) \union \dom(r_2)$ be defined by
		$r(j) \isdef {[(r_1(j),1)]}_{\sim_R}$ for all $j \in \dom(r_1)$ and
		$r(j) \isdef {[(r_2(j),2)]}_{\sim_R}$ for all $j \in \dom(r_2) \setminus \dom(r_1)$.

		\item[\textup{(b)}] 
		Let $b\colon [k] \pto \B(\I)$ with $\dom(b) = \dom(b_1) \union \dom(b_2)$ be defined by
		$b(j) \isdef {[(b_1(j),1)]}_{\sim_B}$ for all $j \in \dom(b_1)$ and
		$b(j) \isdef {[(b_2(j),2)]}_{\sim_B}$ for all $j \in \dom(b_2) \setminus \dom(b_1)$.

		\item[\textup{(g)}] 
		Let $g \isdef g_1 \union g_2$ \ (recall from \cref{sec:preliminaries} that	this means that $g(j) = g_1(j)$ for all $j \in \dom(g_1)$ and $g(j) = g_2(j)$ for all	$j \in \dom(g_2)\setminus \dom(g_1)$).
	\end{description}
	\smallskip
	
	\noindent	Finally, for both $i \in [2]$ we define mappings $\RvertexMap{\kLI_i}\colon \R(\I_i) \to \R(\I)$ and $\BvertexMap{\kLI_i}\colon \B(\I_i) \to \B(\I)$ by
	\begin{align*}
		\RvertexMap{\kLI_i}(v) &\isdef {[(v,i)]}_{\sim_R} \quad\all\; v \in \R(\I_i)\; \text{and} \\
		\BvertexMap{\kLI_i}(e) &\isdef {[(e,i)]}_{\sim_B} \quad\all\; e \in \B(\I_i).
	\end{align*}
	Note that $\RvertexMap{\kLI_i}(v)$ is the red vertex of	$\I_{(\kLI_1 \cdot \kLI_2)}$ that corresponds to $v \in \R(\I_i)$, and $\BvertexMap{\kLI_i}(e)$ is the blue vertex of
	$\I_{(\kLI_1 \cdot \kLI_2)}$ that corresponds to $e \in \B(\I_i)$.
\end{definition}

\begin{definition}\label{def:mf}
	Let $f\colon \natpos \pto [k]$ with finite, non-empty $\dom(f)$.
	The $k$-labeled incidence graph $\Mf$ \emph{defined by $f$} is the $k$-labeled incidence graph $\kLI=(\I,r,b,g)$ with $g \isdef f$, where $\I$ consists of a red vertex $v_i$ for every $i \in \dom(f)$, 
	a blue vertex $e_j$ for every $j \in \img(f)$, and an edge $(e_{f(i)},v_i)$	for every $i \in \dom(f)$, and where $\dom(r) = \dom(f)$ and $r(i) = v_i$	for all $i \in \dom(r)$, and $\dom(b) = \img(f)$ and $b(j) = e_j$ for all	$j \in \dom(b)$.
	\end{definition}

\begin{definition}\label{def:precise-def:transition}
	Consider a partial function $g\colon \natpos \pto [k]$.
	A \emph{transition for $g$} is a partial function $f\colon \natpos \pto [k]$ with $\emptyset \neq \dom(f) \subseteq \dom(g)$ satisfying the following:
	for every $i\in \dom(g)$ with $g(i) \in \img(f)$ we have $i \in \dom(f)$.
	
	Let $\kLI = (\I, r, b, g)$ be a $k$-labeled incidence graph.
	$f$ is a transition for $\kLI$, if it is a transition for $g$.
	We can \emph{apply} the transition $f$ to $\kLI$ and obtain the $k$-labeled incidence graph $\transition{\kLI}{f} \isdef (\Mf \cdot \removeB{\kLI}{\XB})$, where  $\XB \isdef \img(g) \intersect \img(f) \intersect \dom(b)$ and $\Mf$ is defined as in \cref{def:mf}.
\end{definition} 
\section{Exemplary Construction}%
\label{app:characterising-hd}

\begin{figure}[t]
	\centering
	\begin{subfigure}{.53\textwidth}
		\centering

\begin{tikzpicture}[
	every node/.style={regular},
]
	\node[blue] (e7) at (\cebase,\cheight) {$a$};						%
	\node[blue] (e6) at (\cebase+1*\cstep,\cheight) {$b$};	%
	\node[blue] (e5) at (\cebase+2*\cstep,\cheight) {$c$};	%
	\node[blue] (e4) at (\cebase+3*\cstep,\cheight) {$d$};	%
	\node[blue] (e3) at (\cebase+4*\cstep,\cheight) {$e$};	%
	\node[blue] (e2) at (\cebase+5*\cstep,\cheight) {$f$};	%
	\node[blue] (e1) at (\cebase+6*\cstep,\cheight) {$g$};	%

	\node[red] (v1) at (\cvbase, 0) {$s$};								%
	\node[red] (v2) at (\cvbase+1*\cstep, 0) {$t$};				%
	\node[red] (v3) at (\cvbase+2*\cstep, 0) {$u$};				%
	\node[red] (v4) at (\cvbase+3*\cstep, 0) {$v$};				%
	\node[red] (v5) at (\cvbase+4*\cstep, 0) {$w$};				%
	\node[red] (v6) at (\cvbase+5*\cstep, 0) {$x$};				%
	\node[red] (v7) at (\cvbase+6*\cstep, 0) {$y$};				%
	\node[red] (v8) at (\cvbase+7*\cstep, 0) {$z$};				%

	\draw[edge] (e7) -- (v1);
	\draw[edge] (e7) -- (v2);

	\draw[edge] (e6) -- (v2);
	\draw[edge] (e6) -- (v3);

	\draw[edge] (e5) -- (v3);
	\draw[edge] (e5) -- (v4);

	\draw[edge] (e4) -- (v4);
	\draw[edge] (e4) -- (v5);

	\draw[edge] (e3) -- (v5);
	\draw[edge] (e3) -- (v6);

	\draw[edge] (e2) -- (v6);
	\draw[edge] (e2) -- (v7);

	\draw[edge] (e1) -- (v7);
	\draw[edge] (e1) -- (v8);
\end{tikzpicture} 		\caption{The path $\I_{\Pn_7}$ on 8 vertices.}\label{fig:construction:reference-path}%
	\end{subfigure}
	\hfill
	\begin{subfigure}{.42\textwidth}
		\centering

\begin{tikzpicture}[
	every node/.style={draw, rectangle, font={\footnotesize}},
	level distance=8mm,level/.style={sibling distance=30mm/#1}
	]
	\node[label=left:{$\troot$:}, level distance=2cm] { $t_1{:} \set{ v,w }$ }
	child { node {$t_2{:} \set{ t,u }$} edge from parent [Stealth-] 
		child { node {$t_4{:} \set{ s,t }$} }
		child { node {$t_5{:} \set{ u,v }$} }
	}
	child { node {$t_3{:} \set{ x,y }$} edge from parent [Stealth-] 
		child { node {$t_6{:} \set{ w,x }$} }
		child { node {$t_7{:} \set{ y,z }$} }
	};
\end{tikzpicture} 		\caption{A (strict) elimination tree for $\I_{\Pn_7}$.}%
		\label{fig:construction:reference-path-tree}%
	\end{subfigure}
	\caption{}\label{fig:construction}
\end{figure}

\begin{example}\label{app:example-construction}
	Consider the path $\Pn_7$ with $\V(\Pn_7) = \set{ s,t,u,v,w,x,y,z }$ and $\E(\Pn_7) = \set{ \set{ s,t }, \set{ t,u }, \set{ u,v }, \set{ v,w }, \set{ w,x }, \set{ x,y }, \set{ y,z } }$ and its incidence graph $\I_{\Pn_7}$ as depicted in \cref{fig:construction:reference-path}.
	We show how the method behind the proof of \cref{lem:bounded-shd-in-glik} applied on $\I_{\Pn_7}$ and the elimination tree depicted in \cref{fig:construction:reference-path-tree} constructs a label-free $\kLI \in \GLI_3^3$ whose skeleton is isomorphic to $\I_{\Pn_7}$.

	We construct $\kLI$ bottom up along the elimination tree.
	In the following, we encode the label of a vertex as its exponent and the guard function as thicker edges between the red vertex and its guard.
	For every leaf $t_4, t_5, t_6, t_7$, we take the following $k$-labeled incidence graphs $\kLI_{t_i} \in \GLI_3^0$.
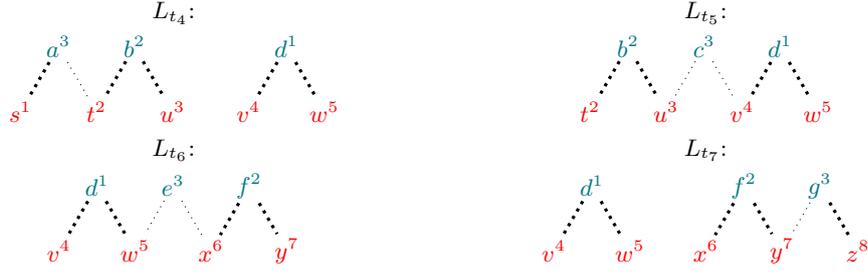
	\begin{center}
		\begin{tabularx}{\columnwidth}{C C}
			$\kLI_{t_4}$:
			&
			$\kLI_{t_5}$:
			\\
			\adjustbox{valign=c}{

\begin{tikzpicture}[
	every node/.style={regular},
]
	\node[blue] (e7) at (\cebase,\cheight) {$a^3$};						%
	\node[blue] (e6) at (\cebase+1*\cstep,\cheight) {$b^2$};	%
	\node[blue] (e4) at (\cebase+3*\cstep,\cheight) {$d^1$};	%

	\node[red] (v1) at (\cvbase, 0) {$s^1$};								%
	\node[red] (v2) at (\cvbase+1*\cstep, 0) {$t^2$};				%
	\node[red] (v3) at (\cvbase+2*\cstep, 0) {$u^3$};				%
	\node[red] (v4) at (\cvbase+3*\cstep, 0) {$v^4$};				%
	\node[red] (v5) at (\cvbase+4*\cstep, 0) {$w^5$};				%

	\draw[guard] (e7) -- (v1);
	\draw[edge] (e7) -- (v2);

	\draw[guard] (e6) -- (v2);
	\draw[guard] (e6) -- (v3);

	\draw[guard] (e4) -- (v4);
	\draw[guard] (e4) -- (v5);
\end{tikzpicture} 			}
			&
			\adjustbox{valign=c}{

\begin{tikzpicture}[
	every node/.style={regular},
]
	\node[blue] (e6) at (\cebase+1*\cstep,\cheight) {$b^2$};	%
	\node[blue] (e5) at (\cebase+2*\cstep,\cheight) {$c^3$};	%
	\node[blue] (e4) at (\cebase+3*\cstep,\cheight) {$d^1$};	%

	\node[red] (v2) at (\cvbase+1*\cstep, 0) {$t^2$};				%
	\node[red] (v3) at (\cvbase+2*\cstep, 0) {$u^3$};				%
	\node[red] (v4) at (\cvbase+3*\cstep, 0) {$v^4$};				%
	\node[red] (v5) at (\cvbase+4*\cstep, 0) {$w^5$};				%

	\draw[guard] (e6) -- (v2);
	\draw[guard] (e6) -- (v3);

	\draw[edge] (e5) -- (v3);
	\draw[edge] (e5) -- (v4);

	\draw[guard] (e4) -- (v4);
	\draw[guard] (e4) -- (v5);
\end{tikzpicture} 			}\\
			$\kLI_{t_6}$:
			&
			$\kLI_{t_7}$:\\
			\adjustbox{valign=c}{

\begin{tikzpicture}[
	every node/.style={regular},
]
	\node[blue] (e4) at (\cebase+3*\cstep,\cheight) {$d^1$};	%
	\node[blue] (e3) at (\cebase+4*\cstep,\cheight) {$e^3$};	%
	\node[blue] (e2) at (\cebase+5*\cstep,\cheight) {$f^2$};	%

	\node[red] (v4) at (\cvbase+3*\cstep, 0) {$v^4$};				%
	\node[red] (v5) at (\cvbase+4*\cstep, 0) {$w^5$};				%
	\node[red] (v6) at (\cvbase+5*\cstep, 0) {$x^6$};				%
	\node[red] (v7) at (\cvbase+6*\cstep, 0) {$y^7$};				%

	\draw[guard] (e4) -- (v4);
	\draw[guard] (e4) -- (v5);

	\draw[edge] (e3) -- (v5);
	\draw[edge] (e3) -- (v6);

	\draw[guard] (e2) -- (v6);
	\draw[guard] (e2) -- (v7);
\end{tikzpicture} 			}
			&
			\adjustbox{valign=c}{

\begin{tikzpicture}[
	every node/.style={regular},
]
	\node[blue] (e4) at (\cebase+3*\cstep,\cheight) {$d^1$};	%
	\node[blue] (e2) at (\cebase+5*\cstep,\cheight) {$f^2$};	%
	\node[blue] (e1) at (\cebase+6*\cstep,\cheight) {$g^3$};	%

	\node[red] (v4) at (\cvbase+3*\cstep, 0) {$v^4$};				%
	\node[red] (v5) at (\cvbase+4*\cstep, 0) {$w^5$};				%
	\node[red] (v6) at (\cvbase+5*\cstep, 0) {$x^6$};				%
	\node[red] (v7) at (\cvbase+6*\cstep, 0) {$y^7$};				%
	\node[red] (v8) at (\cvbase+7*\cstep, 0) {$z^8$};				%

	\draw[guard] (e4) -- (v4);
	\draw[guard] (e4) -- (v5);

	\draw[guard] (e2) -- (v6);
	\draw[guard] (e2) -- (v7);

	\draw[edge] (e1) -- (v7);
	\draw[guard] (e1) -- (v8);
\end{tikzpicture} 			}
		\end{tabularx}
	\end{center}
	$\kLI_{t_2}$ is then constructed as $(\kLI_{t_4}' \cdot \kLI_{t_5}')$, where
	\begin{center}
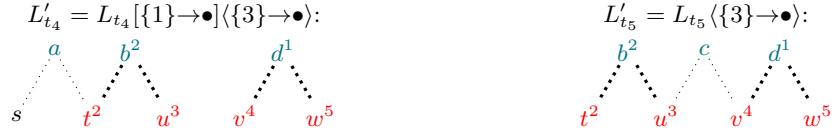
		\begin{tabularx}{\columnwidth}{C C}
			$\kLI_{t_4}' = \removeB{\removeR{\kLI_{t_4}}{\set{ 1 }}}{\set{ 3 }}$:
			&
			$\kLI_{t_5}' = \removeB{\kLI_{t_5}}{\set{ 3 }}$:
			\\
			\adjustbox{valign=c}{

\begin{tikzpicture}[
	every node/.style={regular},
]
	\node[blue] (e7) at (\cebase,\cheight) {$a$};						%
	\node[blue] (e6) at (\cebase+1*\cstep,\cheight) {$b^2$};	%
	\node[blue] (e4) at (\cebase+3*\cstep,\cheight) {$d^1$};	%

	\node (v1) at (\cvbase, 0) {$s$};								%
	\node[red] (v2) at (\cvbase+1*\cstep, 0) {$t^2$};				%
	\node[red] (v3) at (\cvbase+2*\cstep, 0) {$u^3$};				%
	\node[red] (v4) at (\cvbase+3*\cstep, 0) {$v^4$};				%
	\node[red] (v5) at (\cvbase+4*\cstep, 0) {$w^5$};				%

	\draw[edge] (e7) -- (v1);
	\draw[edge] (e7) -- (v2);

	\draw[guard] (e6) -- (v2);
	\draw[guard] (e6) -- (v3);

	\draw[guard] (e4) -- (v4);
	\draw[guard] (e4) -- (v5);
\end{tikzpicture} 			}
			&
			\adjustbox{valign=c}{

\begin{tikzpicture}[
	every node/.style={regular},
]
	\node[blue] (e6) at (\cebase+1*\cstep,\cheight) {$b^2$};	%
	\node[blue] (e5) at (\cebase+2*\cstep,\cheight) {$c$};	%
	\node[blue] (e4) at (\cebase+3*\cstep,\cheight) {$d^1$};	%

	\node[red] (v2) at (\cvbase+1*\cstep, 0) {$t^2$};				%
	\node[red] (v3) at (\cvbase+2*\cstep, 0) {$u^3$};				%
	\node[red] (v4) at (\cvbase+3*\cstep, 0) {$v^4$};				%
	\node[red] (v5) at (\cvbase+4*\cstep, 0) {$w^5$};				%

	\draw[guard] (e6) -- (v2);
	\draw[guard] (e6) -- (v3);

	\draw[edge] (e5) -- (v3);
	\draw[edge] (e5) -- (v4);

	\draw[guard] (e4) -- (v4);
	\draw[guard] (e4) -- (v5);
\end{tikzpicture} 			}
		\end{tabularx}
	\end{center}
	I.e., $\kLI_{t_2} \in \GLI_3^1$ is the following $k$-labeled incidence graph:
	
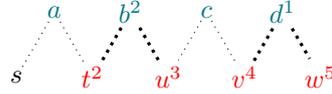
\begin{center}

\begin{tikzpicture}[
	every node/.style={regular},
]
	\node[blue] (e7) at (\cebase,\cheight) {$a$};						%
	\node[blue] (e6) at (\cebase+1*\cstep,\cheight) {$b^2$};	%
	\node[blue] (e5) at (\cebase+2*\cstep,\cheight) {$c$};	%
	\node[blue] (e4) at (\cebase+3*\cstep,\cheight) {$d^1$};	%

	\node (v1) at (\cvbase, 0) {$s$};								%
	\node[red] (v2) at (\cvbase+1*\cstep, 0) {$t^2$};				%
	\node[red] (v3) at (\cvbase+2*\cstep, 0) {$u^3$};				%
	\node[red] (v4) at (\cvbase+3*\cstep, 0) {$v^4$};				%
	\node[red] (v5) at (\cvbase+4*\cstep, 0) {$w^5$};				%

	\draw[edge] (e7) -- (v1);
	\draw[edge] (e7) -- (v2);

	\draw[guard] (e6) -- (v2);
	\draw[guard] (e6) -- (v3);

	\draw[edge] (e5) -- (v3);
	\draw[edge] (e5) -- (v4);

	\draw[guard] (e4) -- (v4);
	\draw[guard] (e4) -- (v5);
\end{tikzpicture} 	\end{center}
	Analogously, $\kLI_{t_3}$ is constructed as $(\kLI_{t_6}' \cdot \kLI_{t_7}')$, where
	\begin{center}
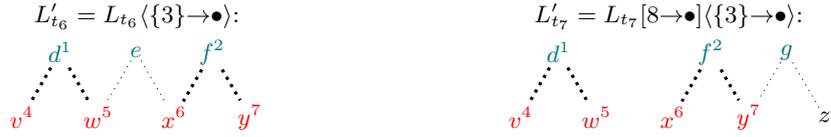
		\begin{tabularx}{\columnwidth}{C C}
			$\kLI_{t_6}' = \removeB{\kLI_{t_6}}{\set{ 3 }}$:
			&
			$\kLI_{t_7}' = \removeB{\removeR{\kLI_{t_7}}{8}}{\set{ 3 }}$:
			\\
			\adjustbox{valign=c}{

\begin{tikzpicture}[
	every node/.style={regular},
]
	\node[blue] (e4) at (\cebase+3*\cstep,\cheight) {$d^1$};	%
	\node[blue] (e3) at (\cebase+4*\cstep,\cheight) {$e$};	%
	\node[blue] (e2) at (\cebase+5*\cstep,\cheight) {$f^2$};	%

	\node[red] (v4) at (\cvbase+3*\cstep, 0) {$v^4$};				%
	\node[red] (v5) at (\cvbase+4*\cstep, 0) {$w^5$};				%
	\node[red] (v6) at (\cvbase+5*\cstep, 0) {$x^6$};				%
	\node[red] (v7) at (\cvbase+6*\cstep, 0) {$y^7$};				%

	\draw[guard] (e4) -- (v4);
	\draw[guard] (e4) -- (v5);

	\draw[edge] (e3) -- (v5);
	\draw[edge] (e3) -- (v6);

	\draw[guard] (e2) -- (v6);
	\draw[guard] (e2) -- (v7);
\end{tikzpicture} 			}
			&
			\adjustbox{valign=c}{

\begin{tikzpicture}[
	every node/.style={regular},
]
	\node[blue] (e4) at (\cebase+3*\cstep,\cheight) {$d^1$};	%
	\node[blue] (e2) at (\cebase+5*\cstep,\cheight) {$f^2$};	%
	\node[blue] (e1) at (\cebase+6*\cstep,\cheight) {$g$};	%

	\node[red] (v4) at (\cvbase+3*\cstep, 0) {$v^4$};				%
	\node[red] (v5) at (\cvbase+4*\cstep, 0) {$w^5$};				%
	\node[red] (v6) at (\cvbase+5*\cstep, 0) {$x^6$};				%
	\node[red] (v7) at (\cvbase+6*\cstep, 0) {$y^7$};				%
	\node (v8) at (\cvbase+7*\cstep, 0) {$z$};				%

	\draw[guard] (e4) -- (v4);
	\draw[guard] (e4) -- (v5);

	\draw[guard] (e2) -- (v6);
	\draw[guard] (e2) -- (v7);

	\draw[edge] (e1) -- (v7);
	\draw[edge] (e1) -- (v8);
\end{tikzpicture} 			}
		\end{tabularx}
	\end{center}
	I.e., $\kLI_{t_3} \in \GLI_3^1$ is the following $k$-labeled incidence graph:
	
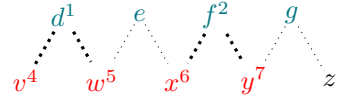
\begin{center}

\begin{tikzpicture}[
	every node/.style={regular},
]
	\node[blue] (e4) at (\cebase+3*\cstep,\cheight) {$d^1$};	%
	\node[blue] (e3) at (\cebase+4*\cstep,\cheight) {$e$};	%
	\node[blue] (e2) at (\cebase+5*\cstep,\cheight) {$f^2$};	%
	\node[blue] (e1) at (\cebase+6*\cstep,\cheight) {$g$};	%

	\node[red] (v4) at (\cvbase+3*\cstep, 0) {$v^4$};				%
	\node[red] (v5) at (\cvbase+4*\cstep, 0) {$w^5$};				%
	\node[red] (v6) at (\cvbase+5*\cstep, 0) {$x^6$};				%
	\node[red] (v7) at (\cvbase+6*\cstep, 0) {$y^7$};				%
	\node (v8) at (\cvbase+7*\cstep, 0) {$z$};				%

	\draw[guard] (e4) -- (v4);
	\draw[guard] (e4) -- (v5);

	\draw[edge] (e3) -- (v5);
	\draw[edge] (e3) -- (v6);

	\draw[guard] (e2) -- (v6);
	\draw[guard] (e2) -- (v7);

	\draw[edge] (e1) -- (v7);
	\draw[edge] (e1) -- (v8);
\end{tikzpicture} 	\end{center}
	Now, $\kLI_{t_4}$ is $(\kLI_{t_2}' \cdot \kLI_{t_3}')$ where
	\begin{center}
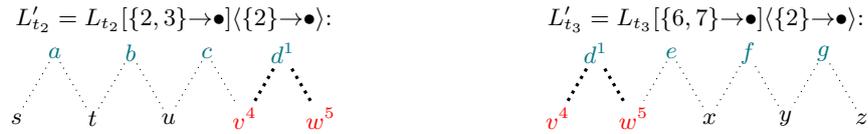
		\begin{tabularx}{\columnwidth}{C C}
			$\kLI_{t_2}' = \removeB{\removeR{\kLI_{t_2}}{\set{ 2,3 }}}{\set{ 2 }}$:
			&
			$\kLI_{t_3}' = \removeB{\removeR{\kLI_{t_3}}{\set{ 6,7 }}}{\set{ 2 }}$:
			\\
			\adjustbox{valign=c}{

\begin{tikzpicture}[
	every node/.style={regular},
]
	\node[blue] (e7) at (\cebase,\cheight) {$a$};						%
	\node[blue] (e6) at (\cebase+1*\cstep,\cheight) {$b$};	%
	\node[blue] (e5) at (\cebase+2*\cstep,\cheight) {$c$};	%
	\node[blue] (e4) at (\cebase+3*\cstep,\cheight) {$d^1$};	%

	\node (v1) at (\cvbase, 0) {$s$};								%
	\node (v2) at (\cvbase+1*\cstep, 0) {$t$};				%
	\node (v3) at (\cvbase+2*\cstep, 0) {$u$};				%
	\node[red] (v4) at (\cvbase+3*\cstep, 0) {$v^4$};				%
	\node[red] (v5) at (\cvbase+4*\cstep, 0) {$w^5$};				%

	\draw[edge] (e7) -- (v1);
	\draw[edge] (e7) -- (v2);

	\draw[edge] (e6) -- (v2);
	\draw[edge] (e6) -- (v3);

	\draw[edge] (e5) -- (v3);
	\draw[edge] (e5) -- (v4);

	\draw[guard] (e4) -- (v4);
	\draw[guard] (e4) -- (v5);
\end{tikzpicture} 			}
			&
			\adjustbox{valign=c}{

\begin{tikzpicture}[
	every node/.style={regular},
]
	\node[blue] (e4) at (\cebase+3*\cstep,\cheight) {$d^1$};	%
	\node[blue] (e3) at (\cebase+4*\cstep,\cheight) {$e$};	%
	\node[blue] (e2) at (\cebase+5*\cstep,\cheight) {$f$};	%
	\node[blue] (e1) at (\cebase+6*\cstep,\cheight) {$g$};	%

	\node[red] (v4) at (\cvbase+3*\cstep, 0) {$v^4$};				%
	\node[red] (v5) at (\cvbase+4*\cstep, 0) {$w^5$};				%
	\node (v6) at (\cvbase+5*\cstep, 0) {$x$};				%
	\node (v7) at (\cvbase+6*\cstep, 0) {$y$};				%
	\node (v8) at (\cvbase+7*\cstep, 0) {$z$};				%

	\draw[guard] (e4) -- (v4);
	\draw[guard] (e4) -- (v5);

	\draw[edge] (e3) -- (v5);
	\draw[edge] (e3) -- (v6);

	\draw[edge] (e2) -- (v6);
	\draw[edge] (e2) -- (v7);

	\draw[edge] (e1) -- (v7);
	\draw[edge] (e1) -- (v8);
\end{tikzpicture} 			}
		\end{tabularx}
	\end{center}
	I.e., $\kLI_{t_4} \in \GLI_3^2$ is the following $k$-labeled incidence graph:
	
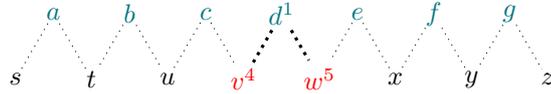
\begin{center}

\begin{tikzpicture}[
	every node/.style={regular},
]
	\node[blue] (e7) at (\cebase,\cheight) {$a$};						%
	\node[blue] (e6) at (\cebase+1*\cstep,\cheight) {$b$};	%
	\node[blue] (e5) at (\cebase+2*\cstep,\cheight) {$c$};	%
	\node[blue] (e4) at (\cebase+3*\cstep,\cheight) {$d^1$};	%
	\node[blue] (e3) at (\cebase+4*\cstep,\cheight) {$e$};	%
	\node[blue] (e2) at (\cebase+5*\cstep,\cheight) {$f$};	%
	\node[blue] (e1) at (\cebase+6*\cstep,\cheight) {$g$};	%

	\node (v1) at (\cvbase, 0) {$s$};								%
	\node (v2) at (\cvbase+1*\cstep, 0) {$t$};				%
	\node (v3) at (\cvbase+2*\cstep, 0) {$u$};				%
	\node[red] (v4) at (\cvbase+3*\cstep, 0) {$v^4$};				%
	\node[red] (v5) at (\cvbase+4*\cstep, 0) {$w^5$};				%
	\node (v6) at (\cvbase+5*\cstep, 0) {$x$};				%
	\node (v7) at (\cvbase+6*\cstep, 0) {$y$};				%
	\node (v8) at (\cvbase+7*\cstep, 0) {$z$};				%

	\draw[edge] (e7) -- (v1);
	\draw[edge] (e7) -- (v2);

	\draw[edge] (e6) -- (v2);
	\draw[edge] (e6) -- (v3);

	\draw[edge] (e5) -- (v3);
	\draw[edge] (e5) -- (v4);

	\draw[guard] (e4) -- (v4);
	\draw[guard] (e4) -- (v5);

	\draw[edge] (e3) -- (v5);
	\draw[edge] (e3) -- (v6);

	\draw[edge] (e2) -- (v6);
	\draw[edge] (e2) -- (v7);

	\draw[edge] (e1) -- (v7);
	\draw[edge] (e1) -- (v8);
\end{tikzpicture} 	\end{center}
	Simply removing the remaining labels gives us $\kLI \in \GLI_3^3$, i.e., $\kLI = \kLI_{t_4}' = \removeB{\removeR{\kLI_{t_4}}{\set{ 4,5 }}}{1}$:
	\begin{center}

\begin{tikzpicture}[
	every node/.style={regular},
]
	\node[blue] (e7) at (\cebase,\cheight) {$a$};						%
	\node[blue] (e6) at (\cebase+1*\cstep,\cheight) {$b$};	%
	\node[blue] (e5) at (\cebase+2*\cstep,\cheight) {$c$};	%
	\node[blue] (e4) at (\cebase+3*\cstep,\cheight) {$d$};	%
	\node[blue] (e3) at (\cebase+4*\cstep,\cheight) {$e$};	%
	\node[blue] (e2) at (\cebase+5*\cstep,\cheight) {$f$};	%
	\node[blue] (e1) at (\cebase+6*\cstep,\cheight) {$g$};	%

	\node (v1) at (\cvbase, 0) {$s$};								%
	\node (v2) at (\cvbase+1*\cstep, 0) {$t$};				%
	\node (v3) at (\cvbase+2*\cstep, 0) {$u$};				%
	\node (v4) at (\cvbase+3*\cstep, 0) {$v$};				%
	\node (v5) at (\cvbase+4*\cstep, 0) {$w$};				%
	\node (v6) at (\cvbase+5*\cstep, 0) {$x$};				%
	\node (v7) at (\cvbase+6*\cstep, 0) {$y$};				%
	\node (v8) at (\cvbase+7*\cstep, 0) {$z$};				%

	\draw[edge] (e7) -- (v1);
	\draw[edge] (e7) -- (v2);

	\draw[edge] (e6) -- (v2);
	\draw[edge] (e6) -- (v3);

	\draw[edge] (e5) -- (v3);
	\draw[edge] (e5) -- (v4);

	\draw[edge] (e4) -- (v4);
	\draw[edge] (e4) -- (v5);

	\draw[edge] (e3) -- (v5);
	\draw[edge] (e3) -- (v6);

	\draw[edge] (e2) -- (v6);
	\draw[edge] (e2) -- (v7);

	\draw[edge] (e1) -- (v7);
	\draw[edge] (e1) -- (v8);
\end{tikzpicture} 	\end{center}
	Obviously, $\I_\kLI \isomorphic \I_{\Pn_7}$.
\end{example}  
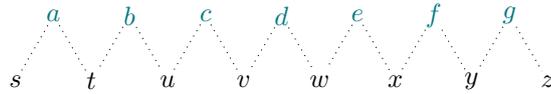
\end{document}